\documentclass[12pt]{article}
\usepackage{smile_1}
\usepackage{mathtools}
\usepackage{booktabs}
\usepackage[english]{babel}
\usepackage[protrusion=true,expansion=true]{microtype}
\usepackage{amsmath,amsfonts,amsthm}
\usepackage{ amssymb }
\usepackage{enumitem}
\usepackage{graphicx}
\usepackage{array}

\usepackage[toc,page]{appendix}
\usepackage{booktabs}
\usepackage{natbib}
\usepackage{microtype}
\usepackage[font = small,labelfont=bf,textfont=it]{caption}
\usepackage{footnote}
\usepackage{multirow}
\usepackage{algorithm}
\usepackage{mathrsfs}
\usepackage[noend]{algorithmic}

\usepackage{caption}
\usepackage[CJKbookmarks=true,
            bookmarksnumbered=true,
			bookmarksopen=true,
			colorlinks=true,
			citecolor=blue,
			linkcolor=blue,
			anchorcolor=red,
			urlcolor=blue]{hyperref}
\usepackage{rotating}
\usepackage{fancyvrb}
\usepackage[dvipsnames]{xcolor}

 \makeatletter
\def\mathcolor#1#{\@mathcolor{#1}}
\def\@mathcolor#1#2#3{\protect\leavevmode
  \begingroup
    \color#1{#2}#3\endgroup
}
\makeatother

\usepackage[margin=1in,hmarginratio=1:1,top=20mm,columnsep=20pt]{geometry}
\usepackage{microtype}
\usepackage[font = small,labelfont=bf,textfont=it]{caption}
\usepackage{footnote}
\usepackage{multirow}

\usepackage{color}

\title{Rate-Preserving Shrinking-Support Gaussian Process Prediction}

\author{Xiaopeng Xiang$^{a}$, Wenlin Dai$^b$, Marc G. Genton$^c$, Wenjia Wang$^d$\\
                 $^a$ The Hong Kong University of Science and Technology (Guangzhou)\\
    $^b$ Renmin University of China\\
    $^c$ King Abdullah University of Science and Technology\\
                $^d$ National University of Singapore\\
}
\date{}

\begin{document}
\maketitle

\begin{abstract}
Gaussian process prediction is a central tool in spatial statistics, but standard implementations require dense matrix operations that become prohibitive for large datasets. We propose a scale-adjusted compactly supported working correlation for Gaussian process prediction, using the generalized Wendland family with a support radius $\phi_n$ that is allowed to decrease with the sample size $n$. Under fixed-domain asymptotics with quasi-uniform designs, a fixed support radius does not yield asymptotic sparsity, whereas a shrinking support radius can make the covariance matrix sparse. We show that $\phi_n$ and the regularization parameter can be jointly chosen so that the resulting predictor preserves the optimal integrated mean squared prediction error rate while reducing the number of nonzero covariance entries and the associated sparse matrix-vector multiplication cost. We also establish an analogous rate-preserving sparsification result for kernel ridge regression. Simulations and an ERA5 temperature application show that the proposed method achieves competitive prediction accuracy while retaining the computational efficiency provided by sparse linear algebra.
\end{abstract}
\vfill
\newpage
		\section{Introduction}
		
		Gaussian processes (GPs) are a powerful and flexible framework for probabilistic modeling, offering a principled approach to regression, classification, and uncertainty quantification. Rooted in Bayesian nonparametrics, GPs define distributions over functions rather than fixed parametric forms, allowing them to adapt to complex data patterns while providing well-calibrated uncertainty estimates. These features have made GPs widely used in spatial statistics \citep{cressie2015statistics,matheron1963principles,stein1999interpolation}, machine learning \citep{rasmussen2006gaussian}, and computer experiments \citep{sacks1989design,santner2013design}.
		
		Practically, GPs are defined by a mean function (often zero for simplicity) and a covariance function that captures domain-specific correlations (e.g., squared-exponential functions). Inference in this framework leverages Bayesian updating to refine the posterior distribution over functions as new data are observed, seamlessly balancing flexibility with regularization. The conditional mean can be applied for predicting the function value, while the conditional variance quantifies uncertainty through pointwise confidence intervals. 
		
		Despite their strengths, GPs also face practical challenges, particularly in scalability for large datasets, where dense matrix operations become prohibitive. For \(n\) observations, exact GP inference requires solving an $n\times n$ linear system, which typically scales as $O(n^3)$ in time for dense factorization and $O(n^2)$ in memory for $n$ observed input-output pairs, which becomes expansive for large spatial datasets. Much effort has been devoted to scalable GP methods; see \citet{heaton2019case}, \citet{katzfuss2021general}, and \citet{sun2012chapter} for comprehensive reviews.
		
		Methods for likelihood approximation via sparse precision matrices include the Vecchia approximation \citep{katzfuss2020vecchia,katzfuss2021general,stein2004approximating,vecchia1988estimation}, which factorizes the joint density into a product of conditional densities with small conditioning sets, and the nearest-neighbor Gaussian process \citep{datta2016hierarchical,finley2019efficient}, which enforces sparsity through local neighborhood structures. Covariance tapering multiplies a parent covariance function by a compactly supported taper so that the tapered covariance matrix is sparse and can be factorized efficiently \citep{furrer2006covariance,kaufman2008covariance,stein2013statistical}. The stochastic partial differential equation approach and Gaussian Markov random field approximations of \citet{lindgren2011explicit} and \citet{rue2005gaussian} exploit sparse precision matrices for spatial data on triangulated meshes.
		
		Low-rank approximation is another important approach for massive spatial data. It approximates the original GP by a process represented through a smaller number of basis functions or inducing variables; notable examples include fixed-rank kriging \citep{cressie2008fixed}, LatticeKrig \citep{nychka2015multiresolution}, predictive processes \citep{banerjee2008predictive}, and the Nystr\"om approximation \citep{williams2001using}. \citet{song2025large} recently proposed support points as optimal knot locations for large-scale low-rank GP prediction. Hybrid approaches include the full-scale approximation of \citet{sang2012full}, which decomposes the process into a predictive-process component and a tapered residual process, and the multi-resolution approximation of \citet{katzfuss2017multi}, which recursively applies low-rank approximations across spatial scales. \citet{chen2022kernelpacket} introduced kernel packets for exact and scalable Mat\'ern GP regression under grid designs. Distributed computing \citep{abdulah2018exageostat} and meta-kriging \citep{guhaniyogi2018meta} provide another route to large-scale spatial prediction.
		
		This paper studies a different but closely related question: can a compactly supported covariance be made increasingly local as $n$ grows, thereby producing asymptotic sparsity, while preserving the optimal prediction rate? We provide an affirmative answer to this question in the present work by establishing a \emph{rate-preserving shrinking-support compactly supported GP prediction}. Specifically, we directly impose a compactly supported correlation function $\Phi((\bx-\bx')/\phi_n)$ for GP prediction, where $\bx,\bx'\in \RR^d$ and $\phi_n$ is the scale parameter. One prominent class of such functions is the generalized Wendland family \citep{bevilacqua2019estimation,gneitingstationary,wendland2004scattered}. If $\|\bx-\bx'\|_2\geq \phi_n$, where $\|\cdot\|_2$ is the Euclidean norm, then the imposed correlation is zero, and the corresponding covariance matrix entry is exactly zero.
		


		Our contribution is twofold. First, we propose a scale-adjusted compactly supported correlation for GP prediction with massive data. The key idea is to let the support radius $\phi_n$ decrease with the sample size $n$. Under fixed-domain asymptotics, a fixed support radius generally leads to $O(n^2)$ nonzero entries in the correlation matrix, whereas a shrinking support radius can produce asymptotic sparsity. For quasi-uniform designs, the number of nonzero entries is of order $O(n^2\phi_n^d+n)$, so choosing $\phi_n=o(1)$ reduces storage and sparse matrix-vector multiplication cost. This gives a simple kernel-level sparsification strategy that does not require inducing points, low-rank basis construction, or special grid structure.
		
		Second, we develop convergence theory for GP prediction when the imposed correlation function changes with $n$. This differs from much of the existing GP convergence literature, where the correlation parameters are fixed or behave similarly as fixed. Numerous studies have focused on Gaussian process interpolation without observational noise. For example, \cite{stein1990uniform}, \cite{yakowitz1985comparison} studied pointwise convergence rates or settings where input points are not general scattered data points, while \cite{tuo2020kriging}, \cite{wang2019prediction} established convergence rates for Gaussian process interpolation in the $L_p$ norm, with $1 \leq p \leq \infty$. Few works address the noisy regression setting. For instance, \cite{lederer2019uniform} provided a uniform error bound for Gaussian process regression, and \cite{wang2022gaussian} studied convergence under the $L_2$ norm. The shrinking-scale setting considered in the present work requires a new analysis because the corresponding norm-equivalence constants, which were used to control the posterior variance in previous works, depend on $\phi_n$, which also changes with $n$. We show that this additional scale dependence can be controlled. Specifically, for GP prediction, if the true process has smoothness $m_0>d/2$ and the imposed generalized Wendland correlation has smoothness $m\geq m_0$, then the choice \(\phi_n\asymp n^{-\frac{m-m_0}{2mm_0}}\) preserves the optimal integrated mean squared prediction error (MSPE) rate $O\left(n^{-\frac{2m_0-d}{2m_0}}\right)$
        while reducing the number of nonzero entries in the imposed correlation matrix to $O\left(n^{2-\frac{(m-m_0)d}{2mm_0}}\right).$ Thus, the support radius can shrink with $n$ to improve sparse-linear-algebra efficiency without sacrificing the optimal prediction rate.

		In particular, the distinction between our method and covariance tapering is important. 
		Covariance tapering starts with a parent covariance, such as a Mat\'ern covariance, and constructs a sparse approximation by multiplying it by a compactly supported taper. Its theoretical target is therefore tied to approximation or inference under the parent covariance model. In contrast, this paper studies the prediction risk of a deliberately misspecified, compactly supported covariance whose support radius changes with $n$. The object of our analysis is not entrywise approximation to a dense covariance matrix, but the integrated MSPE of the resulting predictor. This distinction is essential because the reproducing kernel Hilbert space (RKHS) and Fourier norm-equivalence constants depend on $\phi_n$, and controlling this dependence yields the admissible shrinking-support rate.
		
		Compared to GP modeling, the problem of deterministic function estimation with kernel ridge regression (KRR) has a larger body of literature. An incomplete list of results includes the convergence rates for KRR \citep{Hamm2021,lin2017distributed,steinwart2009optimal}, posterior contraction of GP priors in Bayesian statistics \citep{castillo2014bayesian,pati2015optimal,van2008rates}, and scattered data approximation \citep{rieger2009deterministic,wynne2021convergence}, among others. In particular, \cite{Hamm2021} also employed a changing-scale parameter in the kernel function, considered a set of upper box-counting dimensions and derived the convergence rate of the excess expected risk. However, they considered Gaussian kernel functions, which are not compactly supported and therefore do not induce exact matrix sparsity. We also obtain an analogous result under the KRR setting. Specifically, with a deterministic truth in RKHS and a suitable choice of $\phi_n$ and regularization parameter $\lambda_n$, the estimator attains the minimax-optimal $L_2$ rate $O_p(n^{-m_0/(2m_0+d)})$ while using a sparse kernel matrix. Though in this work we focus on GP prediction, the result for KRR may be of its own interest and provides new insights into the field of deterministic function estimation.

		The rest of this paper is arranged as follows. Section~\ref{sec_GP} introduces the scale-adjusted compactly supported GP predictor. Section~\ref{sec_GPtheory} gives the GP convergence and sparsity results. Section~\ref{sec_KRR} extends the analysis to KRR. Section~\ref{sec_numerical} presents simulation studies, Section~\ref{sec_real} provides an ERA5 temperature application, and Section~\ref{sec_conclu} concludes. The supplementary material is organized as follows. Supplementary Section~A reviews the RKHS background used in the analysis; Sections~B--E give the proofs of the main results and auxiliary lemmas; Section~F provides the full experimental settings; and Section~G reports additional numerical tables.
		
		\section{Gaussian Process Prediction with Scale-Adjusted Compactly Supported Correlation Functions}\label{sec_GP}
		
		Let $(\bx_k,y_k)$, $k=1,\ldots,n$, be observed data satisfying
		\begin{align}\label{recoveringGP}
			y_k = f(\bx_k) + \varepsilon_k, k=1,\ldots,n, 
		\end{align}
		where $\bx_k \in \Omega\subset \RR^d$ and $\varepsilon_k$'s are independent and identically distributed (i.i.d.) random errors with mean zero and variance $\sigma_\varepsilon^2$. In spatial statistics, the underlying function $f$ is often modeled as a realization of a Gaussian process $Z$. From this perspective, we shall not differentiate $f$ and $Z$ in Gaussian process regression setting, while the deterministic-function setting will be considered separately in Section~\ref{sec_KRR}. We assume that $Z$ is a mean-zero stationary Gaussian process, that is, the covariance of $Z(\bx)$ and $Z(\bx')$ only depends on the difference of $\bx$ and $\bx'$. Denote the covariance function of $Z$ by $\text{Cov}(Z(\bx),Z(\bx'))=\sigma^2 \Psi(\bx - \bx')$ for $\bx,\bx'\in \RR^d$, where $\sigma^2$ is the variance, and $\Psi$ is the true, but typically unknown,  correlation function which is positive definite and integrable on $\RR^d$ with $\Psi(\mathbf{0})=1$.
		
		If the random errors $\varepsilon_k$'s are normally distributed, then conditional on the observed data, $Z(\bx)$ is also normally distributed with conditional mean and variance 
		\begin{align}
			\mathbb{E}[Z(\bx)|\by]= &\rb(\bx)^\top (\Rb+\rho \Ib_n)^{-1} \by,\label{mean}\\
			\text{Var}[Z(\bx)|\by]= & \sigma^2(\Psi(\bx-\bx)-\rb(\bx)^\top (\Rb+\rho \Ib_n)^{-1}\rb(\bx)),\label{var}
		\end{align}
		where $\rb(\bx)=(\Psi(\bx-\bx_1),\ldots,\Psi(\bx-\bx_n))^\top, \Rb=(\Psi(\bx_j-\bx_k))_{j k}$, $\Ib_n$ is the $n\times n$ identity matrix, $\rho = \sigma_\varepsilon^2/\sigma^2$, and $\by = (y_1,\ldots,y_n)^\top$. The conditional mean is the best linear predictor of $Z(\bx)$ \citep{santner2013design}, which has the minimal  MSPE $\text{Var}[Z(x)|\by]$. 
		

Both the conditional mean and variance depend on the \textit{true correlation function} $\Psi$, which is often unknown in practice. Moreover, even if $\Psi$ is known, exact calculation of the conditional mean \eqref{mean} and variance \eqref{var} becomes increasingly difficult as the sample size $n$ grows, requiring $O(n^2)$ storage and $O(n^3)$ dense factorization cost. To address this challenge, we propose to replace the true correlation function $\Psi$ by an imposed correlation function $\Phi$ that is chosen to be compactly supported. 
Our proposed method is related to, but distinct from, covariance tapering \citep{furrer2006covariance,kaufman2008covariance,stein2013statistical}, which induces sparsity by multiplying a covariance function by a compactly supported taper. By contrast, we build compact support directly into the correlation function and allow its scale parameter to vary with $n$. Although compactly supported correlation functions with spatial scale parameters that vary with temporal lag have been studied in the space--time literature \citep{porcu2020nonseparable,senoussi2022nonstationary,faouzi2025compatibility}, our construction instead indexes the scale parameter by the sample size $n$. The compact support induces sparsity in the resulting correlation matrix, while the scale parameter of $\Phi$ is adjusted to preserve statistical accuracy.

		In this work, we take $\Phi$ from the generalized Wendland family \citep{bevilacqua2019estimation,chernih2014closed,fasshauer2015kernel,gneitingstationary,wendland2004scattered}, indexed by the imposed smoothness \(m\). Writing $h=\|\bx-\bx'\|_2$, define
		\begin{align}\label{eq_GWcorr}
			\Phi(\bx-\bx')=\varphi(h) := \left\{
			\begin{array}{lc}
				\frac{1}{B(2m-d-1,\mu+1)}\int_{h/\phi}^1 u(u^2-(h/\phi)^2)^{m-(d+3)/2}(1-u)^\mu {\rm d}u, & 0\leq h<\phi, \\
				0, & h\geq \phi,
			\end{array}\right.
		\end{align}
		where $\phi>0$ is the scale, or support-radius, parameter, $m>(d+1)/2$, $\mu \geq m$, and $B$ denotes the beta function. While numerical integration can be applied to computing \eqref{eq_GWcorr}, closed form solutions can be obtained when $m-(d+1)/2$ is a nonnegative integer. Table~\ref{tab_wendlandkernel} lists several closed-form examples. A larger $\phi$ enlarges the support radius, while a larger $m$ yields a smoother decay near the origin (Figure~\ref{fig:wendland_kernel_shapes}).

		\begin{figure}[t]
			\centering
			\includegraphics[width=0.8\linewidth]{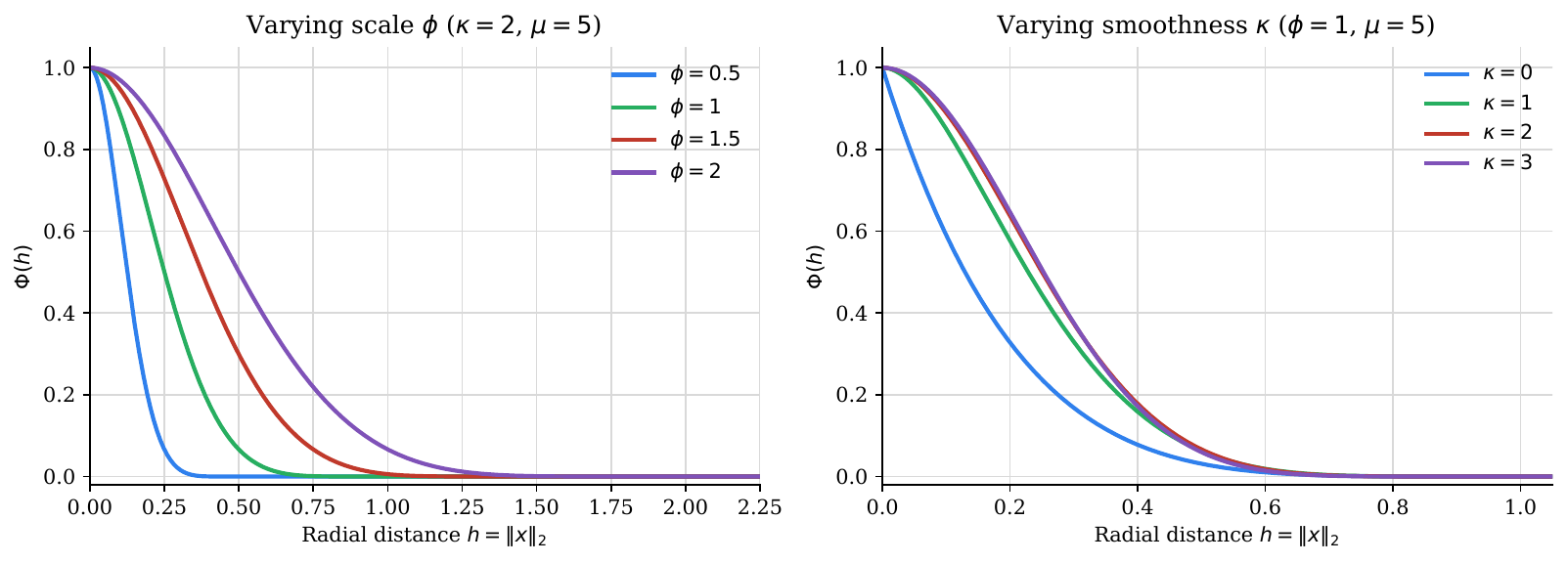}
			\caption{Illustration of generalized Wendland correlation functions as a function of radial distance $h=\|\bx-\bx'\|_2$.} 
		\label{fig:wendland_kernel_shapes}
	\end{figure}
	
	\begin{table}[ht]
		\centering
		\caption{Closed-form generalized Wendland correlations indexed by the imposed smoothness $m$. The final column reports the corresponding sample-path differentiability as $m-(d+1)/2$.}\label{tab_wendlandkernel}
		\vspace{0.5em}
		\small
		\setlength{\tabcolsep}{3pt}
		\resizebox{\linewidth}{!}{%
			\begin{tabular}{ccc}
				\toprule
				$m$ & Correlation function & $m-(d+1)/2$ \\
				\midrule
				$(d+1)/2$ & $(1 - r)^\mu_+$ & 0 \\
				$(d+3)/2$ & $(1 - r)^{\mu+1}_+(1 + r(\mu + 1))$ & 1 \\
				$(d+5)/2$ & $(1 - r)^{\mu+2}_+(1 + r(\mu + 2) + \frac{r^2}{3}(\mu^2 + 4\mu + 3))$ & 2 \\
				$(d+7)/2$ & $(1 - r)^{\mu+3}_+ (1 + r(\mu + 3) + \frac{r^2}{5}(2\mu^2 + 12\mu + 15) + \frac{r^3}{15}(\mu^3 + 9\mu^2 + 23\mu + 15))$ & 3\\
				\bottomrule
			\end{tabular}
		}
	\end{table}
	
	With the imposed correlation function $\Phi$, we define the Gaussian process predictor as
	\begin{align}\label{eq:GPpredic}
		\hat f_\Phi(\bx) = \rb_\Phi(\bx)^\top (\Rb_\Phi+\rho_n \Ib_n)^{-1} \by, \quad \bx\in \Omega,
	\end{align}
	where $\rb_\Phi(\bx)=(\Phi(\bx-\bx_1),\ldots,\Phi(\bx-\bx_n))^\top$ and $\Rb_\Phi=(\Phi(\bx_j-\bx_k))_{j k}$. We treat $\rho_n$ as a regularization parameter \citep{cressie2003spatial,stein1999interpolation}. By \eqref{eq_GWcorr}, as long as $\|\bx_i-\bx_j\|_2\geq \phi$, $\Phi(\bx_i-\bx_j) = 0$, and thus entries of $\Rb_\Phi$ corresponding to sufficiently distant input pairs are exactly zero. If $\phi$ is fixed as $n$ increases, as in the settings considered by \citet{tuo2020kriging} and \citet{wang2022gaussian}, the number of nonzero entries in $\Rb_\Phi$ is still typically of order $n^2$ under fixed-domain asymptotics. To induce asymptotic sparsity, we allow $\phi$ to decrease as the sample size increases, and denote it by $\phi_n$. If $\phi_n = o(1)$, then for well-distributed input points, the number of non-zero elements in $\Rb_\Phi$ is $o(n^2)$.
	
	To compute the predictor in \eqref{eq:GPpredic}, one needs to solve $(\Rb_\Phi+\rho_n \Ib_n)\bs=\by$, where many numerical methods can be applied. Once $\bs$ is obtained, it can be reused to evaluate $\hat f_\Phi(\bx)=\rb_\Phi(\bx)^\top\bs$ for any prediction point $\bx$. Usually, this linear system can be solved by iterative methods, which typically requires the matrix--vector multiplication $(\Rb_\Phi+\rho_n \Ib_n)\ba$ for a vector $\ba$, together with $O(n)$ inner products and vector updates. Because $\Phi$ is compactly supported with radius $\phi_n$, only input pairs within distance $\phi_n$ contribute; for well-distributed input points under fixed-domain asymptotics, each row therefore has order $n\phi_n^d$ nonzero entries, so the total number of nonzero entries of $\Rb_\Phi+\rho_n \Ib_n$ is $O(n^2\phi_n^d+n)$. Hence one matrix--vector multiplication iteration costs $O(n^2\phi_n^d+n)$ operations and the storage requirement is of the same order, compared with $O(n^2)$ storage and $O(n^3)$ factorization cost for a dense direct solve. If the linear system is solved in $T_n$ iterations (where $T_n$ is typically at most $n$), the overall solve cost is $O(T_n(n^2\phi_n^d+n))$, while each new prediction costs only $O(n\phi_n^d)$ because $\rb_\Phi(\bx)$ is itself also sparse. Thus, when $\phi_n=o(1)$, both storage and per-iteration arithmetic become subquadratic in $n$.
	
	The main statistical quantity of interest in this work is the MSPE using $\hat f_\Phi$, namely, 
	\begin{align}\label{eq:errorGP}
		\EE\|Z - \hat f_\Phi\|_{L_2(\Omega)}^2.
	\end{align} 
	This quantity measures the statistical cost of using the imposed compactly supported correlation in place of the true correlation. The scale parameter $\phi_n$ affects this risk and the computation in opposite ways: a smaller $\phi_n$ yields a sparser correlation matrix and cheaper sparse matrix-vector multiplication, but overly aggressive localization may discard dependence information needed for accurate prediction. The theoretical analysis in Section~\ref{sec_GPtheory} identifies sufficient joint choices of $\phi_n$ and $\rho_n$ under which this localization effect is controlled, so that the optimal integrated MSPE rate is preserved while the sparse-linear-algebra cost is reduced.
	
	In the rest of this work, we will use the following notation. Let $C,C',c_j,C_j, j\geq 0$ be positive constants, of which value can change from line to line. Let $a_n$ and $b_n$ be two positive sequences. We write $a_n\asymp b_n$ if, for some $C,C'>0$, $C\leq a_n/b_n \leq C'$, $a_n\gtrsim b_n$ if $a_n\geq Cb_n$ for some constant $C>0$, and $a_n\lesssim b_n$ if $a_n\leq C'b_n$ for some constant $C'>0$. The Fourier transform of $f\in L_1(\mathbb{R}^d)$ is given by
	$$\mathcal{F}(f)(\bomega)=(2\pi)^{-d/2}\int_{\mathbb{R}^d} f(\bx) e^{-{\rm i} \bx^\top\bomega}{\rm d} \bx.$$
	
	\section{Theoretical Analysis}\label{sec_GPtheory}
	
	In this section, we introduce our main theoretical results for GP regression using the imposed correlation function $\Phi$ with adjusted scale parameter $\phi_n$. Specifically, we show that the optimal integrated prediction rate can be preserved while the number of nonzero entries in the imposed covariance matrix is reduced, and thus leads to sparse matrix-vector multiplication.

	Throughout this work, we consider fixed designs, i.e., the input points $\bx_1,\ldots,\bx_n$ are nonrandom. Fixed designs are widely used in spatial statistics \citep{cressie2015statistics,stein1999interpolation}, where locations are typically predetermined rather than randomly sampled. They also play a central role in computer experiments \citep{santner2013design} and in engineering design and optimization \citep{jones1998efficient}. In particular, we focus on \textit{quasi-uniform designs} \citep{borodachov2007asymptotics,utreras1988convergence}.
	
	\begin{definition}[Fill distance and separation radius]
		Let $\bX= \{\bx_1,...,\bx_n\}$ be a set of input points. The fill distance and the separation radius of $\bX$ are respectively defined by 
		\begin{align*}
			h_{\bX,\Omega}= \sup_{\bx\in\Omega}\inf_{\bx_j\in \bX}\|\bx-\bx_j\|_2, \qquad q_{\bX}=\min_{1\leq j\neq k\leq n}\|\bx_j-\bx_k\|_2/2.
		\end{align*}
	\end{definition}
	
	\begin{definition}[Quasi-uniform designs]\label{defn_quasi}
		Let $\cX = \{\bX_1,\bX_2,...\}$ be a sequence of sets of input points. Without loss of generality, assume that ${\rm card}(\bX_n)=n$, where $n$ takes its value in an infinite subset of $\mathbb{N}$, and card$(X)$ denote the cardinality of set $X$. We call $\cX$ a \textit{sampling scheme} \citep{tuo2020kriging}. A sampling scheme $\cX = \{\bX_1,\bX_2,...\}$ is quasi-uniform if $h_{\bX_n,\Omega}/q_{\bX_n}\leq C$ for all $n$. 
	\end{definition}
	It is easy to check that $h_{\bX,\Omega}\geq q_{\bX}$ \citep{wendland2004scattered} for any set of points $\bX$. For a quasi-uniform sampling scheme, $h_{\bX_n,\Omega}\asymp q_{\bX_n}\asymp n^{-1/d}$, while Example 1 of \cite{tuo2020kriging} shows that random samplings are not quasi-uniform.
	
	We impose the following assumptions in this paper.
	
	\begin{assumption}\label{assum_region}
		The region of interest $\Omega\subset \RR^d$ is a compact set with Lipschitz boundary, and satisfies an interior cone condition. That is, there exist $\alpha\in (0,\pi/2)$ and $ \mathcal{R} > 0$ such that for every $\bx\in\Omega$, a unit vector $\bxi(\bx)$ exists such that the cone
		$\mathcal{C}(\bx,\bxi(\bx),\alpha,\mathcal{R}):=\left\{\bx+t \bbeta:\bbeta\in\mathbb{R}^d,\|\bbeta\|=1,\bbeta^\top\bxi(\bx)\geq \cos\alpha,t\in[0,\mathcal{R}]\right\} $
		is contained in $\Omega$.
	\end{assumption}
	
	\begin{assumption}\label{assum_Psi}
		The true correlation function $\Psi$ has smoothness \(m_0>d/2\) and is either a generalized Wendland correlation function as in \eqref{eq_GWcorr} with $\phi$ as a fixed constant and smoothness \(m_0\), or the isotropic Mat\'ern correlation functions \citep{stein1999interpolation} defined by
		\begin{align}\label{materngai}
			\Psi_M(\bx)=\frac{1}{\Gamma(m_0-d/2)2^{m_0-d/2 -1}}
			\left(2\sqrt{m_0-d/2}\,\theta \|\bx\|_2\right)^{m_0-d/2}
			K_{m_0-d/2}\left(2\sqrt{m_0-d/2}\,\theta\| \bx\|_2\right),
		\end{align}
		where $\theta>0$ is the scale parameter and $K_{m_0-d/2}$ is the modified Bessel function of the second kind of order $m_0-d/2$.
	\end{assumption}
	
	\begin{assumption}\label{assum_noise}
		Suppose $\varepsilon_k$'s in \eqref{recoveringGP} are i.i.d. sub-Gaussian random variables \citep{geer2000empirical}, i.e., satisfying
		\begin{align*}
			C^2 (\mathbb{E} e^{|\varepsilon_k|^2/C^2}-1)\leq C', \quad k=1,...,n,
		\end{align*}
		for constants $C,C'>0$ independent of $k$ and $n$.
	\end{assumption}
	
	\begin{assumption}\label{assum_quasiuniform}
		The sampling scheme $\cX = \{\bX_1,\bX_2,...\}$ is quasi-uniform.
	\end{assumption}
	
	Assumption \ref{assum_region} imposes a geometric condition on the region $\Omega$, which is quite standard in the literature \citep{wendland2004scattered}. In particular, if $\Omega$ is compact and convex, then Assumption \ref{assum_region} is satisfied \citep{niculescu2006convex}. Assumption \ref{assum_Psi} implies that the underlying GP is generated by two popular correlation functions used in GP regression \citep{bevilacqua2019estimation,stein1999interpolation}.  Assumptions \ref{assum_noise} and \ref{assum_quasiuniform} are standard assumptions used in GP regression \citep{tuo2020kriging}.
	
	With these assumptions, we establish the following theorem.
	
	\begin{theorem}\label{thm_GPconvergence_rate}
		Suppose $m\geq m_0>d/2$. Let $q_n=q_{\bX_n}$ and	\(\gamma_n=\left(\phi_n^d\rho_n^{-1}q_n^{-d}\right)^{1/(2m)}.\)
		Under Assumptions \ref{assum_region}-\ref{assum_quasiuniform}, suppose that $0<\phi_n\leq 1$, $\rho_n>0$, $\gamma_n\geq w_0$, and $\phi_n^{-1}\gamma_n q_n\leq C_\gamma$ for all sufficiently large $n$, where $w_0$ is the constant in Lemma \ref{lem:wendland_fourier_scaled} in the Supplementary materials and $C_\gamma>0$ is fixed. Then
		\begin{align}\label{eq_thm_GPconvergence_rate}
			& \EE\|Z - \hat f_\Phi\|_{L_2(\Omega)}^2\nonumber\\
			\leq C & \bigg((\phi_n^{\frac{m_0(2m-d)}{m}}\left(\frac{\rho_n}{n}\right)^{\frac{m_0-m}{m}} + \phi_n^{-d} +\rho_n^{-1})\left(\phi_n^{-\frac{d(2m-d)}{2m}}\left(n^{-\frac{2m}{d}}\phi_n^{-2m+d} + \frac{\rho_n}{n}\right)^{1-\frac{d}{2m}}\right)\nonumber\\
			& + \phi_n^{\frac{(2m-d)(2m_0-d)}{2m}}\left(\frac{\rho_n}{n}\right)^{\frac{2m_0 - d}{2m}}\bigg),
		\end{align}
		where $C$ is a constant not depending on $n$, $\phi_n$, or $\rho_n$. 
	\end{theorem}
	Theorem \ref{thm_GPconvergence_rate} presents an integrated MSPE bound for prediction using a generalized Wendland correlation whose smoothness is at least as large as the smoothness of the true correlation function. By balancing the terms in \eqref{eq_thm_GPconvergence_rate}, we obtain the following sufficient rate choices.
	
	\begin{proposition}\label{prop_GP_phimunrelation}
		Let $\zeta_Z(\phi_n,\rho_n)$ denote the right-hand-side of \eqref{eq_thm_GPconvergence_rate}. It holds that $\zeta_Z(\phi_n,\rho_n) \asymp n^{-\frac{2m_0-d}{2m_0}}$ by taking $\rho_n\phi_n^{2m-d} \asymp n^{-\frac{m-m_0}{m_0}}$
		with $\phi_n^d \geq \rho_n$ and $n^{1-\frac{2m}{d}}\phi_n^{-2m+d}\leq \rho_n$.
	\end{proposition}
	
	
	Proposition \ref{prop_GP_phimunrelation} reveals a specific relationship between the parameters $\phi_n$ and $\rho_n$, suggesting that appropriate choices of these parameters can lead to optimal convergence rates of the MSPE $\EE\|Z - \hat f_\Phi\|_{L_2(\Omega)}^2$. Based on Proposition \ref{prop_GP_phimunrelation}, the following corollary shows that if the underlying truth is a GP generated by a Mat\'ern correlation function and we use a generalized Wendland correlation function for prediction, then we can also obtain the convergence rate in Corollary 7 of \cite{wang2022gaussian}, where they showed the case that both the true and imposed correlation functions are Mat\'ern correlation functions.
	
	\begin{corollary}\label{coro_GP_phi1}
		Let $\phi_n \asymp 1$ and $\rho_n\asymp n^{-\frac{m}{m_0}+1}$. Under conditions in Theorem \ref{thm_GPconvergence_rate}, we have
		\begin{align}\label{ineq_coro_GPconver}
			\EE\|Z - \hat f_\Phi\|_{L_2(\Omega)}^2\lesssim n^{-\frac{2m_0-d}{2m_0}}.
		\end{align}
	\end{corollary}
	
	Choosing $\phi_n \asymp 1$, as described in Corollary \ref{coro_GP_phi1}, results in an asymptotically dense matrix under fixed-domain asymptotics. In contrast, if $\phi_n = o(1)$, the matrix $\Rb_\Phi$ becomes sparse. Specifically, if $\phi_n \asymp n^{-\alpha}$, then the resulting covariance matrix contains $O(n^{2-\alpha d})$ nonzero elements for quasi-uniform designs. From the perspective of sparse matrix-vector cost, smaller $\phi_n$ is desirable, but Proposition \ref{prop_GP_phimunrelation} constrains how quickly $\phi_n$ can shrink while preserving the optimal rate. Taking $\rho_n\asymp \phi_n^d$ gives $\phi_n\asymp n^{-\frac{m-m_0}{2mm_0}}$. Based on this reasoning, we present the following corollary, which demonstrates that the same convergence rate as in Corollary \ref{coro_GP_phi1} can be achieved with a shrinking support radius.
	
	\begin{corollary}\label{coro_GP_phismall}
		Let $\phi_n\asymp n^{-\frac{m-m_0}{2mm_0}}$ and $\rho_n\asymp \phi_n^d$. Under conditions in Theorem \ref{thm_GPconvergence_rate}, \eqref{ineq_coro_GPconver} holds.
	\end{corollary}
	
	Corollary \ref{coro_GP_phismall} shows that by choosing $\phi_n\asymp n^{-\frac{m-m_0}{2mm_0}}$ and $\rho_n\asymp \phi_n^d$, we can still have an estimator that has an optimal convergence rate. Regarding computation, note that the sparsity result controls the cost of a sparse matrix-vector multiplication. We summarize the sparsity and per-iteration implication in the following proposition.


	\begin{proposition}\label{prop_cgd_gp}
		Let $\phi_n\asymp n^{-\frac{m-m_0}{2mm_0}}$ and $\rho_n\asymp \phi_n^d$. Under Assumption \ref{assum_quasiuniform}, $\Rb_\Phi+\rho_n\Ib_n$ has \(O\left(n^{2-\frac{(m-m_0)d}{2mm_0}}\right)\) nonzero entries. Consequently, one sparse matrix-vector multiplication with $\Rb_\Phi+\rho_n\Ib_n$ costs \(O\left(n^{2-\frac{(m-m_0)d}{2mm_0}}\right)\) operations. 
	\end{proposition}
	
	Proposition \ref{prop_cgd_gp} highlights the sparse-linear-algebra induced by $\phi_n$. The parameter $m_0$, which is the smoothness of the underlying true process, affects the admissible decay of $\phi_n$. A larger $m_0$ implies a slower decay of $\phi_n$, resulting in a less sparse matrix $\Rb_\Phi + \rho_n \Ib_n$ but a faster statistical convergence rate. For fixed true smoothness \(m_0\), a smoother imposed Wendland kernel, corresponding to larger \(m\), permits a faster admissible decay of \(\phi_n\) and hence greater sparsity. As $m$ goes to infinity, the corresponding matrix in the estimator has $O(n^{2-\frac{d}{2m_0}})$ nonzero elements. For fixed imposed smoothness \(m\), a rougher true process, corresponding to smaller \(m_0\), also permits more aggressive sparsification.
	
	\begin{remark}
		Besides the choice of compactly supported covariance, additional solver-time gains can be obtained through preconditioning, sparse Cholesky factorization, Nystr\"om-type approximations, or multiresolution methods. These are largely orthogonal linear-algebra accelerations and can in principle be combined with the present scale-adjusted construction. Covariance tapering is more directly related statistically, because it also uses compact support; the key difference is that tapering sparsifies a parent covariance by a Hadamard product, whereas our analysis studies direct prediction with a shrinking-range compactly supported working covariance.
	\end{remark}
	
	\section{Kernel Ridge Regression with Compactly Supported Correlation Function}\label{sec_KRR}
	
	In this section, we extend the analysis from GP to KRR, focusing on the effect of adjusting the scale parameter to achieve optimal convergence rates while reducing sparse matrix-vector cost. KRR operates in the framework where the underlying truth is a deterministic function residing in an RKHS; a brief introduction to RKHS is provided in Section \ref{app:introtoSoboRKHS} of the supplementary materials. Suppose we have observed a set of input-output pairs $\{(\bx_i, y_i)\}_{i=1}^n \subset \cX \times \mathbb{R}$ following the relationship $y_i = f(\bx_i)+\varepsilon_i$. The KRR recovers the underlying true function $f$ via
	\begin{align}\label{eq_KRRpredict_eq}
		\hat f_n = \arg\min_{f \in \cN_k} \frac{1}{n} \sum_{i=1}^n (f(\bx_i) - y_i)^2 + \lambda_n \|f\|_{\cN_k}^2 ,
	\end{align}
	where $\lambda_n > 0$ is a regularization parameter, $\cN_k$ is the RKHS generated by the kernel function $k(\cdot,\cdot)$, and $\|\cdot\|_{\cN_k}$ is the corresponding RKHS norm. By the representer theorem, $\hat f_n$ has the form
	\begin{align}\label{eq_KRRpredict}
		\hat f_n = \rb_k(\bx)^\top (\Rb_k+n\lambda_n \Ib_n)^{-1} \by,
	\end{align}
	where $\rb_k(\bx)=(k(\bx, \bx_1),\ldots,k(\bx,\bx_n))^\top, \Rb_k=(k(\bx_j,\bx_k))_{j k}$, and $\by = (y_1,...,y_n)^\top$. Therefore, in order to obtain the estimator $\hat f_n$, a linear system must be solved. Following the approach in Section \ref{sec_GP}, we utilize the generalized Wendland kernel defined in \eqref{eq_GWcorr} with scale parameter $\phi_n$ to sparsify the matrix. Specifically, we take the kernel function $k(\bx,\bx')=\Phi(\bx-\bx')$, so \eqref{eq_KRRpredict} has the same algebraic form as \eqref{eq:GPpredic} with $n\lambda_n$ in place of $\rho_n$. However, the RKHS and GP settings lead to distinct convergence-rate analyses. We need the following assumption for the main results in this section.
	
	\begin{assumption}\label{assum_finRKHSPsi}
		The underlying function $f\in \cN_{\Psi}$, where $\cN_{\Psi}$ is the RKHS generated by the kernel function $\Psi$ as in \eqref{materngai}.
	\end{assumption}
	
	The following theorem states convergence rates of prediction error of $\hat f_n$ under $L_2$ metric.
	
	\begin{theorem}\label{thm_KRRratesX}
		Suppose Assumptions \ref{assum_region}, \ref{assum_noise}, \ref{assum_quasiuniform}, and \ref{assum_finRKHSPsi} hold, and suppose \(m\ge m_0>d/2\). If $\lambda_n\phi_n^{2m-d} \asymp n^{-\frac{2m}{2m_0+d}}$ with $\lambda_n \lesssim \phi_n^{d + \frac{2m m_0}{m - m_0}}$ and $\phi_n \gtrsim n^{-\frac{2m_0}{d(2m_0+d)}}$, then
		\[
		\|\hat f_n-f\|_{L^2(\Omega)}=O_P\left(n^{-m_0/(2m_0+d)}\right).
		\]
	\end{theorem}
	
	Theorem \ref{thm_KRRratesX} shows that the optimal $L_2$ convergence rate $n^{-m_0/(2m_0+d)}$ remains attainable even when the imposed kernel is compactly supported and its support radius $\phi_n$ is allowed to shrink with $n$. The condition $\lambda_n \phi_n^{2m-d} \asymp n^{-2m/(2m_0+d)}$ indicates that $\lambda_n$ and $\phi_n$ must be chosen jointly: the product $\lambda_n \phi_n^{2m-d}$ plays the role of an effective smoothing level, whereas the additional constraints on $\lambda_n$ and $\phi_n$ prevent the support radius from shrinking so quickly that the statistical rate degrades. Thus, the theorem identifies a regime in which statistical optimality is preserved while compact support can still be exploited for computation.
	
	
	The corollaries derived from Theorem \ref{thm_KRRratesX} further illustrate the practical implications of these parameter choices. For example, if we set $\phi_n \asymp 1$, we recover established convergence rates from \cite{wang2022gaussian}, demonstrating that our framework generalizes existing results while introducing flexibility through scale parameter adjustments, as shown in the following corollary.
	
	\begin{corollary}
		Suppose $\phi_n\asymp 1$ and $\lambda_n\asymp n^{-\frac{2m}{2m_0+d}}$. Under conditions of Theorem \ref{thm_KRRratesX}, 
		\begin{align}\label{eq_coro_krrrate}
			\|f-\hat f_n\|_{L_2} = O_p\left(n^{-\frac{m_0}{2m_0+d}}\right).
		\end{align}
	\end{corollary}
	
	However, from the perspective of sparse matrix-vector cost, it is often advantageous to choose a smaller $\phi_n$. This choice leads to a sparse matrix and fewer nonzero entries. By appropriately adjusting $\phi_n$ and $\lambda_n$, similar convergence rates can be achieved while reducing the per-iteration sparse matrix-vector cost. For example, setting $\phi_n \asymp n^{-\frac{m - m_0}{m(2m_0 + d)}}$ and $\lambda_n \asymp n^{-\frac{dm + m_0(2m - d)}{m(2m_0 + d)}}$ achieves both optimal convergence and reduced sparse matrix-vector cost, as presented in the following corollary.

	\begin{corollary}\label{coro_KRR_phinsmall}
		Suppose $\lambda_n\phi_n^{2m-d} \asymp n^{-\frac{2m}{2m_0+d}}$ with $\lambda_n \leq \phi_n^{d + \frac{2m m_0}{m - m_0}}$. Under conditions of Theorem \ref{thm_KRRratesX}, \eqref{eq_coro_krrrate} holds.
		In particular, we can take $\phi_n \asymp n^{-\frac{m - m_0}{m(2m_0 + d)}}$ and $\lambda_n \asymp n^{-\frac{dm + m_0(2m - d)}{m(2m_0 + d)}}$.
	\end{corollary}
	
	Similar to our previous analysis in Section \ref{sec_GP}, the choice of parameters $\phi_n$ and $\lambda_n$ plays an important role in balancing statistical accuracy with sparse matrix-vector cost. By carefully choosing the scale parameter $\phi_n$ and the regularization parameter $\lambda_n$, we can achieve optimal convergence while reducing the number of nonzero entries in the kernel matrix. Corollary \ref{coro_KRR_phinsmall} demonstrates that by setting $\phi_n \asymp n^{-\frac{m - m_0}{m(2m_0 + d)}}$ and $\lambda_n \asymp n^{-\frac{dm + m_0(2m - d)}{m(2m_0 + d)}}$, we achieve the optimal convergence rate $O_p\left(n^{-\frac{m_0}{2m_0+d}}\right)$ for the estimator $\hat f_n$ under the specified conditions. The sparsification induced by a smaller $\phi_n$ is particularly beneficial for large datasets, where matrix-vector multiplication and storage costs are limiting factors. The following proposition states the resulting sparsity and sparse matrix-vector multiplication cost.
	
	\begin{proposition}\label{prop_KRR_comp}
		Let $\phi_n \asymp n^{-\frac{m - m_0}{m(2m_0 + d)}}$ and $\lambda_n \asymp n^{-\frac{dm + m_0(2m - d)}{m(2m_0 + d)}}$. Under Assumption \ref{assum_quasiuniform}, $\Rb_\Phi+n\lambda_n\Ib_n$ has \(O\left(n^{2-\frac{(m - m_0)d}{m(2m_0 + d)}}\right)\) nonzero entries. Consequently, one sparse matrix-vector multiplication with $\Rb_\Phi+n\lambda_n\Ib_n$ costs \(O\left(n^{2-\frac{(m - m_0)d}{m(2m_0 + d)}}\right)\) operations.
	\end{proposition}
	
	
	Proposition \ref{prop_KRR_comp}, together with Corollary \ref{coro_KRR_phinsmall}, highlights similar statistical and sparse-linear-algebra phenomenon as in the GP setting. For fixed true smoothness \(m_0\), a larger imposed smoothness \(m\) permits a faster decay of \(\phi_n\) within this sufficient schedule, and in the limiting case \(m\to\infty\), the sparse matrix-vector cost approaches $O\left(n^{2-\frac{d}{2m_0+d}}\right)$ per iteration. On the other hand, a larger \(m_0\) leads to a slower admissible decay of \(\phi_n\), producing a less sparse matrix but a faster statistical convergence rate \(O_p(n^{-m_0/(2m_0+d)})\). Thus, smoother target functions are statistically easier to estimate, but preserving the optimal rate may require retaining more local dependence in the kernel matrix.
	
	
	Similar to the GP regression case, additional linear-algebra accelerations, such as preconditioning, sparse Cholesky factorization, Nystr\"om-type approximations, local approximations, or multiresolution methods, can be combined with the proposed scale-adjusted compactly supported kernel. These techniques are complementary to the shrinking-support construction studied here, and their costs are separate from the sparse matrix-vector bounds in Proposition \ref{prop_KRR_comp}.


	\section{Numerical  Experiments}\label{sec_numerical}
	
	Our numerical studies focus on the statistical and sparse-linear-algebra properties created by the shrinking support radius. We compare against Mat\'ern predictors as dense-kernel accuracy baselines and use the same iterative solver backend for Mat\'ern and Wendland systems when reporting computation time. These experiments are not intended to be an exhaustive benchmark against all scalable GP approximations. The reason is that, methods such as preconditioning, sparse Cholesky factorization, Nystr\"om approximation, localization, and multiresolution approximations are complementary linear-algebra strategies that could be combined with compact support. Covariance tapering is the closest statistical comparator; the distinction from our direct shrinking-range working covariance is discussed in the Introduction. However, even for the covariance tapering, one can still combine it with our proposed method, by multiplying another covariance function with compact support.
	
	The numerical experiments should be interpreted as finite-sample illustrations of the accuracy and sparsity properties, rather than as a literal verification of every asymptotic tuning condition. The theoretical results give sufficient joint scalings of the support radius and regularization parameters under fixed-domain asymptotics. In finite samples, however, the multiplicative constants in these scalings and the nugget or regularization level have a substantial effect and are typically selected or fixed according to modelling convention. We therefore fix the nugget or regularization parameter in each experiment for simplification and use the same value for the dense Matérn baseline and the compactly supported Wendland method for fair comparison. This choice isolates the effect of the compact support and its scale, while avoiding an additional layer of hyperparameter optimization. The support radius is still chosen according to the theory-guided power law \(\phi_n=C_\phi n^{-\alpha}\), with \(C_\phi\) selected by validation and then held fixed across reporting replications. Thus the experiments are intended to show that the shrinking-support construction can maintain competitive prediction accuracy while reducing sparse matrix-vector cost, rather than to claim finite-sample optimality of the full asymptotic tuning prescription.
	
	\subsection{Gaussian Process Regression}\label{subsec:gp_sim}
	
	This section illustrates the finite-sample behavior suggested by the theoretical findings in Section~\ref{sec_GPtheory}.
	Specifically, our numerical investigation shows that compactly supported Wendland kernels achieve prediction accuracy comparable to the Mat\'{e}rn-kernel baseline, and 
	the induced sparsity is reflected in the measured solver times, subject to the usual dependence on conditioning, stopping tolerance, and sparse-solver overhead.
	
	The detailed GP simulation setup, including the data-generating covariance, sample sizes, train/validation/test split, solver backend, and validation rule for selecting $C_\phi$, is reported in Supplementary Section~\ref{app:experimental_settings}.
	
	Figure~\ref{fig:gp_convergence} summarises the scheduled-bandwidth test mean squared error (MSE) trends for both $d=2$ and $d=3$. It can be seen that the Wendland predictor is comparable in accuracy to the Mat\'{e}rn baseline in all GP settings, with the largest reported 30-replication Wendland/Matern MSE ratio equal to $1.746$. Figure~\ref{fig:gp_convergence} and Table~\ref{tab:slopes} further show the same qualitative decreasing behaviour as the Mat\'{e}rn baseline.
	
	\begin{figure}[t!]
		\centering
		\setlength{\subfigcapskip}{-14pt}
		\subfigure[$d=2$]{
				\includegraphics[width=0.95\linewidth]{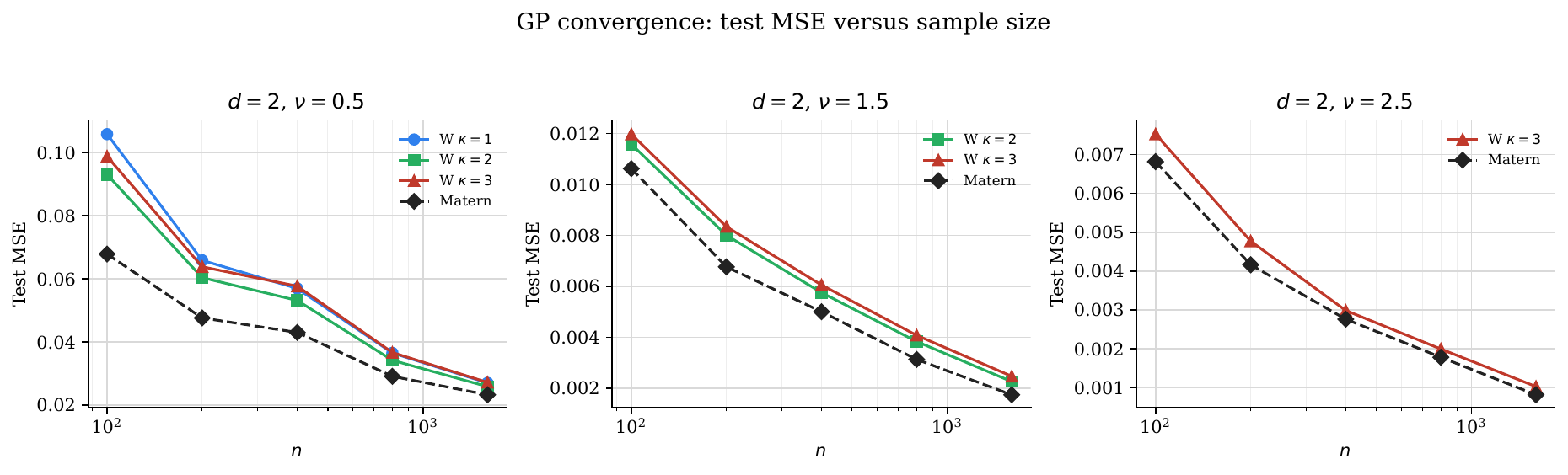}%
				
			}
			\subfigure[$d=3$]{
					\includegraphics[width=0.95\linewidth]{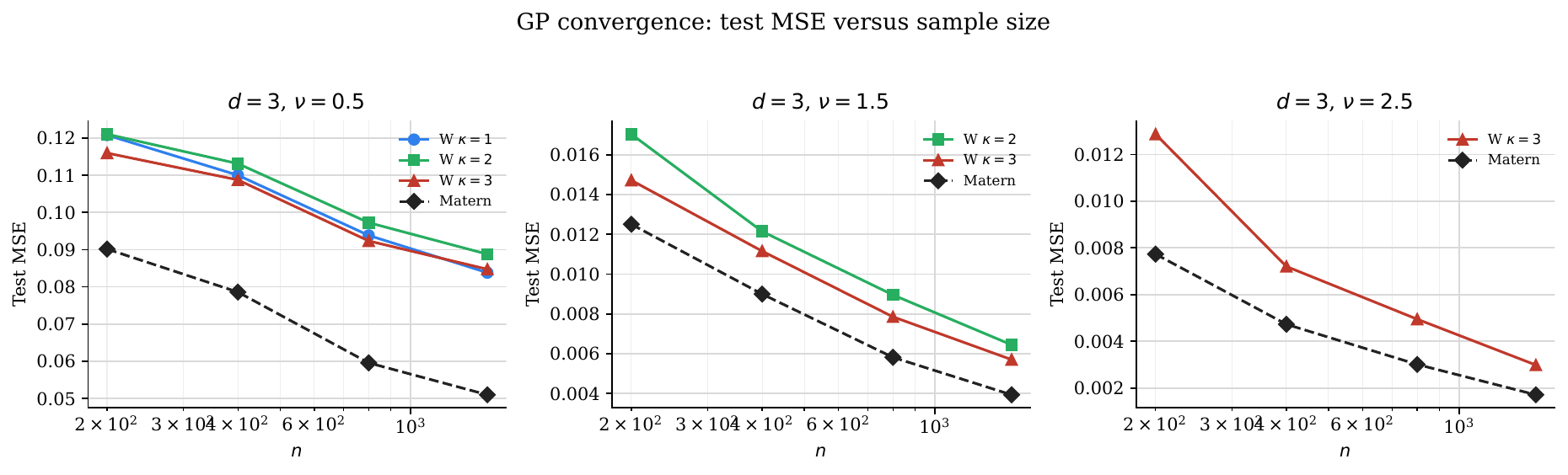}%
					
				}
				\caption{GP test MSE versus sample size $n$ under the scheduled scale parameter for (a) $d=2$ and (b) $d=3$. The Mat\'{e}rn smoothness $m_0-d/2$ (denoted $\nu$ in the panels) takes values $\{0.5,\,1.5,\,2.5\}$, with the corresponding Wendland smoothness values $m-(d+1)/2$ (denoted $\kappa$) in $\{1,\,2,\,3\}$. Each configuration is evaluated over 30 independent replications.}
			\label{fig:gp_convergence}
		\end{figure}
		
		Figure~\ref{fig:gp_computation_time} reports computation time versus $n$ for $d \in \{2,3\}$, using the same simulation settings as the full GP tables. In fact, speed gains depend on whether the resulting Wendland matrix is sufficiently sparse to offset the iterative-solver overhead. In our experiments, all conjugate gradient (CG) solves converged.  Wendland correlation function was faster in 51 of the 54 reported GP table rows, and the median speedup was about $2.27\times$.
		
		\begin{figure}[htbp]
			\centering
			\includegraphics[width=\linewidth]{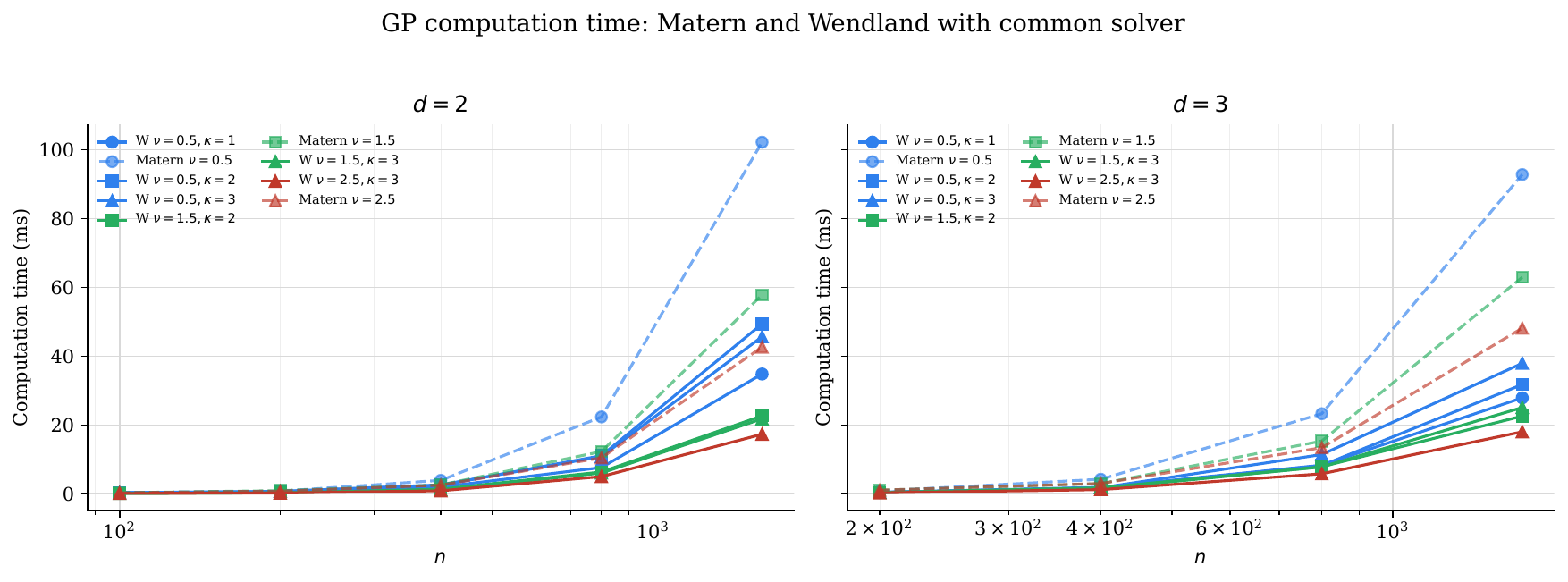}
			\caption{GP computation time under the scheduled scale parameter. Wendland GP is compared with Mat\'{e}rn GP for $d=2$ and $d=3$; both methods use the same CG solver backend.}
			\label{fig:gp_computation_time}
		\end{figure}
		
		Collectively, these experiments numerically confirm that the corrected Wendland implementation has the expected decreasing MSE trends. The local timing results also highlight that solver-time gains require both smaller selected bandwidths and an efficient sparse Cholesky/CG implementation; otherwise the constant factors of generic sparse solvers can dominate at the present sample sizes.
		
		\subsection{Deterministic Function Estimation}
		
		We next examine deterministic function estimation under the RKHS framework of Section~\ref{sec_KRR}. In this subsection, the underlying regression function is fixed rather than sampled from a Gaussian process. The detailed deterministic functions, noise model, sample sizes, data split, and validation rule for $C_\phi$ are reported in Supplementary Section~\ref{app:experimental_settings}.
		
		\begin{figure}[t!]
			\centering
			\setlength{\subfigcapskip}{-14pt}
		\subfigure[$d=2$]{%
					\includegraphics[width=0.98\linewidth]{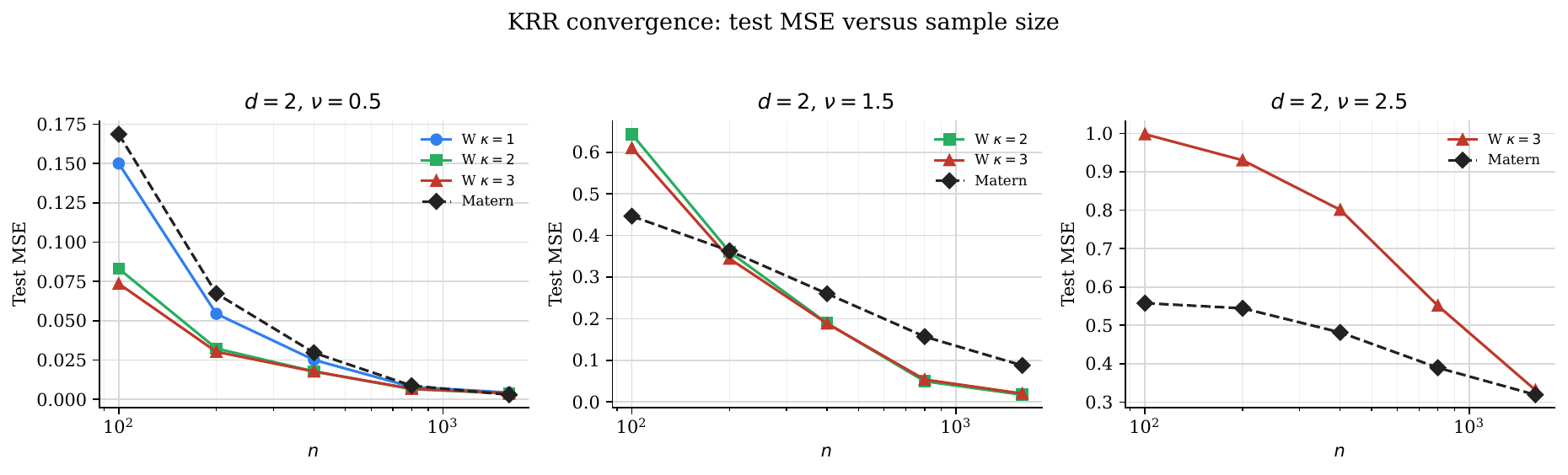}%
					
				}\\[0.75em]
					\subfigure[$d=3$]{%
						\includegraphics[width=0.98\linewidth]{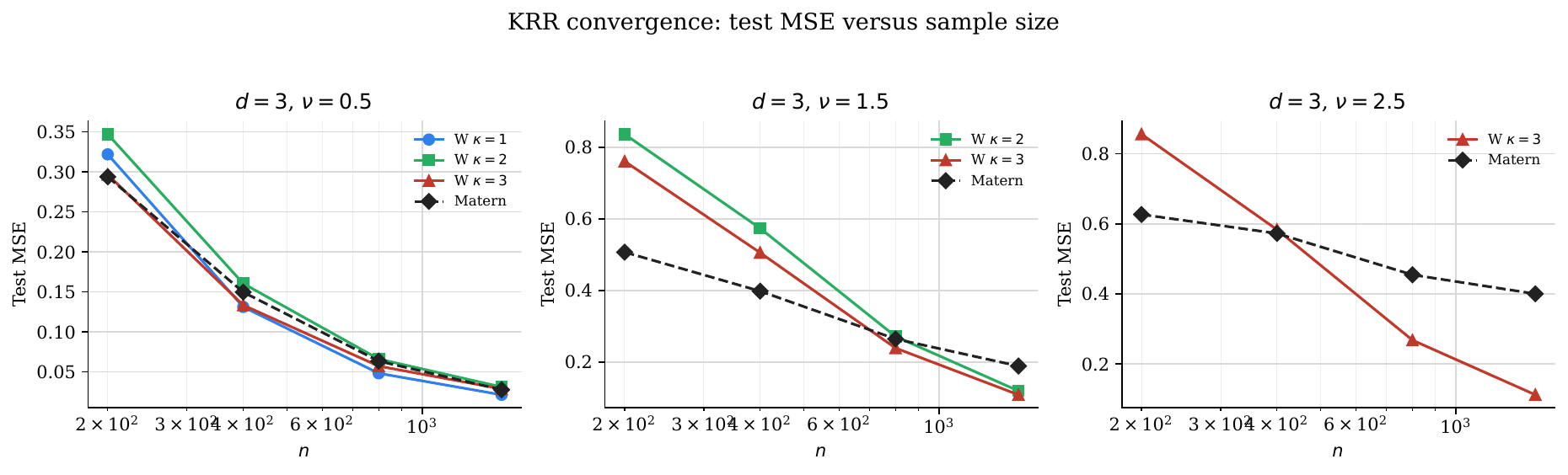}%
						
					}
					\caption{KRR test MSE versus sample size $n$ under the scheduled scale parameter for (a) $d=2$ and (b) $d=3$. The Mat\'{e}rn smoothness $m_0-d/2$ (denoted $\nu$ in the panels) takes values $\{0.5,\,1.5,\,2.5\}$, with the corresponding Wendland smoothness values $m-(d+1)/2$ (denoted $\kappa$) in $\{1,\,2,\,3\}$. Each configuration is evaluated over 30 independent replications.}
					\label{fig:krr_convergence}
				\end{figure}
				
				\begin{figure}[t!]
					\centering
					\includegraphics[width=\linewidth]{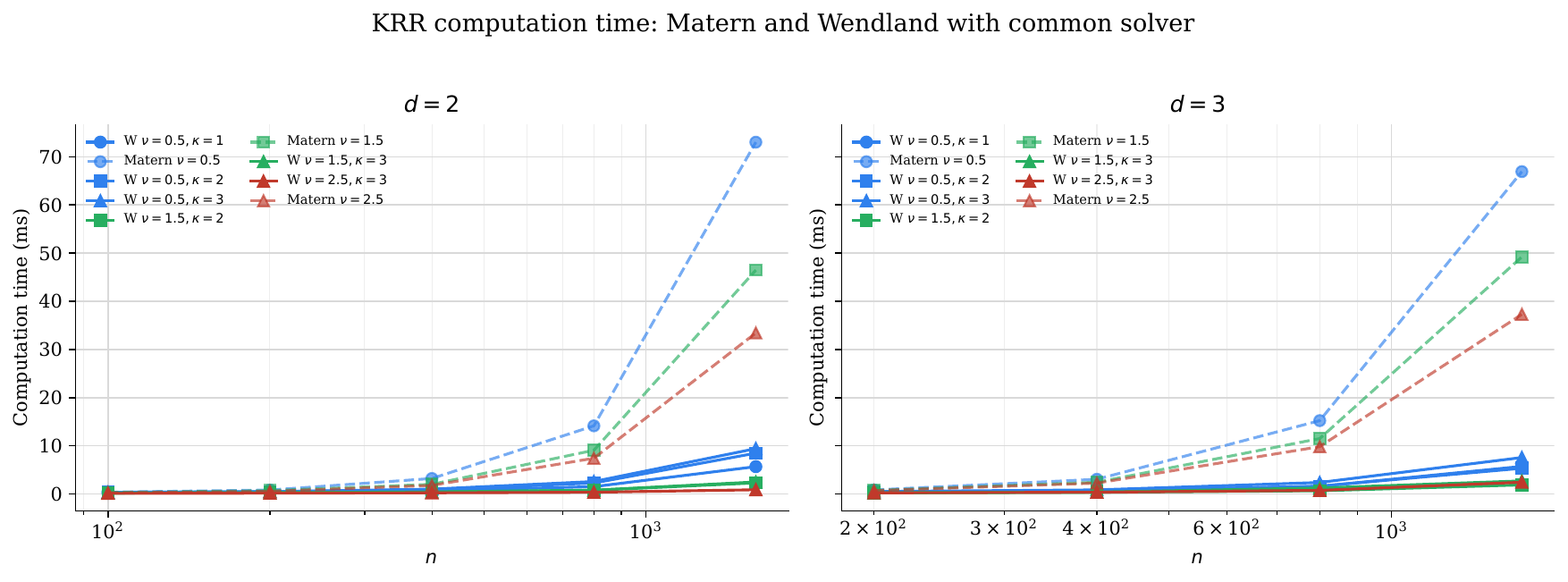}%
					\vspace{-1.25em}
					\caption{KRR computation time versus sample size $n$ under the scheduled scale parameter for $d=2$ (left) and $d=3$ (right). The Mat\'{e}rn smoothness $m_0-d/2$ (denoted $\nu$ in the panels) takes values $\{0.5,\,1.5,\,2.5\}$, with the corresponding Wendland smoothness values $m-(d+1)/2$ (denoted $\kappa$) in $\{1,\,2,\,3\}$. Each configuration is evaluated over 30 independent replications, and the Mat\'{e}rn and Wendland use the same conjugate-gradient solver backend.}
				\label{fig:krr_computation_time}
			\end{figure}

			Figure~\ref{fig:krr_convergence} shows decreasing deterministic test error as $n$ increases in both $d=2$ and $d=3$. The Wendland KRR estimator is competitive with the Mat\'{e}rn KRR baseline throughout the reported configurations; the largest 30-replication Wendland/Mat\'{e}rn MSE ratio is $1.790$. Several deterministic settings have ratios below one, which reflects the fact that the Mat\'{e}rn estimator is a comparison kernel with the same regularization level, not an oracle matched to the deterministic truth. As for the sparse-kernel solver-time, the Wendland matrices retain meaningful sparsity while maintaining satisfactory prediction accuracy.
			In the regenerated deterministic run, all CG solves converged, Wendland KRR is faster in all reported table rows, and the median speedup is about $5.89\times$. These deterministic experiments further confirm that compact support can reduce sparse matrix-vector cost without destroying prediction accuracy.


			\subsection{Additional Comparisons with Scalable Approximations}\label{subsec:additional_gp_baselines}
			
			We next compare the selected scale-adjusted Wendland predictor with several common scalable approximations. 
			The additional comparators are covariance tapering, local GP prediction, and the Nystr\"om approximation; the detailed design, tuning, timing conventions, and full tables are given in Supplementary Section~\ref{app:experimental_settings} and Supplementary Tables~\ref{tab:additional_gp_baselines_full}--\ref{tab:additional_det_baselines_full}. We also report the deterministic counterpart of the same scalable-baseline comparison. The corresponding truth, noise level, kernel choices, and tuning rules are collected with the other experimental settings in Supplementary Section~\ref{app:experimental_settings}. 
			
			
			Figure~\ref{fig:additional_baselines} compares the proposed Wendland predictor with several representative scalable alternatives in the selected \(d=2\) setting. For the regular Mat\'ern GP truth, the dense parent GP and the local GP predictor give the smallest test errors. The Wendland correlations and tapered-parent predictors have slightly larger MSEs, but remain close to the dense baseline as \(n\) increases. Thus, in this GP benchmark, the compact-support global methods introduce only a modest accuracy loss relative to the dense parent predictor. The computation-time panel shows the benefit of sparse global linear algebra. For example, at \(n=1600\), the Wendland predictor has solver time \(37.3\) ms, compared with \(128\) ms for the dense parent and \(55.4\) ms for the tapered-parent predictor. The faster Wendland solve is explained by the scheduled shrinking support radius, which gives about \(65.4\%\) exact zero entries, compared with about \(51.7\%\) for the fixed-radius tapered-parent matrix. Local GP is accurate but requires repeated local calculations at the prediction sites, giving mean computation time \(98.8\) ms under the timing convention used here. The Nystr\"om approximation is the fastest method in this benchmark, but its accuracy is less stable at smaller sample sizes; it becomes close to the compact-support methods only as \(n\) increases.
			
			Panels (c)--(d) in Figure~\ref{fig:additional_baselines} report the corresponding deterministic-function comparison. In this setting, the fixed-radius tapered-parent method gives the smallest test MSE, illustrating that a larger fixed taper radius can preserve more of the parent covariance structure and improve accuracy. At \(n=1600\), the tapered parent gives MSE \(0.00123\), compared with \(0.00259\) for the dense parent and \(0.00314\) for the direct Wendland predictor. However, this accuracy gain comes with a denser matrix and higher solver time: the tapered-parent method takes \(51.6\) ms, whereas the direct Wendland method takes only \(9.29\) ms because its scheduled support radius yields about \(88.0\%\) zero entries. Local GP has accuracy close to the dense parent, with MSE \(0.00273\) and computation time \(41.6\) ms. Nystr\"om is again the fastest method, but it pays a noticeable accuracy cost, especially at smaller \(n\).
			
			Overall, Figure~\ref{fig:additional_baselines} shows that the proposed Wendland method is not intended to dominate every scalable approximation in every finite-sample regime. Rather, it provides a simple global sparse-kernel construction with competitive prediction accuracy and substantially reduced global-system solver time. Covariance tapering can be more accurate when a relatively large fixed taper radius is used, but then the matrix is less sparse. Local GP and Nystr\"om approximations provide useful complementary baselines, but they rely on different approximation mechanisms and timing conventions. These results therefore support the main message of the paper: the shrinking-support Wendland construction offers an effective kernel-level route to sparse linear algebra while maintaining reasonably accurate prediction.


			\begin{figure}[t!]
				\centering
				\subfigure[Regular GP truth: mean test MSE.\label{fig:additional_gp_baselines_mse}]{%
					\includegraphics[width=0.48\linewidth]{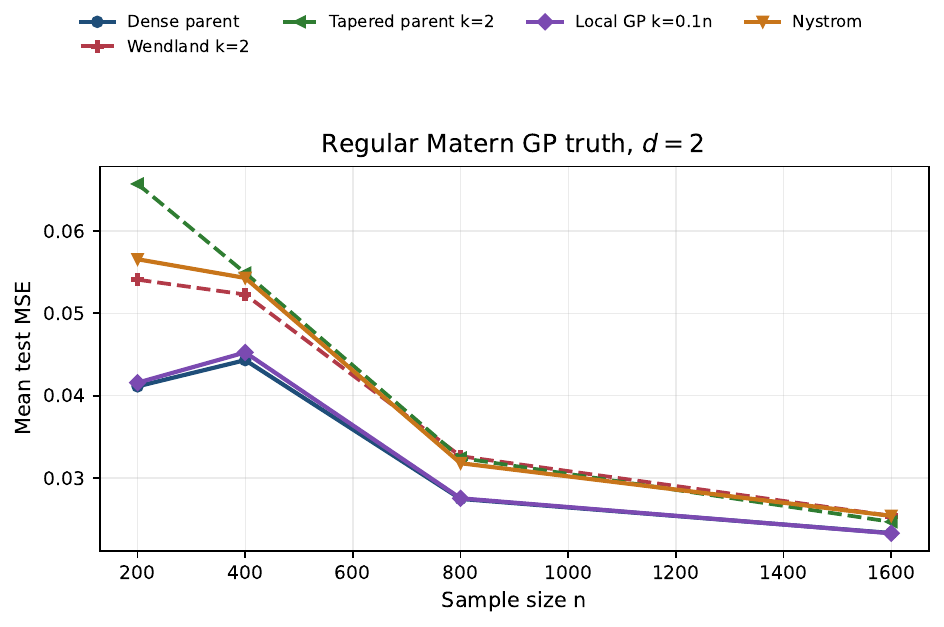}%
				}\hfill
				\subfigure[Regular GP truth: computation time.\label{fig:additional_gp_baselines_time}]{%
					\includegraphics[width=0.48\linewidth]{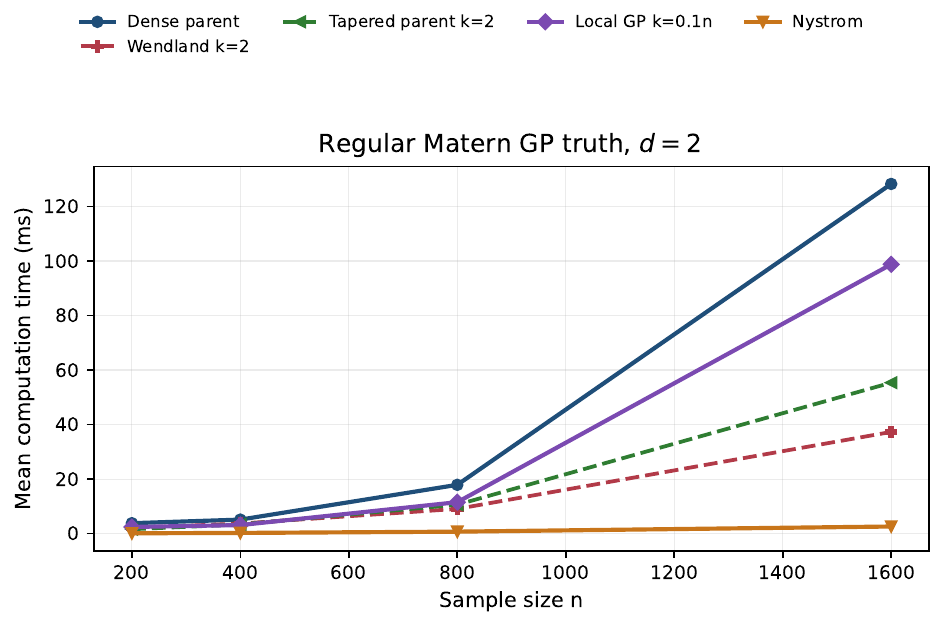}%
				}\\[0.75em]
				\subfigure[Deterministic truth: mean test MSE.\label{fig:additional_det_baselines_mse}]{%
					\includegraphics[width=0.48\linewidth]{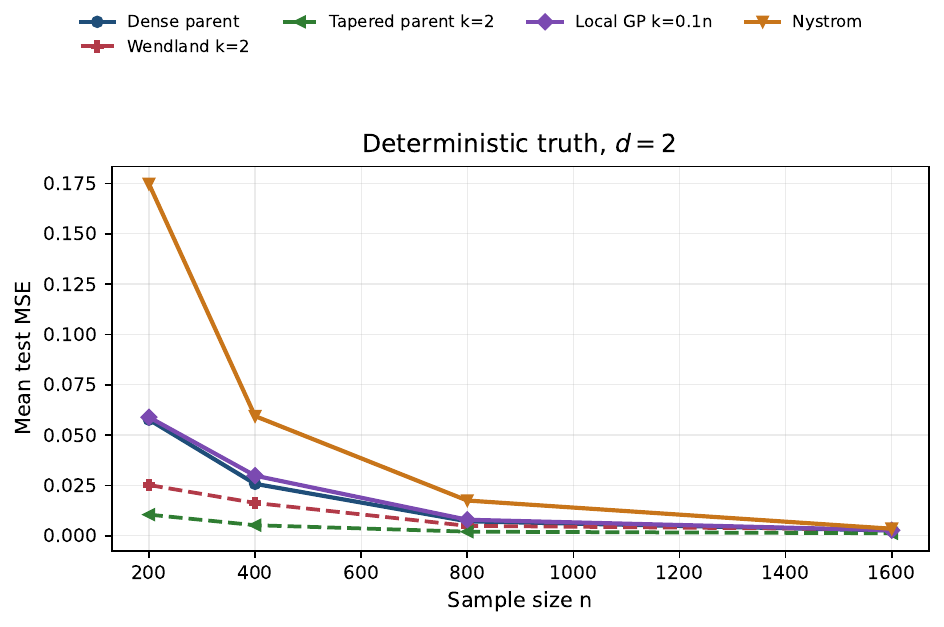}%
				}\hfill
				\subfigure[Deterministic truth: computation time.\label{fig:additional_det_baselines_time}]{%
					\includegraphics[width=0.48\linewidth]{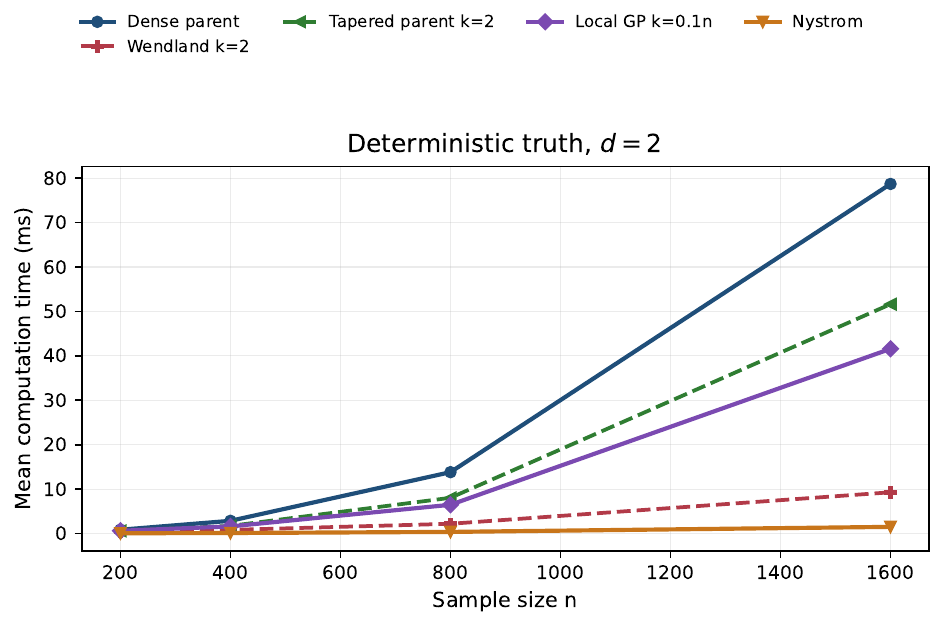}%
				}
				\caption{Additional scalable-baseline comparisons for $d=2$ as functions of sample size $n$. Panels (a) and (b) show mean test MSE and computation time, respectively, for the regular Mat\'{e}rn GP truth; panels (c) and (d) show the corresponding quantities for the deterministic truth. Curves compare the dense Mat\'{e}rn, the scheduled scale parameter for Wendland predictor with $m-(d+1)/2=2$, covariance tapering, local GP, and a Nystr\"om approximation. Results are averaged over five independent replications. Timing denotes linear-solver time for the three global systems and natural evaluation time for local GP and Nystr\"om.}
				\label{fig:additional_baselines}
			\end{figure}
			

			\section{Application to ERA5 Temperature Data}\label{sec_real}\label{sec:application}
			
			We validate the compactly supported Wendland approach on ERA5 two-meter temperature data. This real-data analysis is intended as an operational illustration of the accuracy--cost properties. The detailed data description, preprocessing, train/test construction, solver convention, and real-data tuning choices are reported in Supplementary Section~\ref{app:experimental_settings}. The dense Mat\'{e}rn baseline is reported through $n=5000$ only because repeated dense solves over all time points are computationally expensive; this truncation is a computational cap, not a statistical stopping rule.
			
			\begin{figure}[t!]
				\centering
				\includegraphics[width=0.8\linewidth]{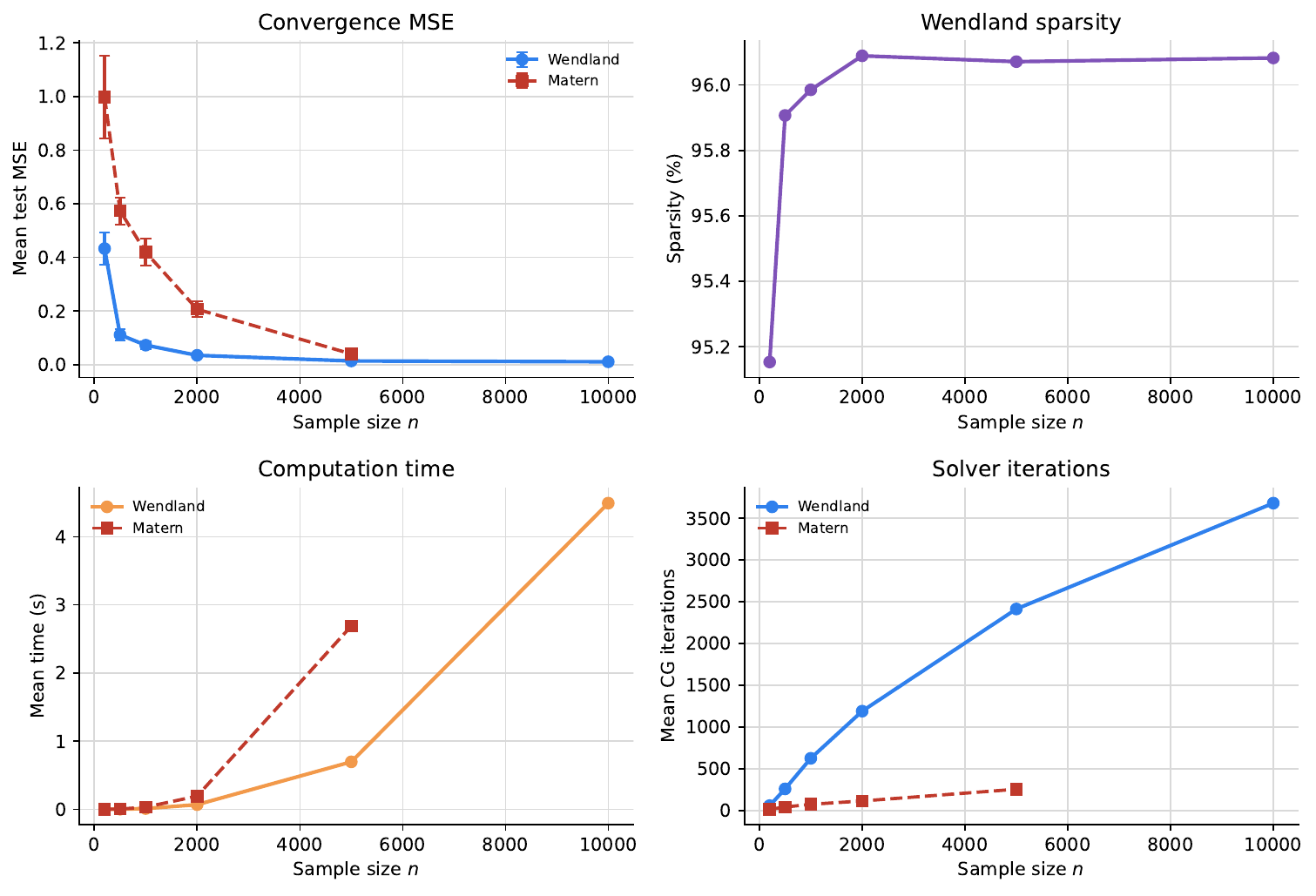}
				\caption{ERA5 two-meter temperature results as functions of sample size $n$. The upper-left panel compares MSE for the Wendland and Mat\'{e}rn predictors; the upper-right panel shows the percentage of zero entries in the Wendland covariance matrix; the lower-left panel compares mean computation time; and the lower-right panel compares mean conjugate-gradient (CG) iteration counts. Temperatures are standardized before fitting. }
				\label{fig:real_mean_metrics}
			\end{figure}
			
			\begin{figure}[h!]
				\centering
				\includegraphics[width=0.8\linewidth]{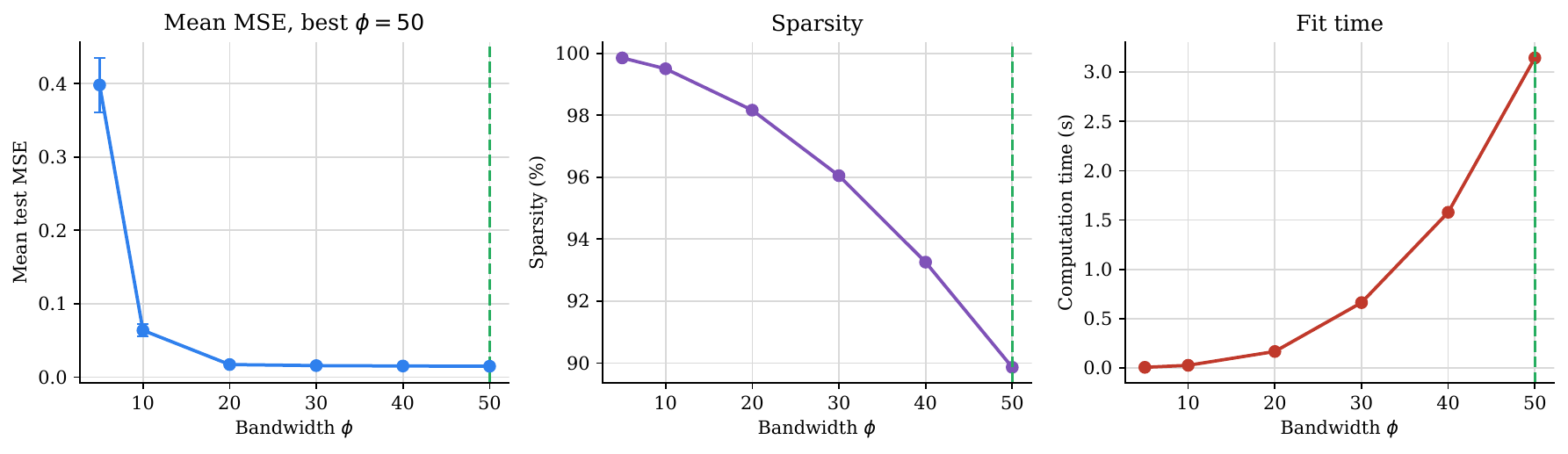}
				\caption{Bandwidth sensitivity for the ERA5 temperature experiment. The panels show mean held-out MSE, sparsity, and computation time as functions of the Wendland support radius $\phi$. MSE error bars show one standard deviation across time points.}
				\label{fig:real_bandwidth}
			\end{figure}

			\begin{figure}[h!]
				\centering
				\includegraphics[width=\linewidth]{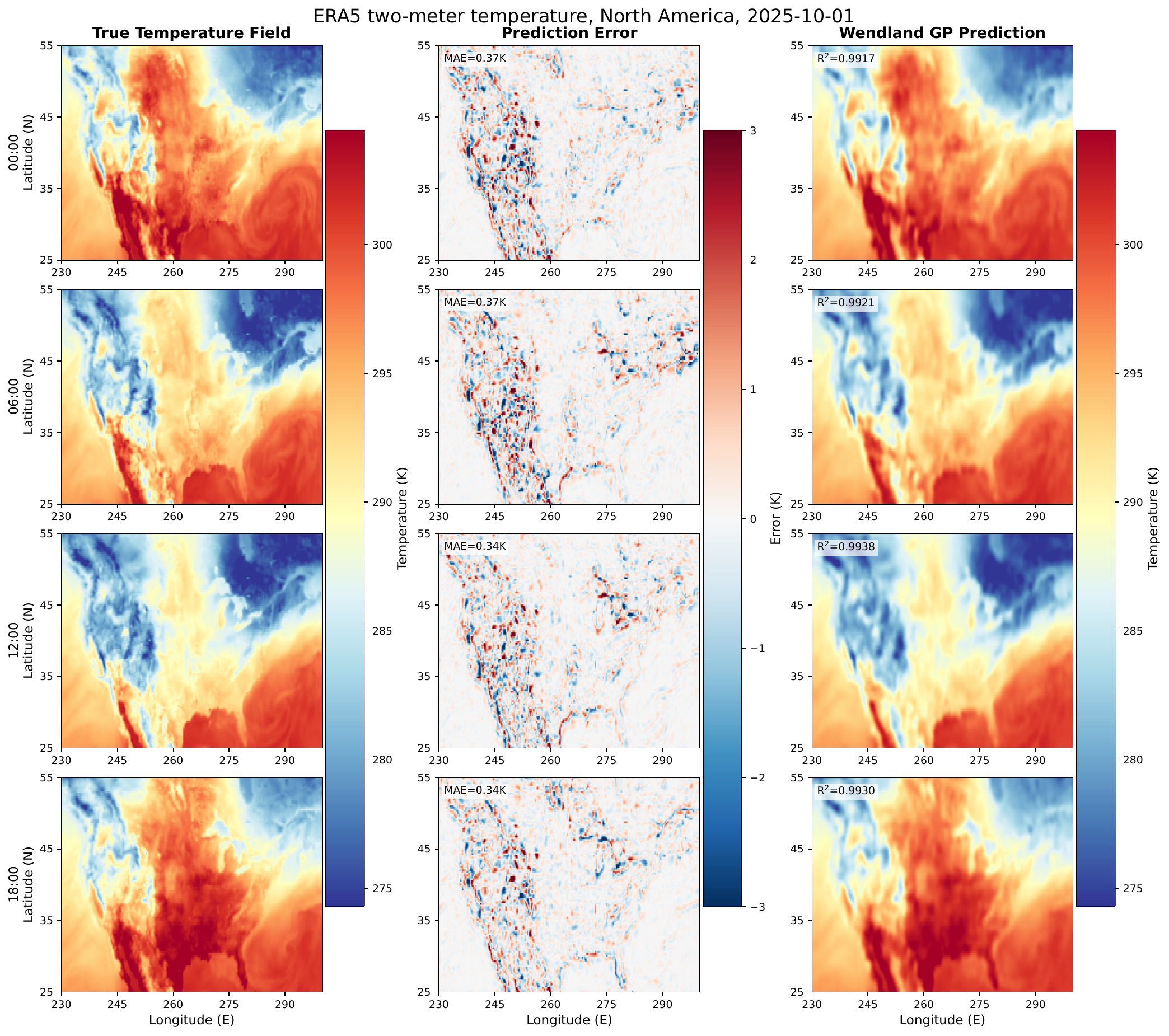}
				\caption{North America regional visualization over $230^\circ$E--$300^\circ$E and $25^\circ$N--$55^\circ$N at four time points on 2025-10-01. Rows correspond to 00:00, 06:00, 12:00, and 18:00, and columns show the true two-meter temperature field, Wendland prediction error, and Wendland prediction in Kelvin. Prediction-error panels are annotated with MAE in Kelvin.}
				\label{fig:real_north_america_four_times}
			\end{figure}
			
			The real-data experiment further empirically confirms a similar phenomena observed in the controlled simulations. At the chosen point, the mean MSE decreases from $0.4327$ at $n=200$ to $0.0106$ at $n=10000$. For the completeness of this experiment, we also show the performance under different values of $\phi_n$ at $n=5000$. It can be seen that the rougher kernel already achieves low error at moderate support: at $n=5000$, increasing $\phi$ from $20$ to $30$ changes the mean MSE from $0.0172$ to $0.0158$, while increasing $\phi$ further to $50$ only improves the mean MSE to $0.0149$ but raises average computation time from $0.6625$ seconds to $3.143$ seconds. Thus, enlarging the Wendland support radius improves accuracy, but the marginal gain beyond moderate support is small compared with the additional cost. The North America visualization in Figure~\ref{fig:real_north_america_four_times} gives a mean $R^2$ of $0.9926$ and a mean absolute error (MAE) of $0.354$ Kelvin across the four same-day snapshots, showing that the regional reconstruction remains stable under intraday temperature changes.

			\section{Conclusions and Discussion}\label{sec_conclu}
			
			In this work, we propose a scale-adjusted compactly supported approach for scalable Gaussian process prediction. By using generalized Wendland correlation functions with a support radius $\phi_n$ that decreases with the sample size, the proposed method induces sparsity in the correlation matrix and reduces sparse matrix-vector cost. Theoretical results show that, with suitable choices of $\phi_n$ and the regularization parameter, the resulting predictor preserves the optimal $L_2(\Omega)$ convergence rate in Gaussian process regression. The same idea also applies to kernel ridge regression, where the estimator attains the minimax-optimal rate while benefiting from sparse matrix-vector operations. Numerical experiments and real case analysis support the theoretical findings, and show that the scale-adjusted Wendland method achieves prediction accuracy comparable to or better than Mat\'ern-kernel baselines while often requiring less computation time.
			
			Our results also highlight several directions for future work. First, it would be useful to develop uncertainty quantification theory for scale-adjusted compactly supported predictors, including the validity of predictive intervals under covariance misspecification. 
			Second, extending the theory to nonstationary, anisotropic, random-design, and high-dimensional settings would further broaden the applicability of scale-adjusted compactly supported Gaussian process methods. Finally, extending the framework to non-Gaussian and generalized-response models, such as binary, count, or heavy-tailed spatial data, would broaden its applicability beyond squared-error regression.

\bibliography{fast_computation_GP/ref}
\bibliographystyle{apalike}

		\appendix
		
		\newpage
		\setcounter{page}{1}
		
		\numberwithin{figure}{section}
		\numberwithin{table}{section}
		\numberwithin{example}{section}

			\def\spacingset#1{\renewcommand{\baselinestretch}%
				{#1}\small\normalsize}
			
\spacingset{1.8}
			\section*{Supplementary Material for ``Rate-Preserving Shrinking-Support Gaussian Process Prediction''}
			
			In this supplement, we provide more introduction to reproducing kernel Hilbert spaces (RKHSs), details of technical proofs, and more experimental details and results.
			

			\section{Reproducing Kernel Hilbert Spaces}\label{app:introtoSoboRKHS}
			
			We start with introducing the RKHS. Let $\cX$ be a nonempty set and $k : \cX \times \cX \to \mathbb{R}$ be a positive definite kernel. The following definition of RKHS is from \cite{kanagawa2018gaussian}.
			
			\begin{definition}[RKHS]\label{def:rkhs}
				A Hilbert space $\cN_k$ of functions $f : \cX \to \RR$ is called a \emph{reproducing kernel Hilbert space} (RKHS) with reproducing kernel $k$ if the following properties hold:
				\begin{enumerate}
					\item For all $\bx \in \cX$, the function $k(\cdot, \bx)$ belongs to $\cN_k$;
					\item For all $\bx \in \cX$ and $f \in \cN_k$, the reproducing property $f(x) = \langle f, k(\cdot, \bx) \rangle_{\cN_k}$ holds.
				\end{enumerate}
			\end{definition}
			
			The following theorem provides a characterization of the RKHS when $K$ is defined by a stationary kernel function $\Phi$, via the Fourier transform of $\Phi$. 
			\begin{theorem}[Theorem 10.12 of Wendland (2004)]\label{thm:NativeSpace}
				Let $\Phi$ be a positive definite kernel function which is continuous and integrable in $\RR^d$. Define 
				$$\mathcal{G}:=\{f\in L_2(\RR^d)\cap C(\RR^d):\cF(f)/\sqrt{\cF(\Phi)}\in L_2(\RR^d)\},$$
				with the inner product
				$$\langle f,g\rangle_{\mathcal{N}_\Phi(\RR^d)}=(2\pi)^{-d}\int_{\RR^d}\frac{\cF(f)(\bm \omega)\overline{\cF(g)(\bomega)}}{\cF(\Phi)(\bm \omega)}d \bm \omega.$$ Then $\mathcal{G} = \mathcal{N}_\Phi(\RR^d)$, and both inner products coincide.
			\end{theorem}
			
			The Sobolev norm for functions defined on the entire space is given by 
			\begin{eqnarray}\label{Sobolev}
				\|f\|_{H^m(\mathbb{R}^d)}=\left(\int_{\mathbb{R}^d} |\mathcal{F} (f)(\bomega)|^2 (1+\|\bomega\|_2^2)^{m} {\rm d} \bomega\right)^{1/2}.
			\end{eqnarray}
			The Sobolev space $H^m(\RR^d)$ consists of functions for which this norm is finite. The Sobolev space defined on a subset $\Omega$ can be derived through restriction.
			
			If $\Psi$ is a Mat\'ern correlation function with Sobolev smoothness $m$, as described in \eqref{materngai}, its Fourier transform is given by
			\begin{align*}
				\mathcal{F}(\Psi)(\omega;m,\phi)= 4^m\pi^{\frac{d}{2}}\frac{\Gamma(m)}{\Gamma(m-\frac{d}{2})}\left((m-\frac{d}{2})\phi^2\right)^{m-\frac{d}{2}} \left(4(m-\frac{d}{2})\phi^2+\|\omega\|_2^2\right)^{-m}.
			\end{align*}
			From this expression, it can be seen that $\mathcal{F}(\Psi)$ satisfies the bounds 
			\begin{align*}
				c_1(1+\|\bomega\|_2^2)^{-m} \leq  \mathcal{F}(\Psi_m)(\bomega) \leq c_2(1+\|\bomega\|_2^2)^{-m}
			\end{align*}
			for two positive constants $c_1$ and $c_2$. From Theorem \ref{thm:NativeSpace} and the Sobolev norm defined in \eqref{Sobolev}, it follows that the RKHS $\mathcal{N}_{\Psi}(\RR^d)$ generated by $\Psi$ is equivalent to $H^m(\RR^d)$. By applying restriction techniques (see \cite{adams2003sobolev} and \cite{wendland2004scattered}), it can be further demonstrated that $\mathcal{N}_{\Psi_m}(\Omega) = H^m(\Omega)$ with equivalent norms. For further details, we refer readers to \cite{wendland2004scattered}.
			
			\section{Proof of Theorem \ref{thm_GPconvergence_rate}}\label{app_pfofGPConvergence}
			
			We first present several lemmas used in this proof. 
			\begin{lemma}\label{LEM:BOFKMMAX}
				Suppose the design points $\bX=\{\bx_1,...,\bx_n\}$ and the separation radius of $\bX$ $q_{\bX}\lesssim 1$. Let $\Psi$ be a correlation function satisfy Assumption \ref{assum_Psi}, and $\Lambda_{\bX}$ be the maximum eigenvalue of matrix $(\Psi(\bx_j-\bx_k))_{jk}$. Then 
				\begin{align*}
					\Lambda_{\bX}\leq Cq_{\bX}^{-d},
				\end{align*}
				where $C$ is a constant depending on $\Psi$ and $\Omega$.
			\end{lemma}
			
			\begin{lemma}\label{LEMMAERRORWNUG}
				Let \(h_n=h_{\bX_n,\Omega}\) and \(\Phi=\Phi_{\phi_n}\) be the generalized Wendland correlation with smoothness \(m>d/2\) and \(0<\phi_n\le 1\). Under Assumptions \ref{assum_region}, \ref{assum_Psi}, and \ref{assum_quasiuniform}, we have, for any \(\rho_n>0\) and all \(\bx\in\Omega\),
				\[
				\Phi(\bx-\bx) - \rb_\Phi(\bx)^\top(\Rb_\Phi+\rho_n \Ib_n)^{-1}\rb_\Phi(\bx)\le C\phi_n^{-\frac{d(2m-d)}{2m}}\left(h_n^{2m}\phi_n^{-2m+d}+\frac{\rho_n}{n}\right)^{1-\frac{d}{2m}},
				\]
				where \(C>0\) does not depend on \(n,\phi_n,\rho_n\), and \(\bx\).
			\end{lemma}
			
			The following lemma states the two-frequency-regime Fourier transform bounds for the scaled generalized Wendland correlation. 
			
			\begin{lemma}\label{lem:wendland_fourier_scaled}
				Let \(\Phi_\phi\) be the generalized Wendland correlation in \eqref{eq_GWcorr} with imposed smoothness \(m\), and support radius \(\phi\in(0,1]\). There exist constants \(c_1,c_2,c_3,c_4,w_0>0\), independent of \(\phi\), such that
				\begin{align}
					c_3\phi^d \leq  \cF(\Phi_\phi)(\bomega) \leq c_4\phi^d,
					\qquad \|\bomega\|_2 \leq \phi^{-1}w_0,\label{eq_FtransPhi2}
				\end{align}
				and
				\begin{align}
					c_1\phi^{-2m+d}\|\bomega\|_2^{-2m} \leq  \cF(\Phi_\phi)(\bomega) \leq c_2\phi^{-2m+d}\|\bomega\|_2^{-2m},
					\qquad \|\bomega\|_2 > \phi^{-1}w_0.\label{eq_FtransPhi1}
				\end{align}
			\end{lemma}
			
			\noindent\textit{Proof.}
			Let \(\Phi_1\) denote the corresponding unit-support Wendland correlation. Since \(\Phi_\phi(\bx)=\Phi_1(\bx/\phi)\), the scaling property of the Fourier transform gives
			\[
			\cF(\Phi_\phi)(\bomega)=\phi^d\cF(\Phi_1)(\phi\bomega).
			\]
			For the unit-support generalized Wendland correlation, Theorem 1 of \citet{bevilacqua2019estimation}, equivalently the hypergeometric expansion in \citet{chernih2014closed}, gives
			\[
			\cF(\Phi_1)(\bomega)\asymp \|\bomega\|_2^{-2m},\qquad \|\bomega\|_2\rightarrow\infty.
			\]
			Therefore, after choosing \(w_0\) sufficiently large, there are constants \(a_1,a_2>0\) such that
			\[
			a_1\|\bomega\|_2^{-2m}\leq \cF(\Phi_1)(\bomega)\leq a_2\|\bomega\|_2^{-2m},\qquad \|\bomega\|_2>w_0.
			\]
			The same theorem implies that \(\cF(\Phi_1)\) is strictly positive and continuous. Hence, on the compact set \(\{\bomega:\|\bomega\|_2\leq w_0\}\), there are constants \(a_3,a_4>0\) such that
			\[
			a_3\leq \cF(\Phi_1)(\bomega)\leq a_4.
			\]
			Substituting \(\phi\bomega\) for \(\bomega\) in these unit-scale bounds and multiplying by \(\phi^d\) yields \eqref{eq_FtransPhi2} and \eqref{eq_FtransPhi1}.

			Now we begin to prove Theorem \ref{thm_GPconvergence_rate}. The proof follows the line of the proof of Theorem 6 in \cite{wang2022gaussian}. 
			We can decompose $\hat f_\Phi(\bx)$ in \eqref{eq:GPpredic} as
			\begin{align*}
				\hat f_\Phi(\bx) =\rb_\Phi(\bx)^\top (\Rb_\Phi+\rho_n \Ib_n)^{-1} \bZ + \rb_\Phi(\bx)^\top (\Rb_\Phi+\rho_n \Ib_n)^{-1} \bvarepsilon,
			\end{align*}
			where $\bvarepsilon = (\varepsilon_1,...,\varepsilon_n)^\top$ and $\bZ = (Z(\bx_1),...,Z(\bx_n))^\top$ (recall that we do not differentiate $Z$ and $f$ in Sections \ref{sec_GP} and \ref{sec_GPtheory}). Then the MSPE can be bounded by
			\begin{align}\label{eq_gpoverde}
				\EE(Z(\bx) - \hat f_{\Phi}(\bx))^2 = & \EE(Z(\bx) - \rb_\Phi(\bx)^\top (\Rb_\Phi+\rho_n \Ib_n)^{-1} \bZ - \rb_\Phi(\bx)^\top (\Rb_\Phi+\rho_n \Ib_n)^{-1} \bvarepsilon)^2\nonumber\\
				\leq & 2\EE(Z(\bx) - \rb_\Phi(\bx)^\top (\Rb_\Phi+\rho_n \Ib_n)^{-1} \bZ)^2 + 2\EE(\rb_\Phi(\bx)^\top (\Rb_\Phi+\rho_n \Ib_n)^{-1} \bvarepsilon)^2\nonumber\\
				= & 2I_1(\bx) +2I_2(\bx),
			\end{align}
			where the inequality is by the Cauchy-Schwarz inequality. By computing the expectation in the terms $I_1(\bx)$ and $I_2(\bx)$ in \eqref{eq_gpoverde}, it holds that
			\begin{align}\label{eq_gpoverde_I1}
				I_1(\bx) = & \Psi(\bx-\bx) - 2 \rb_\Phi(\bx)^\top (\Rb_\Phi+\rho_n \Ib_n)^{-1} \rb(\bx)\nonumber\\
				& + \rb_\Phi(\bx)^\top (\Rb_\Phi+\rho_n \Ib_n)^{-1} \Rb (\Rb_\Phi+\rho_n \Ib_n)^{-1}\rb_\Phi(\bx),
			\end{align}
			and
			\begin{align}\label{eq_gpoverde_I2}
				I_2(\bx) = & \sigma_\varepsilon^2\rb_\Phi(\bx)^\top (\Rb_\Phi+\rho_n \Ib_n)^{-2}\rb_\Phi(\bx),
			\end{align}
			where $\rb(\bx)$ and $\Rb$ are as in \eqref{mean}.
			
			For the scaled generalized Wendland working correlation \(\Phi=\Phi_{\phi_n}\), Lemma \ref{lem:wendland_fourier_scaled} gives the low- and high-frequency bounds \eqref{eq_FtransPhi2} and \eqref{eq_FtransPhi1}, with constants independent of \(n\) and \(\phi_n\). Together with Theorem \ref{thm:NativeSpace}, these bounds imply the usual equivalence between \(H^m(\Omega)\) and \(\cN_\Phi(\Omega)\), with the scale dependence made explicit in Lemma \ref{lem:wendland_fourier_scaled}.
			By the Fourier inversion theorem, for $\bu=(u_1,...,u_n)^T = (\Rb_\Phi+\rho_n \Ib_n)^{-1}\rb_\Phi(\bx)$, we have 
			\begin{align}\label{eq_eI1gpover}
				& \mathbb{E}I_1(\bx) =  \int_{\RR^d}\left|\sum_{j=1}^nu_je^{-\ib \bx_j^\top \bomega}-e^{-\ib \bx^\top\bomega}\right|^2\mathcal{F}(\Psi)(\bomega){\rm d}\bomega\nonumber\\
				= & \int_{\|\bomega\|_2 \leq \phi_n^{-1}w_0}\left|\sum_{j=1}^nu_je^{-\ib \bx_j^\top \bomega}-e^{-\ib \bx^\top\bomega}\right|^2\mathcal{F}(\Psi)(\bomega){\rm d}\bomega\nonumber\\
				& + \int_{\phi_n^{-1}w_0 < \|\bomega\|_2\leq \phi_n^{-1}\gamma }\left|\sum_{j=1}^nu_je^{-\ib \bx_j^\top \bomega}-e^{-\ib \bx^\top\bomega}\right|^2\mathcal{F}(\Psi)(\bomega){\rm d}\bomega\nonumber\\
				& + \int_{\|\omega\|_2> \phi_n^{-1}\gamma}\left|\sum_{j=1}^nu_je^{-\ib \bx_j^\top \bomega}-e^{-\ib \bx^\top\bomega}\right|^2\mathcal{F}(\Psi)(\bomega){\rm d}\bomega\nonumber\\
				= & I_{11} + I_{12}+ I_{13},
			\end{align}
			where $\gamma>w_0$ will be determined later. 
			
			The first term $I_{11}$ in \eqref{eq_eI1gpover} can be bounded by
			\begin{align}\label{eq_I11GPover}
				I_{11} \lesssim &  \int_{\|\bomega\|_2 \leq \phi_n^{-1}w_0}\left|\sum_{j=1}^nu_je^{-\ib \bx_j^\top \bomega}-e^{-\ib \bx^\top\bomega}\right|^2(1+\|\bomega\|_2^2)^{-m_0}{\rm d}\bomega \nonumber\\
				\lesssim &  \int_{\|\bomega\|_2 \leq \phi_n^{-1}w_0}\left|\sum_{j=1}^nu_je^{-\ib \bx_j^\top \bomega}-e^{-\ib \bx^\top\bomega}\right|^2{\rm d}\bomega \nonumber\\
				\lesssim & \phi_n^{-d}\int_{\|\omega\|_2\leq \phi_n^{-1}w_0 }\left|\sum_{j=1}^nu_je^{-\ib \bx_j^\top \bomega}-e^{-\ib \bx^\top\bomega}\right|^2\mathcal{F}(\Phi)(\bomega){\rm d}\bomega,
			\end{align}
			where the first inequality is because of Assumption \ref{assum_Psi}, and the last inequality is by \eqref{eq_FtransPhi2}.
			
			The second term $I_{12}$ in \eqref{eq_eI1gpover} can be bounded by
			\begin{align}\label{eq_I12GPover}
				I_{12} \lesssim &  \int_{\phi_n^{-1}w_0 < \|\bomega\|_2\leq \phi_n^{-1}\gamma }\left|\sum_{j=1}^nu_je^{-\ib \bx_j^\top \bomega}-e^{-\ib \bx^\top\bomega}\right|^2(1+\|\bomega\|_2^2)^{-m_0}{\rm d}\bomega \nonumber\\
				\lesssim  & (\phi_n^{-1}\gamma)^{2m-2m_0}\int_{\phi_n^{-1}w_0 < \|\bomega\|_2\leq \phi_n^{-1}\gamma }\left|\sum_{j=1}^nu_je^{-\ib \bx_j^\top \bomega}-e^{-\ib \bx^\top\bomega}\right|^2(1+\|\bomega\|_2^2)^{-m}{\rm d}\bomega \nonumber\\
				\lesssim & \gamma^{2m-2m_0}\phi_n^{2m_0-d}\int_{\phi_n^{-1}w_0 < \|\bomega\|_2\leq \phi_n^{-1}\gamma }\left|\sum_{j=1}^nu_je^{-\ib \bx_j^\top \bomega}-e^{-\ib \bx^\top\bomega}\right|^2\mathcal{F}(\Phi)(\bomega){\rm d}\bomega,
			\end{align}
			where the first inequality is because of Assumption \ref{assum_Psi}, and the last inequality is by \eqref{eq_FtransPhi1}.
			
			The third term $I_{13}$ in \eqref{eq_eI1gpover} can be bounded similarly as shown in the proof of Theorem 6 in \cite{wang2022gaussian}. Specifically, see equations (35)--(38) of \cite{wang2022gaussian}. As a result,
			\begin{align}\label{eq_I13GPover}
				I_{13} \lesssim & (\phi_n^{-1}\gamma)^{d-2m_0}\Lambda_{\phi_n^{-1}\gamma \bX}\|\bu\|_2^2+  (\phi_n^{-1}\gamma)^{-2m_0+d},
			\end{align}
			where $\phi_n^{-1}\gamma \bX = \{\phi_n^{-1}\gamma \bx_1,...,\phi_n^{-1}\gamma \bx_n\}$. The separation distance of $\phi_n^{-1}\gamma \bX$ is $\phi_n^{-1}\gamma q_{\bX}$, and thus Lemma \ref{LEM:BOFKMMAX} implies
			\begin{align}\label{eq_gp_UbofLambda1}
				\Lambda_{\phi_n^{-1}\gamma \bX}\leq C(\phi_n^{-1}\gamma q_{\bX})^{-d},
			\end{align}
			where we further require $\phi_n^{-1}\gamma q_{\bX}\lesssim 1$, which will be checked later. Plugging \eqref{eq_I11GPover}, \eqref{eq_I12GPover} and \eqref{eq_I13GPover} into \eqref{eq_eI1gpover}, we have
			\begin{align}\label{eq_eI1gpover2}
				& \mathbb{E}I_1(\bx)\nonumber\\
				\lesssim & \phi_n^{-d}\int_{\|\omega\|_2\leq \phi_n^{-1}w_0 }\left|\sum_{j=1}^nu_je^{-\ib \bx_j^\top \bomega}-e^{-\ib \bx^\top\bomega}\right|^2\mathcal{F}(\Phi)(\bomega){\rm d}\bomega \nonumber\\
				& +\gamma^{2m-2m_0}\phi_n^{2m_0-d}\int_{\phi_n^{-1}w_0 < \|\bomega\|_2\leq \phi_n^{-1}\gamma }\left|\sum_{j=1}^nu_je^{-\ib \bx_j^\top \bomega}-e^{-\ib \bx^\top\bomega}\right|^2\mathcal{F}(\Phi)(\bomega){\rm d}\bomega\nonumber\\
				& + (\phi_n^{-1}\gamma)^{d-2m_0}\Lambda_{\phi_n^{-1}\gamma \bX}\|\bu\|_2^2+  (\phi_n^{-1}\gamma)^{-2m_0+d}\nonumber\\
				\leq & (\gamma^{2m-2m_0}\phi_n^{2m_0-d} + \phi_n^{-d})\int_{\|\omega\|_2\leq \phi_n^{-1}\gamma }\left|\sum_{j=1}^nu_je^{-\ib \bx_j^\top \bomega}-e^{-\ib \bx^\top\bomega}\right|^2\mathcal{F}(\Phi)(\bomega){\rm d}\bomega\nonumber\\
				& + \phi_n^{2m_0}\gamma^{-2m_0}q_{\bX}^{-d}\|\bu\|_2^2+  (\phi_n^{-1}\gamma)^{-2m_0+d}\nonumber\\
				\leq & (\gamma^{2m-2m_0}\phi_n^{2m_0-d} + \phi_n^{-d})\left(\int_{\RR^d}\left|\sum_{j=1}^nu_je^{-\ib \bx_j^\top \bomega}-e^{-\ib \bx^\top\bomega}\right|^2\mathcal{F}(\Phi)(\bomega){\rm d}\bomega + \phi_n^d\gamma^{-2m}(q_{\bX})^{-d}\|\bu\|_2^2\right)\nonumber\\
				& + (\phi_n^{-1}\gamma)^{-2m_0+d},
			\end{align}
			where the second inequality is by \eqref{eq_gp_UbofLambda1}.
			
			Taking $\gamma = \phi_n^{\frac{d}{2m}}\rho_n^{-\frac{1}{2m}} (q_{\bX})^{-\frac{d}{2m}}$ such that $\phi_n^d\gamma^{-2m}(q_{\bX})^{-d} = \rho_n$ and $\phi_n^{-1}\gamma q_{\bX}\lesssim 1$. Then, following \eqref{eq_eI1gpover2}, $\mathbb{E}I_1(\bx)$ can be further bounded by
			\begin{align}\label{eq_eI1gpover3}
				& \mathbb{E}I_1(\bx)\nonumber\\
				\lesssim & (\phi_n^{\frac{m_0(2m-d)}{m}}\rho_n^{-\frac{m-m_0}{m}}q_{\bX}^{-\frac{(m-m_0)d}{m}} + \phi_n^{-d})\left(\int_{\RR^d}\left|\sum_{j=1}^nu_je^{-\ib \bx_j^\top \bomega}-e^{-\ib \bx^\top\bomega}\right|^2\mathcal{F}(\Phi)(\bomega){\rm d}\bomega + \rho_n\|\bu\|_2^2\right)\nonumber\\
				& + \phi_n^{\frac{(2m-d)(2m_0-d)}{2m}}\rho_n^{\frac{2m_0-d}{2m}} (q_{\bX})^{\frac{(2m_0-d)d}{2m}}\nonumber\\
				= &  (\phi_n^{\frac{m_0(2m-d)}{m}}\rho_n^{-\frac{m-m_0}{m}}q_{\bX}^{-\frac{(m-m_0)d}{m}} + \phi_n^{-d})\nonumber\\
				& \times\bigg(\Phi(\bx-\bx) -2\sum_{j=1}^n u_j\Phi(\bx-\bx_j) + \sum_{j,k=1}^n u_ju_k\Phi(\bx_j-\bx_k) + \rho_n\|\bu\|_2^2\bigg)
				\nonumber\\
				& + \phi_n^{\frac{(2m-d)(2m_0-d)}{2m}}\rho_n^{\frac{2m_0-d}{2m}} (q_{\bX})^{\frac{(2m_0-d)d}{2m}}\nonumber\\
				=  &  (\phi_n^{\frac{m_0(2m-d)}{m}}\rho_n^{-\frac{m-m_0}{m}}q_{\bX}^{-\frac{(m-m_0)d}{m}} + \phi_n^{-d})(\Phi(\bx-\bx) - \rb_\Phi(\bx)^T(\Rb_\Phi+\rho_n \Ib_n)^{-1}\rb_\Phi(\bx))\nonumber\\
				& + \phi_n^{\frac{(2m-d)(2m_0-d)}{2m}}\rho_n^{\frac{2m_0-d}{2m}} (q_{\bX})^{\frac{(2m_0-d)d}{2m}}\nonumber\\
				\leq & (\phi_n^{\frac{m_0(2m-d)}{m}}\rho_n^{-\frac{m-m_0}{m}}q_{\bX}^{-\frac{(m-m_0)d}{m}} + \phi_n^{-d})\left(\phi_n^{-\frac{d(2m-d)}{2m}}\left(h_n^{2m}\phi_n^{-2m+d} + \frac{\rho_n}{n}\right)^{1-\frac{d}{2m}}\right)\nonumber\\
				& + \phi_n^{\frac{(2m-d)(2m_0-d)}{2m}}\rho_n^{\frac{2m_0-d}{2m}} (q_{\bX})^{\frac{(2m_0-d)d}{2m}}\nonumber\\
				\lesssim & (\phi_n^{\frac{m_0(2m-d)}{m}}\left(\frac{\rho_n}{n}\right)^{\frac{m_0-m}{m}} + \phi_n^{-d})\left(\phi_n^{-\frac{d(2m-d)}{2m}}\left(n^{-\frac{2m}{d}}\phi_n^{-2m+d} + \frac{\rho_n}{n}\right)^{1-\frac{d}{2m}}\right)\nonumber\\
				& + \phi_n^{\frac{(2m-d)(2m_0-d)}{2m}}\left(\frac{\rho_n}{n}\right)^{\frac{2m_0 - d}{2m}},
			\end{align}
			where the second inequality is by Lemma \ref{LEMMAERRORWNUG}, and the last inequality is because \(q_{\bX}\asymp n^{-1/d}\) by Assumption \ref{assum_quasiuniform}.
			
			Next, we consider $I_2(\bx)$. In fact, 
			\begin{align}\label{eq_eI2gpover1}
				I_2(\bx) = & \sigma_\varepsilon^2\|\bu\|_2^2 \leq \frac{\sigma_\varepsilon^2}{\rho_n}\left(\int_{\RR^d}\left|\sum_{j=1}^nu_je^{-\ib \bx_j^\top \bomega}-e^{-\ib \bx^\top\bomega}\right|^2\mathcal{F}(\Phi)(\bomega){\rm d}\bomega + \rho_n\|\bu\|_2^2\right)\nonumber\\
				= & \frac{\sigma_\varepsilon^2}{\rho_n}(\Phi(\bx-\bx) - \rb_\Phi(\bx)^T(\Rb_\Phi+\rho_n \Ib_n)^{-1}\rb_\Phi(\bx))\nonumber\\
				\leq & C_9 \rho_n^{-1}\phi_n^{-\frac{d(2m-d)}{2m}}\left(h_n^{2m}\phi_n^{-2m+d} + \frac{\rho_n}{n}\right)^{1-\frac{d}{2m}},
			\end{align}
			where the last inequality is by Lemma \ref{LEMMAERRORWNUG}. Plugging \eqref{eq_eI1gpover3} and \eqref{eq_eI2gpover1} into \eqref{eq_gpoverde} yields
			\begin{align}\label{eq_gpoverall}
				& \EE(Z(\bx) - \hat f_{\Phi}(\bx))^2\nonumber\\
				\leq & C\bigg((\phi_n^{\frac{m_0(2m-d)}{m}}\left(\frac{\rho_n}{n}\right)^{\frac{m_0-m}{m}} + \phi_n^{-d} +\rho_n^{-1})\left(\phi_n^{-\frac{d(2m-d)}{2m}}\left(n^{-\frac{2m}{d}}\phi_n^{-2m+d} + \frac{\rho_n}{n}\right)^{1-\frac{d}{2m}}\right)\nonumber\\
				& + \phi_n^{\frac{(2m-d)(2m_0-d)}{2m}}\left(\frac{\rho_n}{n}\right)^{\frac{2m_0 - d}{2m}}\bigg).
			\end{align}
			Finally, since the bound in \ref{eq_gpoverall} is uniform in \(\bx\in\Omega\), we have
			\[
			E\|Z-\hat f_\Phi\|_{L^2(\Omega)}^2=\int_\Omega E\{Z(x)-\hat f_\Phi(x)\}^2\,dx\le|\Omega|\sup_{x\in\Omega}E\{Z(x)-\hat f_\Phi(x)\}^2,
			\]
			which finishes the proof of Theorem \ref{thm_GPconvergence_rate}.

			\section{Proof of Proposition \ref{prop_GP_phimunrelation}}
			
			Direct computation shows that 
			\begin{align}\label{eq_pf_prop_gp_phimurela}
				& (\phi_n^{\frac{m_0(2m-d)}{m}}\left(\frac{\rho_n}{n}\right)^{\frac{m_0-m}{m}} + \phi_n^{-d} +\rho_n^{-1})\left(\phi_n^{-\frac{d(2m-d)}{2m}}\left(n^{-\frac{2m}{d}}\phi_n^{-2m+d} + \frac{\rho_n}{n}\right)^{1-\frac{d}{2m}}\right)\nonumber\\
				& + \phi_n^{\frac{(2m-d)(2m_0-d)}{2m}}\left(\frac{\rho_n}{n}\right)^{\frac{2m_0 - d}{2m}}\nonumber\\
				\geq & \phi_n^{\frac{(2m_0-d)(2m-d)}{2m}}\left(\frac{\rho_n}{n}\right)^{\frac{2m_0-d}{2m}} + \rho_n^{-1}\phi_n^{-\frac{d(2m-d)}{2m}}\left(\frac{\rho_n}{n}\right)^{1-\frac{d}{2m}}\nonumber\\
				= & \left(\rho_n\phi_n^{2m-d}\right)^{\frac{2m_0-d}{2m}}n^{-\frac{2m_0-d}{2m}} + \left(\rho_n\phi_n^{2m-d}\right)^{-\frac{d}{2m}}n^{-1+\frac{d}{2m}}.
			\end{align}
			By balancing the two terms in \eqref{eq_pf_prop_gp_phimurela}, the minimum is achieved by taking $\rho_n\phi_n^{2m-d} \asymp n^{-\frac{m-m_0}{m_0}}$.
			
			To show that this rate is achievable, note that 
			\begin{align}\label{eq_pf_prop_gp_phimurel2}
				& (\phi_n^{\frac{m_0(2m-d)}{m}}\left(\frac{\rho_n}{n}\right)^{\frac{m_0-m}{m}} + \phi_n^{-d} +\rho_n^{-1})\left(\phi_n^{-\frac{d(2m-d)}{2m}}\left(n^{-\frac{2m}{d}}\phi_n^{-2m+d} + \frac{\rho_n}{n}\right)^{1-\frac{d}{2m}}\right)\nonumber\\
				& + \phi_n^{\frac{(2m-d)(2m_0-d)}{2m}}\left(\frac{\rho_n}{n}\right)^{\frac{2m_0 - d}{2m}}\nonumber\\
				\leq & 4\phi_n^{\frac{(2m_0-d)(2m-d)}{2m}}\left(\frac{\rho_n}{n}\right)^{\frac{2m_0-d}{2m}} + 4\rho_n^{-1}\phi_n^{-\frac{d(2m-d)}{2m}}\left(\frac{\rho_n}{n}\right)^{1-\frac{d}{2m}}+ \phi_n^{\frac{(2m-d)(2m_0-d)}{2m}}\left(\frac{\rho_n}{n}\right)^{\frac{2m_0 - d}{2m}}\nonumber\\
				\leq & 5\left(\left(\rho_n\phi_n^{2m-d}\right)^{\frac{2m_0-d}{2m}}n^{-\frac{2m_0-d}{2m}} + \left(\rho_n\phi_n^{2m-d}\right)^{-\frac{d}{2m}}n^{-1+\frac{d}{2m}}\right),
			\end{align}
			where the first inequality is by letting $\phi_n^d \geq \rho_n$ and $n^{-\frac{2m}{d}}\phi_n^{-2m+d}\leq \frac{\rho_n}{n}$, and the second inequality is because $\phi_n \lesssim 1$. Comparing \eqref{eq_pf_prop_gp_phimurel2} with \eqref{eq_pf_prop_gp_phimurela}, we finish the proof.

			\section{Proof of Theorem \ref{thm_KRRratesX}}\label{app_pf_thm_KRRrates}
			
			We will show a general version of Theorem \ref{thm_KRRratesX} as follows.
			
			\begin{theorem}\label{thm_KRRrates}
				Suppose Assumptions \ref{assum_region}, \ref{assum_noise}, \ref{assum_quasiuniform}, and \ref{assum_finRKHSPsi} hold. Then we have 
				\begin{align}\label{eq_thm_KRRrates}
					\|f-\hat f_n\|_{L_2} = O_p\left(T + n^{-\frac{m}{d}}\phi_n^{-m} + n^{-\frac{m}{d}}\lambda_n^{-\frac{1}{2}(1-\frac{m_0}{m})}\phi_n^{-\frac{2m-d}{2}(1-\frac{m_0}{m})}+n^{-\frac{m}{d}}\phi_n^{-\frac{2m-d}{2}}\lambda_n^{-1/2}T\right),
				\end{align}
				where $T=\max(T_1,T_2, T_3)$ with 
				\begin{align*}
					T_1 = & n^{-\frac{m_0}{d}}\phi_n^{-m_0} + n^{-\frac{m_0}{d}}\phi_n^{-\frac{(2m-d)m_0}{2m}(1-\frac{m_0}{m})}\lambda_n^{-\frac{m_0}{2m}(1-\frac{m_0}{m})} + \lambda_n^{1/2}\phi_n^{-d/2}+ \lambda_n^{\frac{m_0}{2m}}\phi_n^{\frac{(2m-d)m_0}{2m}},\\
					T_2 = & \phi_n^{-\frac{(2m-d)d}{2(2m+d)}}n^{-\frac{m}{2m+d}}\left(\phi_n^{-d}+ \lambda_n^{\frac{m_0}{m}-1}\phi_n^{\frac{(2m-d)m_0}{m}}\right)^{\frac{d}{2(2m+d)}},\\
					T_3 = & \lambda_n^{-\frac{d}{4m}}\phi_n^{-\frac{(2m-d)d}{4m}}n^{-1/2}. 
				\end{align*}
			\end{theorem}
			
			Theorem \ref{thm_KRRrates} states a relationship between the convergence rate of the KRR estimator and the scale parameter $\phi_n$ and the regularization parameter $\lambda_n$ under quasi-uniform design. To translate these theoretical insights into practical guidelines, we explore specific parameter choices that optimize convergence. Proposition \ref{prop_RKHS_phimunrelation} identifies optimal conditions for $\phi_n$ and $\lambda_n$, ensuring the minimization of prediction error.
			
			\begin{proposition}\label{prop_RKHS_phimunrelation}
				Let $\zeta_f(\phi_n,\lambda_n)$ denote the right-hand-side of \eqref{eq_thm_KRRrates}. It can be shown that $\zeta_f(\phi_n,\lambda_n) \asymp n^{-\frac{m_0}{2m_0+d}}$ by taking $\lambda_n\phi_n^{2m-d} \asymp n^{-\frac{2m}{2m_0+d}}$ with $\lambda_n \lesssim \phi_n^{d + \frac{2m m_0}{m - m_0}}$ and $\phi_n \gtrsim n^{-\frac{2m_0}{d(2m_0+d)}}$.
			\end{proposition}
			
			Combining Theorem \ref{thm_KRRrates} and Proposition \ref{prop_RKHS_phimunrelation} yields the results of Theorem \ref{thm_KRRratesX}.
			
			Before presenting the proof of Theorem \ref{thm_KRRrates}, we introduce two lemmas that are used in this proof. Lemma \ref{lem:approx1NN} provide bounds on the approximation error of penalized least square estimators with $L_2$ norm, whose proof can be found in Section \ref{app_subsec_pfL2_approx}.
			Lemma \ref{lemmaefsmall} is Lemma 8.4 of van de Geer (2000), which states that the inner product $\langle \epsilon, g \rangle_n$ is small.
			
			\begin{lemma}\label{lem:approx1NN}
				Let $f_n^*$ be the solution to the optimization problem
				\begin{align}\label{eq:krrestoraNN1}
					\min_{g\in \cN_\Phi(\Omega)}\|f-g\|_{L_2(\Omega)}^2 + \lambda_n\|g\|_{\cN_\Phi(\Omega)}^2.
				\end{align}
				Under conditions of Theorem \ref{thm_KRRrates}, we have
				\begin{align}\label{eq_approx_error}
					\|f_n^* - f\|_{L^2(\Omega)}^2 + \lambda_n \|f_n^*\|_{\cN_\Phi(\Omega)}^2 \lesssim \left(\lambda_n\phi_n^{-d}+ \lambda_n^{\frac{m_0}{m}}\phi_n^{\frac{(2m-d)m_0}{m}}\right)\|f\|_{H^{m_0}(\Omega)}^2.
				\end{align}
				In particular,
				\begin{align*}
					\|f_n^* - f\|_{L^2(\Omega)}^2 \lesssim & \left(\lambda_n\phi_n^{-d}+ \lambda_n^{\frac{m_0}{m}}\phi_n^{\frac{(2m-d)m_0}{m}}\right)\|f\|_{H^{m_0}(\Omega)}^2,\nonumber\\
					\|f_n^*\|_{\cN_\Phi(\Omega)}^2 \lesssim & \lambda_n^{-1}\left(\lambda_n\phi_n^{-d}+ \lambda_n^{\frac{m_0}{m}}\phi_n^{\frac{(2m-d)m_0}{m}}\right)\|f\|_{H^{m_0}(\Omega)}^2.
				\end{align*}
			\end{lemma}

			\begin{lemma}\label{lemmaefsmall}
				Suppose Assumption \ref{assum_noise} holds. 
				Then for all $t > C$, with probability at least $1 - C_1\exp(-C_2 t^2)$,
				\begin{align*}
					\sup_{g\in H^m(\Omega)}\frac{|\langle \epsilon, g \rangle_n|}{\|g\|_n^{1 - \frac{d}{2m}}\|g\|_{H^m(\Omega)}^{\frac{d}{2m}}} \leq tn^{-\frac{1}{2}}.
				\end{align*}
			\end{lemma}

			\noindent\textit{Proof of Theorem \ref{thm_KRRrates}.}
			Because $\hat f_n$ is the solution to the optimization problem \eqref{eq_KRRpredict_eq}, we have
			\begin{align}\label{eq_lemapprx2NNbi1}
				\|\hat f_n - \by\|_n^2 + \lambda_n\|\hat f_n\|_{\cN_\Phi(\Omega)}^2 \leq \|f_n^* - \by\|_n^2 + \lambda_n\|f_n^*\|_{\cN_\Phi(\Omega)}^2,
			\end{align}
			where $f_n^*$ is as in Lemma \ref{lem:approx1NN}. By rearrangement, \eqref{eq_lemapprx2NNbi1} implies
			\begin{align}\label{eq_lemapprx2Sobonoisebi1}
				\|f^*-\hat f_n\|_n^2 + C_5\lambda_n\|\hat f_n\|_{\cN_\Phi(\Omega)}^2 \leq  \|f^*-f_n^*\|_n^2 + C_6\lambda_n\|f_n^*\|_{\cN_\Phi(\Omega)}^2 + 2\langle \bvarepsilon, \hat f_n - f_n^*\rangle_n
			\end{align}
			for some positive constants $C_5$ and $C_6$. Since both $\hat f_n$ and $f_n^*$ are in $H^m(\Omega)$, Lemma \ref{lemmaefsmall} implies that 
			\begin{align}\label{eq_pfThmRKHS_noiseinner}
				\left|\langle \bvarepsilon, \hat f_n - f_n^*\rangle_n\right| = & O_{\PP}\left(n^{-\frac{1}{2}}\| \hat f_n - f_n^*\|_n^{1 - \frac{d}{2m}}\|\hat f_n - f_n^*\|_{H^m(\Omega)}^{\frac{d}{2m}}\right)\nonumber\\
				= & O_{\PP}\left(\phi_n^{-\frac{(2m-d)d}{4m}}n^{-\frac{1}{2}}\| \hat f_n - f_n^*\|_n^{1 - \frac{d}{2m}}\|f_n - f_n^*\|_{\cN_\Phi(\Omega)}^{\frac{d}{2m}}\right),
			\end{align}
			where the last inequality is by \eqref{eq_pfbound_ghm1}. Hence, by \eqref{eq_lemapprx2Sobonoisebi1} and \eqref{eq_pfThmRKHS_noiseinner}, we obtain
			\begin{align}\label{eq_lemapprx2Sobonoisebi2}
				& \|f^*-\hat f_n\|_n^2 + C_5\lambda_n\|\hat f_n\|_{\cN_\Phi(\Omega)}^2\nonumber\\
				\leq  & \|f^*-f_n^*\|_n^2 + C_6\lambda_n\|f_n^*\|_{\cN_\Phi(\Omega)}^2 + O_{\PP}\left(\phi_n^{-\frac{(2m-d)d}{4m}}n^{-\frac{1}{2}}\| \hat f_n - f_n^*\|_n^{1 - \frac{d}{2m}}\|\hat f_n - f_n^*\|_{\cN_\Phi(\Omega)}^{\frac{d}{2m}}\right)\nonumber\\
				\leq  & \|f^*-f_n^*\|_n^2 + C_6\lambda_n\|f_n^*\|_{\cN_\Phi(\Omega)}^2 + O_{\PP}\left(\phi_n^{-\frac{(2m-d)d}{4m}}n^{-\frac{1}{2}}\| \hat f_n - f_n^*\|_n^{1 - \frac{d}{2m}}\|f_n^*\|_{\cN_\Phi(\Omega)}^{\frac{d}{2m}}\right)\nonumber\\
				& + O_{\PP}\left(\phi_n^{-\frac{(2m-d)d}{4m}}n^{-\frac{1}{2}}\| \hat f_n - f_n^*\|_n^{1 - \frac{d}{2m}}\|\hat f_n\|_{\cN_\Phi(\Omega)}^{\frac{d}{2m}}\right),
			\end{align}
			where the last inequality is by the basic inequality $(a+b)^{\frac{d}{2m}}\leq a^{\frac{d}{2m}} + b^{\frac{d}{2m}}$ with $a,b\geq 0$ and $m>d/2$.

			For quasi-uniform design points, we can apply Lemma 28 in \cite{wang2022gaussian}, which yields
			\begin{align}\label{eq_diffofff1overfix}
				\|f^*-f_n^*\|_n \lesssim & n^{-\frac{m_0}{d}}\|f_n^*-f^*\|_{N_{\Psi}(\Omega)} + \|f^*-f_n^*\|_{L_2(\Omega)}\nonumber\\
				\lesssim & n^{-\frac{m_0}{d}}\|f^*\|_{N_{\Psi}(\Omega)} +n^{-\frac{m_0}{d}}\|f^*_n\|_{N_{\Psi}(\Omega)} + \|f^*-f_n^*\|_{L_2(\Omega)}\nonumber\\
				\lesssim & n^{-\frac{m_0}{d}}\|f^*\|_{N_{\Psi}(\Omega)} +n^{-\frac{m_0}{d}}\|f^*_n\|_{L_2(\Omega)}^{\frac{m-m_0}{m}}\|f^*_n\|_{H^m(\Omega)}^{\frac{m_0}{m}} + \|f^*-f_n^*\|_{L_2(\Omega)}\nonumber\\
				\lesssim & n^{-\frac{m_0}{d}}\|f^*\|_{N_{\Psi}(\Omega)} + n^{-\frac{m_0}{d}}\phi_n^{-\frac{(2m-d)m_0}{2m}}\|f^*_n\|_{L_2(\Omega)}^{\frac{m-m_0}{m}}\|f^*_n\|_{\cN_\Phi(\Omega)}^{\frac{m_0}{m}} + \|f^*-f_n^*\|_{L_2(\Omega)},
			\end{align}
			where the second inequality is by the triangle inequality, the third inequality is by the Gagliardo–Nirenberg interpolation inequality, and the last inequality is by \eqref{eq_pfbound_ghm1}.
			
			By Lemma \ref{lem:approx1NN}, \eqref{eq_diffofff1overfix} implies
			\begin{align}\label{eq_diffofff1overfix1}
				& \|f^*-f_n^*\|_n^2 + C_6\lambda_n\|f_n^*\|_{\cN_\Phi(\Omega)}^2 \nonumber\\
				\lesssim & n^{-\frac{2m_0}{d}} + n^{-\frac{2m_0}{d}}\phi_n^{-\frac{(2m-d)m_0}{m}}\lambda_n^{-\frac{m_0}{m}}\left(\lambda_n\phi_n^{-d}+ \lambda_n^{\frac{m_0}{m}}\phi_n^{\frac{(2m-d)m_0}{m}}\right)^{\frac{m_0}{m}} + \lambda_n\phi_n^{-d}+ \lambda_n^{\frac{m_0}{m}}\phi_n^{\frac{(2m-d)m_0}{m}}\nonumber\\
				\lesssim & n^{-\frac{2m_0}{d}} + n^{-\frac{2m_0}{d}}\phi_n^{-\frac{(2m-d)m_0}{m}}\lambda_n^{-\frac{m_0}{m}}\left(\lambda_n^{\frac{m_0}{m}}\phi_n^{-^{\frac{m_0d}{m}}}+ \lambda_n^{\frac{m_0^2}{m^2}}\phi_n^{\frac{(2m-d)m_0^2}{m^2}}\right) + \lambda_n\phi_n^{-d}+ \lambda_n^{\frac{m_0}{m}}\phi_n^{\frac{(2m-d)m_0}{m}}\nonumber\\
				\lesssim & n^{-\frac{2m_0}{d}}\phi_n^{-2m_0} + n^{-\frac{2m_0}{d}}\phi_n^{-\frac{(2m-d)m_0}{m}(1-\frac{m_0}{m})}\lambda_n^{-\frac{m_0}{m}(1-\frac{m_0}{m})} + \lambda_n\phi_n^{-d}+ \lambda_n^{\frac{m_0}{m}}\phi_n^{\frac{(2m-d)m_0}{m}}\nonumber\\
				\leq & T_1^2, 
			\end{align}
			where 
			\begin{align*}
				T_1 = & n^{-\frac{m_0}{d}}\phi_n^{-m_0} + n^{-\frac{m_0}{d}}\phi_n^{-\frac{(2m-d)m_0}{2m}(1-\frac{m_0}{m})}\lambda_n^{-\frac{m_0}{2m}(1-\frac{m_0}{m})} + \lambda_n^{1/2}\phi_n^{-d/2}+ \lambda_n^{\frac{m_0}{2m}}\phi_n^{\frac{(2m-d)m_0}{2m}}.
			\end{align*}
			In \eqref{eq_diffofff1overfix1}, the second inequality is by the basic inequality $(a+b)^{\frac{m_0}{m}} \leq a^{\frac{m_0}{m}} + b^{\frac{m_0}{m}}$ with $a,b\geq 0$ and $m\geq m_0$, and the third inequality is because $\phi_n\lesssim 1$. Combining \eqref{eq_lemapprx2Sobonoisebi2} and \eqref{eq_diffofff1overfix1} implies
			\begin{align}\label{eq_lemapprx2Sobonoisebi3}
				\|f^*-\hat f_n\|_n^2 + C_5\lambda_n\|\hat f_n\|_{\cN_\Phi(\Omega)}^2 \lesssim{} & T_1^2 + O_{\PP}\left(\phi_n^{-\frac{(2m-d)d}{4m}}n^{-\frac{1}{2}}\| \hat f_n - f_n^*\|_n^{1 - \frac{d}{2m}}\|f_n^*\|_{\cN_\Phi(\Omega)}^{\frac{d}{2m}}\right)\nonumber\\
				& + O_{\PP}\left(\phi_n^{-\frac{(2m-d)d}{4m}}n^{-\frac{1}{2}}\| \hat f_n - f_n^*\|_n^{1 - \frac{d}{2m}}\|\hat f_n\|_{\cN_\Phi(\Omega)}^{\frac{d}{2m}}\right).
			\end{align}
			By the Cauchy-Schwarz inequality, 
			\begin{align*}
				\|f^*-\hat f_n\|_n^2 = & \|f_n^*-\hat f_n\|_n^2 + \|f_n^*-f^*\|_n^2 + 2\langle f_n^*-\hat f_n, f_n^*-f^*\rangle_n\nonumber\\
				\geq & \frac{1}{2}\|f_n^*-\hat f_n\|_n^2 - \|f_n^*-f^*\|_n^2,
			\end{align*}
			which, together with \eqref{eq_lemapprx2Sobonoisebi3} and \eqref{eq_diffofff1overfix}, implies
			\begin{align}\label{eq_lemapprx2Sobonoisebi4}
				\|f_n^*-\hat f_n\|_n^2 + C_6\lambda_n\|\hat f_n\|_{\cN_\Phi(\Omega)}^2
				\lesssim{} & T_1^2 + O_{\PP}\left(\phi_n^{-\frac{(2m-d)d}{4m}}n^{-\frac{1}{2}}\| \hat f_n - f_n^*\|_n^{1 - \frac{d}{2m}}\|f_n^*\|_{\cN_\Phi(\Omega)}^{\frac{d}{2m}}\right)\nonumber\\
				& + O_{\PP}\left(\phi_n^{-\frac{(2m-d)d}{4m}}n^{-\frac{1}{2}}\| \hat f_n - f_n^*\|_n^{1 - \frac{d}{2m}}\|\hat f_n\|_{\cN_\Phi(\Omega)}^{\frac{d}{2m}}\right).
			\end{align}
			It can be seen that \eqref{eq_lemapprx2Sobonoisebi3} implies that at least one of the following inequality holds:
			\begin{align}
				\|f_n^*-\hat f_n\|_n^2 + C_6\lambda_n\|\hat f_n\|_{\cN_\Phi(\Omega)}^2 \lesssim & T_1^2,\label{eq_krr_empC1}\\
				\|f_n^*-\hat f_n\|_n^2 + C_6\lambda_n\|\hat f_n\|_{\cN_\Phi(\Omega)}^2 \lesssim &O_{\PP}\left(\phi_n^{-\frac{(2m-d)d}{4m}}n^{-\frac{1}{2}}\| \hat f_n - f_n^*\|_n^{1 - \frac{d}{2m}}\|f_n^*\|_{\cN_\Phi(\Omega)}^{\frac{d}{2m}}\right),\label{eq_krr_empC21}\\
				\|f_n^*-\hat f_n\|_n^2 + C_6\lambda_n\|\hat f_n\|_{\cN_\Phi(\Omega)}^2 \lesssim &O_{\PP}\left(\phi_n^{-\frac{(2m-d)d}{4m}}n^{-\frac{1}{2}}\| \hat f_n - f_n^*\|_n^{1 - \frac{d}{2m}}\|\hat f_n\|_{\cN_\Phi(\Omega)}^{\frac{d}{2m}}\right).\label{eq_krr_empC31}
			\end{align}
			Hence, we consider three cases separately.
			
			\noindent\textbf{Case 1. \eqref{eq_krr_empC1} holds.}
			In this case, we directly have
			\begin{align}\label{eq_krr_empC11}
				\|f_n^*-\hat f_n\|_n \leq T_1 \mbox{ and } \|\hat f_n\|_{\cN_\Phi(\Omega)}\leq \lambda_n^{-1/2}T_1.
			\end{align}
			
			\noindent\textbf{Case 2. \eqref{eq_krr_empC21} holds.} 
			In this case, solving \eqref{eq_krr_empC21} yields
			\begin{align}\label{eq_krr_empC2L2}
				\|f_n^*-\hat f_n\|_n  \lesssim & O_p\left(\phi_n^{-\frac{(2m-d)d}{4m}}n^{-\frac{1}{2}}\|f_n^*\|_{\cN_\Phi(\Omega)}^{\frac{d}{2m}}\right)^{\frac{2m}{2m+d}}\nonumber\\
				\lesssim & O_p\left(\phi_n^{-\frac{(2m-d)d}{2(2m+d)}}n^{-\frac{m}{2m+d}}\left(\phi_n^{-d}+ \lambda_n^{\frac{m_0}{m}-1}\phi_n^{\frac{(2m-d)m_0}{m}}\right)^{\frac{d}{2(2m+d)}}\right)\nonumber\\
				=: & O_p(T_2)
			\end{align}
			and 
			\begin{align}\label{eq_krr_empC2L3}
				\|\hat f_n\|_{\cN_\Phi(\Omega)}  \lesssim & O_p\left(\lambda_n^{-1/2} \left(\phi_n^{-\frac{(2m-d)d}{4m}}n^{-\frac{1}{2}}\|f_n^*\|_{\cN_\Phi(\Omega)}^{\frac{d}{2m}}\right)^{\frac{2m}{2m+d}}\right)
				\lesssim O_p\left(\lambda_n^{-1/2}T_2\right),
			\end{align}
			where the second inequality of \eqref{eq_krr_empC2L2} is by Lemma \ref{lem:approx1NN}.
			
			\noindent\textbf{Case 3. \eqref{eq_krr_empC31} holds.} 
			By Young's inequality, we have 
			\begin{align}\label{eq_krr_emp_middle_1}
				& \phi_n^{-\frac{(2m-d)d}{4m}} n^{-1/2}\|\hat f_n-f_n^*\|_n^{1-\frac{d}{2m}}\|\hat f_n\|_{\cN_\Phi}^{\frac{d}{2m}} \nonumber\\
				\leq & \frac{2m-d}{4m}c^2\|\hat f_n-f_n^*\|_n^2 + \frac{2m+d}{4m}\left(\frac{1}{c}\phi_n^{-\frac{(2m-d)d}{4m}} n^{-1/2}\|\hat f_n\|_{\cN_\Phi}^{\frac{d}{2m}}\right)^{\frac{4m}{2m+d}},
			\end{align}
			where $c>0$ is a constant small enough so that it can be absorbed by the first term in the left hand side of \eqref{eq_krr_empC31}. Then, absorb the first term in \eqref{eq_krr_emp_middle_1} into the left-hand side of \eqref{eq_krr_empC31} to obtain
			\begin{align}\label{eq_afterYoung}
				\|\hat f_n-f_n^*\|_n^2 + C_7\lambda_n\|\hat f_n\|_{\cN_\Phi}^2
				\lesssim & O_p\left(\phi_n^{-\frac{(2m-d)d}{4m}} n^{-1/2}\|\hat f_n\|_{\cN_\Phi}^{\frac{d}{2m}}\right)^{\frac{4m}{2m+d}}\nonumber\\
				= & O_p\left(\phi_n^{-\frac{(2m-d)d}{2m+d}} n^{-\frac{2m}{2m+d}}\|\hat f_n\|_{\cN_\Phi}^{\frac{2d}{2m+d}}\right).
			\end{align}
			Solving \eqref{eq_afterYoung} implies
			\begin{align}\label{eq_krr_empC3L2}
				\|f_n^*-\hat f_n\|_n  \lesssim & O_p\left(\phi_n^{-\frac{(2m-d)d}{4m}} n^{-1/2}\left(\lambda_n^{-1}\phi_n^{-\frac{(2m-d)d}{2m+d}} n^{-\frac{2m}{2m+d}}\right)^{\frac{(2m+d)d}{8m^2}}\right)^{\frac{2m}{2m+d}}\nonumber\\
				= & O_p\left(\lambda_n^{-\frac{d}{4m}}\phi_n^{-\frac{(2m-d)d}{4m}}n^{-1/2} \right)\nonumber\\
				=: & O_p(T_3),
			\end{align}
			and 
			\begin{align}\label{eq_krr_empC3L3}
				\|\hat f_n\|_{\cN_\Phi(\Omega)}  \lesssim &  O_p\left(\lambda_n^{-1}\phi_n^{-\frac{(2m-d)d}{2m+d}} n^{-\frac{2m}{2m+d}}\right)^{\frac{2m+d}{4m}}\nonumber\\
				= & O_p\left(\lambda_n^{-\frac{2m+d}{4m}}\phi_n^{-\frac{(2m-d)d}{4m}} n^{-1/2}\right) = O_p\left(\lambda_n^{-1/2}T_3\right).
			\end{align}
			
			It remains to bound the difference between the empirical norm and $L_2$ norm of $f^*-\hat f_n$, which can be done by using Lemma \ref{lemmafixdesign}:
			\begin{align}\label{eq_krr_emp_L2_diff}
				& \|f_n^*-\hat f_n\|_{L_2}\nonumber\\
				\lesssim & \|f^*-\hat f_n\|_n + n^{-\frac{m}{d}}\|f_n^*-\hat f_n\|_{H_m(\Omega)}\nonumber\\
				\lesssim & \|f^*-\hat f_n\|_n + n^{-\frac{m}{d}}\phi_n^{-\frac{2m-d}{2}}\|f_n^*-\hat f_n\|_{\cN_\Phi(\Omega)}\nonumber\\
				\lesssim & \|f^*-\hat f_n\|_n + n^{-\frac{m}{d}}\phi_n^{-\frac{2m-d}{2}}\|f_n^*\|_{\cN_\Phi(\Omega)}+n^{-\frac{m}{d}}\phi_n^{-\frac{2m-d}{2}}\|\hat f_n\|_{\cN_\Phi(\Omega)}\nonumber\\
				\lesssim & \|f^*-\hat f_n\|_n + n^{-\frac{m}{d}}\phi_n^{-\frac{2m-d}{2}}\lambda_n^{-1/2}\left(\lambda_n\phi_n^{-d}+ \lambda_n^{\frac{m_0}{m}}\phi_n^{\frac{(2m-d)m_0}{m}}\right)^{1/2}+n^{-\frac{m}{d}}\phi_n^{-\frac{2m-d}{2}}\|\hat f_n\|_{\cN_\Phi(\Omega)}\nonumber\\
				\lesssim & \|f^*-\hat f_n\|_n + n^{-\frac{m}{d}}\phi_n^{-m} + n^{-\frac{m}{d}}\lambda_n^{-\frac{1}{2}(1-\frac{m_0}{m})}\phi_n^{-\frac{2m-d}{2}(1-\frac{m_0}{m})}+n^{-\frac{m}{d}}\phi_n^{-\frac{2m-d}{2}}\|\hat f_n\|_{\cN_\Phi(\Omega)},
			\end{align}
			where the second inequality is by \eqref{eq_pfbound_ghm1}, the third inequality is by the triangle inequality, the fourth inequality is by Lemma \ref{lem:approx1NN}, and the last inequality is by the basic inequality $(a+b)^{\frac{1}{2}} \leq a^{\frac{1}{2}} + b^{\frac{1}{2}}$ with $a,b\geq 0$.
			
			By the triangle inequality,
			\begin{align*}
				\|f-\hat f_n\|_{L_2} \leq \|f_n^*-\hat f_n\|_{L_2} + \|f_n^*-f\|_{L_2},
			\end{align*}
			which, together with Lemma \ref{lem:approx1NN} and \eqref{eq_krr_emp_L2_diff}, implies 
			\begin{align}\label{eq_krr_emp_L2_diff2}
				& \|f-\hat f_n\|_{L_2}\nonumber\\
				\lesssim & \|f^*-\hat f_n\|_n + n^{-\frac{m}{d}}\phi_n^{-m} + n^{-\frac{m}{d}}\lambda_n^{-\frac{1}{2}(1-\frac{m_0}{m})}\phi_n^{-\frac{2m-d}{2}(1-\frac{m_0}{m})}+n^{-\frac{m}{d}}\phi_n^{-\frac{2m-d}{2}}\|\hat f_n\|_{\cN_\Phi(\Omega)}\nonumber\\
				& + \left(\lambda_n\phi_n^{-d}+ \lambda_n^{\frac{m_0}{m}}\phi_n^{\frac{(2m-d)m_0}{m}}\right)^{1/2}\nonumber\\
				\lesssim & \|f^*-\hat f_n\|_n + n^{-\frac{m}{d}}\phi_n^{-m} + n^{-\frac{m}{d}}\lambda_n^{-\frac{1}{2}(1-\frac{m_0}{m})}\phi_n^{-\frac{2m-d}{2}(1-\frac{m_0}{m})}+n^{-\frac{m}{d}}\phi_n^{-\frac{2m-d}{2}}\|\hat f_n\|_{\cN_\Phi(\Omega)} + T_1.
			\end{align}
			Plugging \eqref{eq_krr_empC11}, \eqref{eq_krr_empC2L2}, \eqref{eq_krr_empC2L3}, \eqref{eq_krr_empC3L2}, \eqref{eq_krr_empC3L3} into \eqref{eq_krr_emp_L2_diff2} finishes the proof.

			\section{Proof of Auxiliary Lemmas in the Supplementary Material}
			
			\subsection{Proof of Lemma \ref{LEM:BOFKMMAX}}
			
			If $\Psi$ is a Mat\'ern correlation function, the results of Lemma \ref{LEM:BOFKMMAX} have been shown by Lemma 22 of \cite{wang2022gaussian}. 
			
			It remains to show that $\Psi$ is a generalized Wendland correlation function. In this case, $\Psi(\|\bs-\bt\|_2) = 0$ as long as $\|\bs-\bt\|_2\geq C_1$ for some constant $C_1$. Without loss of generality, we let $C_1 > 1$. Define $\cB(\bx,r) = \{\bx'\in \RR^d: \|\bx-\bx'\|_2\leq r\}$ be the ball centered at $\bx$ with radius $r$. For any $\bx_j,\bx_k\in \bX$, $\Psi(\|\bx_j-\bx_k\|_2) > 0$ only if $\bx_k\in \cB(\bx_j,C_1)$. For any $\bx_j\in \bX$ if there are $N$ points $\bx_k\in \bX$ such that $\bx_k\in \cB(\bx_j,C_1)$, by the fact that $q_{\bX}\leq 1 < C_1$, we have 
			\begin{align*}
				\bigcup_{\bx_k\in \cB(\bx_j,C_1)}\cB(\bx_k,q_{\bX}/2) \subseteq \cB(\bx_j,2C_1).
			\end{align*}
			Since the separation radius is $q_{\bX}$, the balls $\cB(\bx_j,q_{\bX}/2)$ are disjoint for $\bx_j\in \bX$. By the comparison of volume, it holds that 
			\begin{align*}
				{\rm Vol}\left(\bigcup_{\bx_k\in \cB(\bx_j,C_1)}\cB(\bx_k,q_{\bX}/2)\right) \leq {\rm Vol}\left(\cB(\bx_j,2C_1)\right),
			\end{align*}
			which further implies
			\begin{align*}
				Nq_{\bX}^d \leq C_2
			\end{align*}
			for some constant $C_2>0$. Thus, $N \leq C_2q_{\bX}^{-d}$, which implies that for each row of the correlation matrix $(\Psi(\bx_j-\bx_k))_{jk}$ has at most $C_2q_{\bX}^{-d}$ non-zero elements. 
			
			By the Ger{\u s}hgorin circle theorem \citep[Corollary~6.1.5]{HornJohnson12}, we have
			\begin{align*}
				\Lambda_{\bX} \leq & \max_{1\leq j\leq n }\sum_{k=1}^n \Psi(\bx_j-\bx_k)\nonumber\\ 
				\leq & N \max_{1\leq j,k\leq n }\Psi(\bx_j-\bx_k) \nonumber\\ 
				\leq &  C_3q_{\bX}^{-d}, 
			\end{align*}
			which finishes the proof for the generalized Wendland correlation function case.
			
			\subsection{Proof of Lemma \ref{LEMMAERRORWNUG}}\label{app:pflemwnug}
			
			We first present the following fixed-design inequality, which follows from Lemma 26 of \citet{wang2022gaussian} under quasi-uniform designs. In the rest of this supplementary material, \(H^m(\Omega)\) denotes the Sobolev space with smoothness \(m\), with norm \(\|\cdot\|_{H^m(\Omega)}\). For simplicity, set \(h_n = h_{\bX,\Omega}\). Also let \(\|g\|_n\) be the empirical norm
			\begin{align*}
				\|g\|_n = \left(\frac{1}{n}\sum_{j=1}^n g(\bx_j)^2\right)^{\frac{1}{2}}.
			\end{align*}
			
			\begin{lemma}\label{lemmafixdesign}
				Suppose $g\in H^m(\Omega)$ for some $m > d/2$. If Assumption \ref{assum_quasiuniform} holds, then we have
				\begin{align*}
					\|g\|_{L_2(\Omega)} \leq C(h_n^m\|g\|_{H^m(\Omega)} + \|g\|_n)
				\end{align*}
				holds for all $n$, 	where $C$ is a positive constant not depending on $g$ and $n$. 
			\end{lemma}
			
			Fix \(\bx\in\Omega\), set
			\[
			A_\Phi=\Rb_\Phi+\rho_n\Ib_n,\qquad\br_{\bx}=\rb_\Phi(\bx),\qquad \bu=A_\Phi^{-1}\br_{\bx},
			\]
			and define
			\[
			I(\bx)=\Phi(\bx-\bx)-\br_{\bx}^\top A_\Phi^{-1}\br_{\bx}.
			\]
			For \(\bt\in\Omega\), define
			\[
			g_{\bx}(\bt)=\Phi(\bx-\bt)-\br_{\bx}^\top A_\Phi^{-1}\rb_\Phi(\bt)=\Phi(\bx-\bt)-\bu^\top\rb_\Phi(\bt).
			\]
			Then \(g_{\bx}(\bx)=I(\bx)\), so
			\begin{align}\label{eq_GpovergI}
				I(\bx)\leq \|g_{\bx}\|_{L_\infty(\Omega)}.
			\end{align}
			
			We first control the native-space and empirical norms of \(g_{\bx}\). Since \(A_\Phi\bu=\br_{\bx}\),
			\begin{align}\label{eq_lemIgNnorm1}
				\|g_{\bx}\|_{\cN_\Phi(\Omega)}^2&= \Phi(\bx-\bx)-2\br_{\bx}^\top\bu+\bu^\top\Rb_\Phi\bu \nonumber\\
				&= \Phi(\bx-\bx)-\br_{\bx}^\top\bu-\rho_n\|\bu\|_2^2= I(\bx)-\rho_n\|\bu\|_2^2\leq I(\bx).
			\end{align}
			In particular, \(I(\bx)\geq0\) and
			\begin{align}\label{eq_u_bound_by_I}
				\rho_n\|\bu\|_2^2\leq I(\bx).
			\end{align}
			Moreover,
			\[
			(g_{\bx}(\bx_1),\ldots,g_{\bx}(\bx_n))^\top=\br_{\bx}-\Rb_\Phi\bu=\rho_n\bu,
			\]
			and therefore
			\begin{align}\label{eq_empirical_norm_g}
				\|g_{\bx}\|_n^2=\frac{\rho_n^2}{n}\|\bu\|_2^2\leq \frac{\rho_n}{n}I(\bx).
			\end{align}
			
			Next we record the scale-dependent comparison between \(H^m(\Omega)\) and \(\cN_\Phi(\Omega)\). By \eqref{eq_FtransPhi1} and \eqref{eq_FtransPhi2}, for \(\Phi=\Phi_{\phi_n}\),
			\[
			(1+\|\bomega\|_2^2)^m\leq C\phi_n^{-2m+d}\{\cF(\Phi)(\bomega)\}^{-1},\qquad \bomega\in\RR^d.
			\]
			Thus, using the natural extension of the native space from \(\Omega\) to \(\RR^d\),
			\begin{align}\label{eq_pfbound_ghm1}
				\|v\|_{H^m(\Omega)}\leq C\phi_n^{-m+d/2}\|v\|_{\cN_\Phi(\Omega)},\qquad v\in\cN_\Phi(\Omega).
			\end{align}
			
			By the Gagliardo--Nirenberg interpolation inequality for Sobolev spaces \citep{BrezisMironescu19,leoni2017first}, because \(m>d/2\),
			\begin{align}\label{eq_interpolation}
				\|g_{\bx}\|_{L_\infty(\Omega)}\leq C\|g_{\bx}\|_{L_2(\Omega)}^{1-\frac{d}{2m}}\|g_{\bx}\|_{H^m(\Omega)}^{\frac{d}{2m}}.
			\end{align}
			Combining \eqref{eq_lemIgNnorm1}, \eqref{eq_pfbound_ghm1}, and \eqref{eq_interpolation} gives
			\begin{align}\label{eq_interpolation_NPhi}
				\|g_{\bx}\|_{L_\infty(\Omega)}\leq C\phi_n^{-\frac{d(2m-d)}{4m}}\|g_{\bx}\|_{L_2(\Omega)}^{1-\frac{d}{2m}}I(\bx)^{\frac{d}{4m}}.
			\end{align}
			
			Lemma \ref{lemmafixdesign}, \eqref{eq_pfbound_ghm1}, \eqref{eq_lemIgNnorm1}, and \eqref{eq_empirical_norm_g} imply
			\begin{align}\label{eq:Gpoverg1}
				\|g_{\bx}\|_{L_2(\Omega)}&\leq C\left(h_n^m\|g_{\bx}\|_{H^m(\Omega)}+\|g_{\bx}\|_n\right)\nonumber\\
				&\leq C\left(h_n^m\phi_n^{-m+d/2}+\left(\frac{\rho_n}{n}\right)^{1/2}\right)I(\bx)^{1/2}.
			\end{align}
			Using \eqref{eq_GpovergI}, \eqref{eq_interpolation_NPhi}, and \eqref{eq:Gpoverg1}, we obtain
			\begin{align}\label{eq_Ixbound_1}
				I(\bx)&\leq C\phi_n^{-\frac{d(2m-d)}{4m}}
				\|g_{\bx}\|_{L_2(\Omega)}^{1-\frac{d}{2m}}I(\bx)^{\frac{d}{4m}}\nonumber\\
				&\leq C\phi_n^{-\frac{d(2m-d)}{4m}}\left(h_n^m\phi_n^{-m+d/2}+\left(\frac{\rho_n}{n}\right)^{1/2}\right)^{1-\frac{d}{2m}}I(\bx)^{1/2}.
			\end{align}
			If \(I(\bx)=0\), the desired result is immediate. Otherwise, divide \eqref{eq_Ixbound_1} by \(I(\bx)^{1/2}\) and square both sides:
			\[
			I(\bx)\leq C\phi_n^{-\frac{d(2m-d)}{2m}}\left(h_n^m\phi_n^{-m+d/2}+\left(\frac{\rho_n}{n}\right)^{1/2}\right)^{2(1-\frac{d}{2m})}.
			\]
			Since \(m>d/2\), \((a+b)^{2(1-d/(2m))}\leq C(a^2+b^2)^{1-d/(2m)}\) for \(a,b\geq0\). Taking \(a=h_n^m\phi_n^{-m+d/2}\) and \(b=(\rho_n/n)^{1/2}\) yields
			\[
			I(\bx)\leq C\phi_n^{-\frac{d(2m-d)}{2m}}\left(h_n^{2m}\phi_n^{-2m+d}+\frac{\rho_n}{n}\right)^{1-\frac{d}{2m}}.
			\]
			This proves Lemma \ref{LEMMAERRORWNUG}.

			\subsection{Proof of Lemma \ref{lem:approx1NN}}\label{app_subsec_pfL2_approx}
			
			Let $\tilde f_n$ be the solution to the optimization problem
			\begin{align*}
				\min_{g\in \cN_\Phi(\RR^d)}\|f-g\|_{L_2(\RR^d)}^2 + \lambda_n\|g\|_{\cN_\Phi(\RR^d)}^2.
			\end{align*}
			Since $\Omega$ is compact with Lipschitz boundary, the Sobolev extension theorem and the natural extension of RKHS imply that there exists a constant $C_1$ such that (with a slight abuse of notation, we still use $f$ to denote the extension of $f$)
			\begin{align}\label{eq_pf_approxRKHS_1}
				\|f_n^* - f\|_{L^2(\Omega)}^2 + \lambda_n \|f_n^*\|_{\cN_\Phi(\Omega)}^2 \leq & \|f-\tilde f_n\|_{L_2(\Omega)}^2 + \lambda_n\|\tilde f_n\|_{\cN_\Phi(\Omega)}^2\nonumber\\
				\leq & C_1\left(\|\tilde f_n - f\|_{L^2(\RR^d)}^2 + \lambda_n \|\tilde f_n\|_{\cN_\Phi(\RR^d)}^2\right).
			\end{align}
			Therefore, it suffices to bound
			\begin{align*}
				\|\tilde f_n - f\|_{L^2(\RR^d)}^2 + \lambda_n \|\tilde f_n\|_{\cN_\Phi(\RR^d)}^2,
			\end{align*}
			which, by the Fourier inversion theorem, can be further bounded by
			\begin{align}\label{eq_pf_approxRKHS_decomp}
				& \|\tilde f_n - f\|_{L^2(\RR^d)}^2 + \lambda_n \|\tilde f_n\|_{\cN_\Phi(\RR^d)}^2\nonumber\\
				= & \int_{\RR^d} |\mathcal{F}(f)(\bomega)-\mathcal{F}(\tilde f_n)(\bomega)|^2 + \lambda_n \frac{|\mathcal{F}(\tilde f_n)(\bomega)|^2}{\cF(\Phi)(\bomega)}{\rm d}\bomega\nonumber\\
				= & \int_{\RR^d} \frac{\lambda_n\cF(\Phi)(\bomega)^{-1}}{1+\lambda_n\cF(\Phi)(\bomega)^{-1}}|\mathcal{F}(f)(\bomega)|^2{\rm d}\bomega\nonumber\\
				= & \int_{\|\bomega\|_2 \leq \phi_n^{-1}w_0}  \frac{\lambda_n\cF(\Phi)(\bomega)^{-1}}{1+\lambda_n\cF(\Phi)(\bomega)^{-1}}|\mathcal{F}(f)(\bomega)|^2{\rm d}\bomega + \int_{\phi_n^{-1}w_0 < \|\bomega\|_2 } \frac{\lambda_n\cF(\Phi)(\bomega)^{-1}}{1+\lambda_n\cF(\Phi)(\bomega)^{-1}}|\mathcal{F}(f)(\bomega)|^2{\rm d}\bomega\nonumber\\
				= & \cI_1 + \cI_2.
			\end{align}
			The term $\cI_1$ in \eqref{eq_pf_approxRKHS_decomp} can be bounded by 
			\begin{align}\label{eq_pf_approxRKHS_I1}
				\cI_1 \leq & \int_{\|\bomega\|_2 \leq \phi_n^{-1}w_0}  \frac{\lambda_nc_3^{-1}\phi_n^{-d}}{1+c_4^{-1}\lambda_n\phi_n^{-d}}|\mathcal{F}(f)(\bomega)|^2{\rm d}\bomega    \leq C_2\lambda_n\phi_n^{-d}\int_{\|\bomega\|_2 \leq \phi_n^{-1}w_0}|\mathcal{F}(f)(\bomega)|^2{\rm d}\bomega,
			\end{align}
			where the first inequality is by \eqref{eq_FtransPhi2}, and the second inequality is by $\lambda_n\leq \phi_n^d$.
			
			Next, we consider $\cI_2$. By \eqref{eq_FtransPhi1}, we have
			\begin{align}\label{eq_pf_approxRKHS_I2}
				\cI_2 \leq & \int_{\|\bomega\|_2 > \phi_n^{-1}w_0}  \frac{\lambda_nc_1^{-1}\phi_n^{2m-d}\|\bomega\|_2^{2m} }{1 + \lambda_nc_2^{-1}\phi_n^{2m-d}\|\bomega\|_2^{2m}}|\mathcal{F}(f)(\bomega)|^2{\rm d}\bomega\nonumber\\
				= & \left(\int_{I_{21}} + \int_{I_{22}}\right)  \frac{c_1^{-1}\lambda_n\phi_n^{2m-d}\|\bomega\|_2^{2m} }{1 + c_2^{-1}\lambda_n\phi_n^{2m-d}\|\bomega\|_2^{2m}}|\mathcal{F}(f)(\bomega)|^2{\rm d}\bomega = \cI_{21} + \cI_{22},
			\end{align}
			where 
			\begin{align*}
				I_{21} & = \left\{ \bomega \in \RR^d: \|\bomega\|_2 > \phi_n^{-1}w_0, \lambda_n\phi_n^{2m-d}\|\bomega\|_2^{2m} \leq 1 \right\},\\
				I_{22} & = \left\{ \bomega \in \RR^d: \|\bomega\|_2 > \phi_n^{-1}w_0, \lambda_n\phi_n^{2m-d}\|\bomega\|_2^{2m} > 1 \right\}. 
			\end{align*}
			The term $\cI_{21}$ can be bounded by
			\begin{align}\label{eq_pf_approxRKHS_I21}
				\cI_{21} \leq & C_3\int_{I_{21}} \lambda_n\phi_n^{2m-d}\|\bomega\|_2^{2m} |\mathcal{F}(f)(\bomega)|^2{\rm d}\bomega \nonumber\\
				\leq & C_3\int_{I_{21}} \left(\lambda_n\phi_n^{2m-d}\|\bomega\|_2^{2m}\right)^{\frac{m_0}{m}} |\mathcal{F}(f)(\bomega)|^2{\rm d}\bomega\nonumber\\
				\leq & C_3 \lambda_n^{\frac{m_0}{m}}\phi_n^{\frac{(2m-d)m_0}{m}}\int_{I_{21}}(1+\|\bomega\|_2^2)^{m_0}|\mathcal{F}(f)(\bomega)|^2{\rm d}\bomega\nonumber\\
				\leq & C_4\lambda_n^{\frac{m_0}{m}}\phi_n^{\frac{(2m-d)m_0}{m}}\|f\|_{H^{m_0}(\RR^d)}^2.
			\end{align}
			The term $\cI_{22}$ can be bounded by
			\begin{align}\label{eq_pf_approxRKHS_I22}
				\cI_{22} \leq & C_5\int_{I_{22}} |\mathcal{F}(f)(\bomega)|^2{\rm d}\bomega \nonumber\\
				\leq & C_5\int_{I_{22}} \left(\lambda_n\phi_n^{2m-d}\|\bomega\|_2^{2m}\right)^{\frac{m_0}{m}} |\mathcal{F}(f)(\bomega)|^2{\rm d}\bomega\nonumber\\
				\leq & C_5 \lambda_n^{\frac{m_0}{m}}\phi_n^{\frac{(2m-d)m_0}{m}}\int_{I_{22}}(1+\|\bomega\|_2^2)^{m_0}|\mathcal{F}(f)(\bomega)|^2{\rm d}\bomega\nonumber\\
				\leq & C_6\lambda_n^{\frac{m_0}{m}}\phi_n^{\frac{(2m-d)m_0}{m}}\|f\|_{H^{m_0}(\RR^d)}^2.
			\end{align}
			Plugging \eqref{eq_pf_approxRKHS_I1}, \eqref{eq_pf_approxRKHS_I2}, \eqref{eq_pf_approxRKHS_I21}, and \eqref{eq_pf_approxRKHS_I22} into \eqref{eq_pf_approxRKHS_decomp} yields 
			\begin{align}\label{eq_pf_approxRKHS_bound}
				\|\tilde f_n - f\|_{L^2(\RR^d)}^2 + \lambda_n \|\tilde f_n\|_{\cN_\Phi(\RR^d)}^2\lesssim & \left(\lambda_n\phi_n^{-d}+ \lambda_n^{\frac{m_0}{m}}\phi_n^{\frac{(2m-d)m_0}{m}}\right)\|f\|_{H^{m_0}(\RR^d)}^2\nonumber\\
				\lesssim & \left(\lambda_n\phi_n^{-d}+ \lambda_n^{\frac{m_0}{m}}\phi_n^{\frac{(2m-d)m_0}{m}}\right)\|f\|_{H^{m_0}(\Omega)}^2,
			\end{align}
			which finishes the proof.

			\section{Experimental Settings}\label{app:experimental_settings}
			
			This section provides the detailed experimental settings used in Sections~\ref{sec_numerical} and~\ref{sec_real} of the main text.
			
			\subsection{Gaussian Process Simulation}
			
			For the Gaussian-process simulation in Section~\ref{subsec:gp_sim}, we consider model~\eqref{recoveringGP}, where the true function $f$ is a realization of a Gaussian process with Mat\'{e}rn covariance function~\eqref{materngai}. The Mat\'{e}rn order $m_0-d/2$ takes values $\{0.5,\,1.5,\,2.5\}$. The scale is fixed at $\theta = 10/3$ and the variance at $\sigma^2=1$. Observation noise satisfies $\varepsilon_i \sim \mathcal{N}(0,\,0.1)$, consistent with Assumption~\ref{assum_noise}. Design points are generated by Latin hypercube sampling on $\Omega = [-1,1]^d$ with $d \in \{2,3\}$; these space-filling designs are used as numerical analogues of the quasi-uniform fixed designs in the theory. Sample sizes are $n \in \{100,\,200,\,400,\,800,\,1600\}$ for $d=2$ and $n \in \{200,\,400,\,800,\,1500\}$ for $d=3$.
			
			Wendland kernels are imposed as the correlation function in prediction. We report the admissible pairs
			\[
			(m_0-d/2,\,m-(d+1)/2)\in\{(0.5,1),(0.5,2),(0.5,3), (1.5,2),(1.5,3),(2.5,3)\},
			\]
			so that the imposed kernel smoothness is at least as large as the truth smoothness. Each configuration is replicated 30 times with data split $70\%/15\%/15\%$ (train/validation/test). The GP nugget regularisation is fixed at $\rho_n = 0.1$. The theoretical choice \(\rho_n\asymp\phi_n^d\) is a sufficient asymptotic scaling for preserving the optimal prediction rate. In the simulations, we fix \(\rho_n\) in order to compare Matérn and Wendland predictors under the same noise-to-signal regularization level and to isolate the effect of the support radius. The power-law schedule for \(\phi_n\) is retained, while the multiplicative constant \(C_\phi\) is selected by validation. This finite-sample convention is closer to common GP practice, where the nugget or regularization level is often fixed or estimated separately from the covariance range. Mat\'{e}rn GP predictors with the correct smoothness $m_0$ serve as dense-kernel baselines under identical $\rho_n$. For a fair timing comparison, the Mat\'{e}rn and Wendland linear systems are solved with the same backend, specifically conjugate gradients for both kernels.
			
			For the compactly supported Wendland predictor, we impose the theory-guided schedule \(\phi_n=C_\phi n^{-(m-m_0)/(2mm_0)}.\) The constant \(C_\phi\) controls the accuracy--efficiency trade-off. A larger \(C_\phi\) gives a larger support radius, a denser correlation matrix, and typically lower prediction error but slower sparse matrix-vector products. A smaller \(C_\phi\) gives a sparser matrix and faster sparse matrix-vector products, but may increase prediction error. Here \(C_\phi\) is selected through a validation set with five replications.
			The selected $C_\phi$ is then fixed across all 30 reporting replications.
			
			\subsection{Deterministic Function Approximation}\label{app_exp_deter}
			
			For the deterministic-function experiment in Section~\ref{sec_KRR}, the underlying regression function is fixed rather than sampled from a Gaussian process. We use smooth nonlinear functions containing oscillatory terms, curvature, and interactions:
			\[
			f_2(x_1,x_2) =1.2\sin(\pi x_1)+0.8\cos(2\pi x_2) +0.6x_1x_2+0.5x_1^2-0.4x_2^3,\quad (x_1,x_2)\in[-1,1]^2,
			\]
			and
			\[
			f_3(x_1,x_2,x_3)=f_2(x_1,x_2)+0.7\sin(\pi x_1x_3)+0.5x_2x_3+0.3x_3^2,\quad (x_1,x_2,x_3)\in[-1,1]^3.
			\]
			Observations follow
			\[
			y_i=f(x_i)+\varepsilon_i,\quad \varepsilon_i\sim\mathcal{N}(0,0.01).
			\]
			The design and fitting settings are as in the GP experiment: Latin hypercube designs on $[-1,1]^d$, the same sample sizes, $70\%/15\%/15\%$ train/validation/test split, regularisation $\lambda_n=0.1$, and the same CG solver backend for both methods. The KRR theory gives sufficient joint scalings of \(\lambda_n\) and \(\phi_n\) for minimax-rate estimation. In the deterministic experiments, we fix \(\lambda_n\) across the Matérn and Wendland estimators so that differences in error and computation are driven primarily by the kernel support and its scale. The resulting experiments should therefore be viewed as finite-sample comparisons under a common regularization convention, not as an exhaustive empirical optimization of the asymptotic KRR tuning rule. The Mat\'{e}rn KRR estimator is used as the baseline, while Wendland KRR uses the same admissible pairs $(m_0-d/2,\,m-(d+1)/2)$ and the scheduled bandwidth
			\[
			\phi_n=C_\phi n^{-(m-m_0)/(2mm_0)}.
			\]
			For each $(m_0,m,d)$ configuration, the constant $C_\phi$ is selected on the validation set. Test error is computed against the noiseless deterministic truth.

			\subsection{Additional Scalable-Baseline Comparison}
			
			For the additional scalable-baseline comparison in Section~\ref{subsec:additional_gp_baselines}, the GP truth is the regular Mat\'{e}rn model with $m_0-d/2=0.5$, as in Section~\ref{subsec:gp_sim}; observations have Gaussian noise variance $0.1$. The design is Latin hypercube sampling on $[-1,1]^2$, with sample sizes $n\in\{200,400,800,1600\}$. Each configuration is replicated five times. To keep the timing definition consistent with Section~\ref{subsec:gp_sim}, the dense parent, Wendland, and tapered parent methods report linear-solver time; local GP and Nystr\"om report their natural evaluation time because they do not have a single global linear solve.
			
			The dense parent GP uses the same Mat\'{e}rn covariance as the data-generating model. Covariance tapering multiplies this parent covariance by a compactly supported Wendland taper; the taper uses the same selected Wendland smoothness $m=(d+5)/2$ and a support radius $\phi=1$ in both the GP and deterministic runs. The local GP predictor uses the same parent covariance restricted to the $0.1n$ nearest training sites for each test point \citep{gramacy2015local,gramacy2015speeding}. The Nystr\"om approximation uses randomly selected landmark points with $m_{\mathrm{land}}=\lceil0.2n_{\mathrm{train}}\rceil$, capped at $300$ landmarks. Dense parent, Wendland, and tapered parent systems are solved with the same CG backend; local GP and Nystr\"om use their natural small dense and low-rank linear algebra.
			
			We also run the same scalable-baseline comparison for the $d=2$ deterministic benchmark used in Section~\ref{sec_KRR}. The truth is the nonlinear function $f_2$ defined in Section \ref{app_exp_deter}, observations have noise variance $0.01$, and test MSE is computed against the noiseless deterministic function. For this deterministic comparison, the dense parent, local, and Nystr\"om methods use the Mat\'{e}rn parent kernel with $m_0-d/2=0.5$. The Wendland predictor uses the deterministic scheduled support radius selected in Section~\ref{sec_KRR}, while covariance tapering uses the fixed support radius $\phi=1$.
			
			\subsection{ERA5 Real-Data Application}

			For the real-data experiment in Section~\ref{sec_real}, we use ERA5 two-meter temperature data. The dataset contains 68 global temperature fields from 2025-01-01 00:00 to 2026-05-01 18:00 on a $0.25^\circ \times 0.25^\circ$ grid, with 1,038,240 spatial observations per time point. Each time point is treated as a separate two-dimensional latitude-longitude regression problem, and the reported MSE values are averages across all available time points. Temperature values are standardized before fitting, and prediction error is measured on held-out observations.
			
			The ERA5 experiment is intended as an operational illustration rather than a direct test of the fixed-domain asymptotic bandwidth formula. The data are nonstationary, anisotropic, and observed on a latitude-longitude grid, so a fixed moderate support radius selected from a bandwidth-sensitivity analysis is used to demonstrate the practical accuracy--cost properties. The fact that the selected support radius is approximately stable over the reported sample sizes should not be interpreted as contradicting the asymptotic sufficient scaling; rather, it reflects finite-sample tuning under a real spatial field.
			
			The real-data experiment follows the same computational convention as the simulations. Mat\'{e}rn and Wendland systems are solved with conjugate gradients for a fair backend comparison, while the Wendland covariance matrix is stored sparsely. We use the rougher Wendland kernel with $m=5/2$ and $\rho_n=10^{-4}$. For the Mat\'{e}rn baseline we choose $m_0=5/2$ in two dimensions. We set $\phi_n=30$ as a moderate support radius suggested by the bandwidth-sensitivity experiment. The dense Mat\'{e}rn baseline is reported through $n=5000$ only because repeated dense solves over all time points are computationally expensive; this truncation is a computational cap, not a statistical stopping rule.
			
			The convergence summary uses $n\in\{200,500,1000,2000,5000,10000\}$ with $\phi_n=30$. The bandwidth-sensitivity experiment fixes $n=5000$ and varies $\phi\in\{5,10,20,30,40,50\}$. The North America regional visualization uses the window $230^\circ$E--$300^\circ$E and $25^\circ$N--$55^\circ$N at 00:00, 06:00, 12:00, and 18:00 on 2025-10-01. The same sampled regional instance is used at each time point with $n=10{,}000$, 7,000 training points, $m=5/2$, and $\phi=30$.
			
			\section{Additional Tables}\label{app:tables}
			
			\subsection{Gaussian Process Experiments}
			
			Tables~\ref{tab:gp_full_d2_a}--\ref{tab:gp_full_d3_b} report Gaussian-process prediction error and computation time across dimensions, sample sizes, and smoothness settings, while Table~\ref{tab:slopes} condenses the same experiments into empirical log--log convergence slopes. Taken together, these results show that the scale-adjusted Wendland model stays reasonably close to the Mat\'{e}rn baseline in test MSE while usually delivering faster solver times when the induced sparsity is strong enough.
			
			\begin{table}[htbp]
				\centering
				\renewcommand{\arraystretch}{1.10}
				\setlength{\tabcolsep}{2.2pt}
				\footnotesize
				\caption{GP regression: test MSE and computation time ($d=2$, $m_0-d/2=0.5$).}
				\label{tab:gp_full_d2_a}
				\begin{tabular}{@{}cc r c c c r r c@{}}
					\toprule
					& & & \multicolumn{3}{c}{Test MSE} & \multicolumn{3}{c}{Computation time} \\
					\cmidrule(lr){4-6}\cmidrule(lr){7-9}
					$m_0-d/2$ & $m-(d+1)/2$ & $n$ & \multicolumn{1}{c}{W} & \multicolumn{1}{c}{M} & \multicolumn{1}{c}{Ratio} & \multicolumn{1}{c}{$T_W$ (ms)} & \multicolumn{1}{c}{$T_M$ (ms)} & \multicolumn{1}{c}{Speedup} \\
					\midrule
					$0.5$ & $1$ & 100 & 0.1058\,(0.052) & 0.06787\,(0.026) & 1.56$\times$ & 0.337 & 0.329 & 0.98$\times$ \\
					&  & 200 & 0.06587\,(0.016) & 0.04763\,(0.017) & 1.38$\times$ & 0.577 & 0.799 & 1.39$\times$ \\
					&  & 400 & 0.05695\,(0.018) & 0.04304\,(0.0094) & 1.32$\times$ & 2.004 & 3.941 & 1.97$\times$ \\
					&  & 800 & 0.03653\,(0.0047) & 0.02914\,(0.0027) & 1.25$\times$ & 7.656 & 22.36 & 2.92$\times$ \\
					&  & 1600 & 0.02703\,(0.0026) & 0.02328\,(0.002) & 1.16$\times$ & 34.82 & 102.18 & 2.93$\times$ \\
					\cmidrule(lr){3-9}
					& $2$ & 100 & 0.09305\,(0.044) & 0.06787\,(0.026) & 1.37$\times$ & 0.357 & 0.329 & 0.92$\times$ \\
					&  & 200 & 0.06038\,(0.018) & 0.04763\,(0.017) & 1.27$\times$ & 0.700 & 0.799 & 1.14$\times$ \\
					&  & 400 & 0.05316\,(0.017) & 0.04304\,(0.0094) & 1.24$\times$ & 2.504 & 3.941 & 1.57$\times$ \\
					&  & 800 & 0.03421\,(0.0045) & 0.02914\,(0.0027) & 1.17$\times$ & 11.10 & 22.36 & 2.01$\times$ \\
					&  & 1600 & 0.02589\,(0.0025) & 0.02328\,(0.002) & 1.11$\times$ & 49.38 & 102.18 & 2.07$\times$ \\
					\cmidrule(lr){3-9}
					& $3$ & 100 & 0.09878\,(0.049) & 0.06787\,(0.026) & 1.46$\times$ & 0.349 & 0.329 & 0.94$\times$ \\
					&  & 200 & 0.06384\,(0.017) & 0.04763\,(0.017) & 1.34$\times$ & 0.743 & 0.799 & 1.08$\times$ \\
					&  & 400 & 0.05769\,(0.02) & 0.04304\,(0.0094) & 1.34$\times$ & 2.628 & 3.941 & 1.50$\times$ \\
					&  & 800 & 0.03668\,(0.0049) & 0.02914\,(0.0027) & 1.26$\times$ & 10.84 & 22.36 & 2.06$\times$ \\
					&  & 1600 & 0.02721\,(0.0027) & 0.02328\,(0.002) & 1.17$\times$ & 45.67 & 102.18 & 2.24$\times$ \\
					\bottomrule
					\multicolumn{9}{@{}p{0.97\linewidth}@{}}{\footnotesize W = Wendland (compactly supported); M = Mat\'{e}rn-kernel baseline. Test MSE values shown as mean\,(std) over 30 replications. $\text{Ratio}=\widehat{\text{MSE}}_W/\widehat{\text{MSE}}_M$;\ $\text{Speedup}=T_M/T_W$.}\\
				\end{tabular}
			\end{table}
			
			\begin{table}[htbp]
				\centering
				\renewcommand{\arraystretch}{1.10}
				\setlength{\tabcolsep}{2.2pt}
				\footnotesize
				\caption{GP regression: test MSE and computation time ($d=2$, $m_0-d/2 \in \{1.5, 2.5\}$).}
				\label{tab:gp_full_d2_b}
				\begin{tabular}{@{}cc r c c c r r c@{}}
					\toprule
					& & & \multicolumn{3}{c}{Test MSE} & \multicolumn{3}{c}{Computation time} \\
					\cmidrule(lr){4-6}\cmidrule(lr){7-9}
					$m_0-d/2$ & $m-(d+1)/2$ & $n$ & \multicolumn{1}{c}{W} & \multicolumn{1}{c}{M} & \multicolumn{1}{c}{Ratio} & \multicolumn{1}{c}{$T_W$ (ms)} & \multicolumn{1}{c}{$T_M$ (ms)} & \multicolumn{1}{c}{Speedup} \\
					\midrule
					$1.5$ & $2$ & 100 & 0.01158\,(0.0057) & 0.01062\,(0.0056) & 1.09$\times$ & 0.147 & 0.343 & 2.33$\times$ \\
					&  & 200 & 0.008006\,(0.0043) & 0.006768\,(0.0038) & 1.18$\times$ & 0.290 & 0.952 & 3.28$\times$ \\
					&  & 400 & 0.005756\,(0.002) & 0.005006\,(0.0019) & 1.15$\times$ & 1.361 & 2.622 & 1.93$\times$ \\
					&  & 800 & 0.003836\,(0.0013) & 0.003133\,(0.0012) & 1.22$\times$ & 6.368 & 12.30 & 1.93$\times$ \\
					&  & 1600 & 0.002261\,(0.00053) & 0.001743\,(0.00037) & 1.30$\times$ & 22.66 & 57.79 & 2.55$\times$ \\
					\cmidrule(lr){3-9}
					& $3$ & 100 & 0.01198\,(0.006) & 0.01062\,(0.0056) & 1.13$\times$ & 0.143 & 0.343 & 2.40$\times$ \\
					&  & 200 & 0.008343\,(0.0045) & 0.006768\,(0.0038) & 1.23$\times$ & 0.305 & 0.952 & 3.12$\times$ \\
					&  & 400 & 0.006054\,(0.002) & 0.005006\,(0.0019) & 1.21$\times$ & 1.185 & 2.622 & 2.21$\times$ \\
					&  & 800 & 0.004079\,(0.0013) & 0.003133\,(0.0012) & 1.30$\times$ & 6.066 & 12.30 & 2.03$\times$ \\
					&  & 1600 & 0.00247\,(0.0006) & 0.001743\,(0.00037) & 1.42$\times$ & 21.80 & 57.79 & 2.65$\times$ \\
					\cmidrule(lr){3-9}
					$2.5$ & $3$ & 100 & 0.007521\,(0.005) & 0.006813\,(0.0045) & 1.10$\times$ & 0.139 & 0.319 & 2.30$\times$ \\
					&  & 200 & 0.004769\,(0.0033) & 0.004162\,(0.0031) & 1.15$\times$ & 0.250 & 0.812 & 3.25$\times$ \\
					&  & 400 & 0.002986\,(0.0011) & 0.002761\,(0.0013) & 1.08$\times$ & 0.900 & 2.530 & 2.81$\times$ \\
					&  & 800 & 0.001995\,(0.00086) & 0.001783\,(0.00085) & 1.12$\times$ & 5.035 & 10.48 & 2.08$\times$ \\
					&  & 1600 & 0.001028\,(0.00037) & 0.000817\,(0.00026) & 1.26$\times$ & 17.33 & 42.65 & 2.46$\times$ \\
					\bottomrule
					\multicolumn{9}{@{}p{0.97\linewidth}@{}}{\footnotesize W = Wendland (compactly supported); M = Mat\'{e}rn-kernel baseline. Test MSE values shown as mean\,(std) over 30 replications. $\text{Ratio}=\widehat{\text{MSE}}_W/\widehat{\text{MSE}}_M$;\ $\text{Speedup}=T_M/T_W$.}\\
				\end{tabular}
			\end{table}
			
			\begin{table}[htbp]
				\centering
				\renewcommand{\arraystretch}{1.10}
				\setlength{\tabcolsep}{2.2pt}
				\footnotesize
				\caption{GP regression: test MSE and computation time ($d=3$, $m_0-d/2=0.5$).}
				\label{tab:gp_full_d3_a}
				\begin{tabular}{@{}cc r c c c r r c@{}}
					\toprule
					& & & \multicolumn{3}{c}{Test MSE} & \multicolumn{3}{c}{Computation time} \\
					\cmidrule(lr){4-6}\cmidrule(lr){7-9}
					$m_0-d/2$ & $m-(d+1)/2$ & $n$ & \multicolumn{1}{c}{W} & \multicolumn{1}{c}{M} & \multicolumn{1}{c}{Ratio} & \multicolumn{1}{c}{$T_W$ (ms)} & \multicolumn{1}{c}{$T_M$ (ms)} & \multicolumn{1}{c}{Speedup} \\
					\midrule
					$0.5$ & $1$ & 200 & 0.1208\,(0.046) & 0.09017\,(0.031) & 1.34$\times$ & 0.417 & 0.856 & 2.05$\times$ \\
					&  & 400 & 0.11\,(0.016) & 0.0786\,(0.014) & 1.40$\times$ & 1.917 & 4.253 & 2.22$\times$ \\
					&  & 800 & 0.09383\,(0.013) & 0.05957\,(0.007) & 1.58$\times$ & 8.322 & 23.29 & 2.80$\times$ \\
					&  & 1500 & 0.08382\,(0.007) & 0.05099\,(0.004) & 1.64$\times$ & 27.90 & 92.77 & 3.33$\times$ \\
					\cmidrule(lr){3-9}
					& $2$ & 200 & 0.121\,(0.047) & 0.09017\,(0.031) & 1.34$\times$ & 0.420 & 0.856 & 2.04$\times$ \\
					&  & 400 & 0.1131\,(0.017) & 0.0786\,(0.014) & 1.44$\times$ & 1.740 & 4.253 & 2.44$\times$ \\
					&  & 800 & 0.09726\,(0.013) & 0.05957\,(0.007) & 1.63$\times$ & 8.245 & 23.29 & 2.82$\times$ \\
					&  & 1500 & 0.08878\,(0.008) & 0.05099\,(0.004) & 1.74$\times$ & 31.77 & 92.77 & 2.92$\times$ \\
					\cmidrule(lr){3-9}
					& $3$ & 200 & 0.116\,(0.047) & 0.09017\,(0.031) & 1.29$\times$ & 0.470 & 0.856 & 1.82$\times$ \\
					&  & 400 & 0.1088\,(0.017) & 0.0786\,(0.014) & 1.38$\times$ & 1.779 & 4.253 & 2.39$\times$ \\
					&  & 800 & 0.09235\,(0.012) & 0.05957\,(0.007) & 1.55$\times$ & 11.43 & 23.29 & 2.04$\times$ \\
					&  & 1500 & 0.08474\,(0.0072) & 0.05099\,(0.004) & 1.66$\times$ & 37.93 & 92.77 & 2.45$\times$ \\
					\bottomrule
					\multicolumn{9}{@{}p{0.97\linewidth}@{}}{\footnotesize W = Wendland (compactly supported); M = Mat\'{e}rn-kernel baseline. Test MSE values shown as mean\,(std) over 30 replications. $\text{Ratio}=\widehat{\text{MSE}}_W/\widehat{\text{MSE}}_M$;\ $\text{Speedup}=T_M/T_W$.}\\
				\end{tabular}
			\end{table}
			
			\begin{table}[htbp]
				\centering
				\renewcommand{\arraystretch}{1.10}
				\setlength{\tabcolsep}{2.2pt}
				\footnotesize
				\caption{GP regression: test MSE and computation time ($d=3$, $m_0-d/2 \in \{1.5, 2.5\}$).}
				\label{tab:gp_full_d3_b}
				\begin{tabular}{@{}cc r c c c r r c@{}}
					\toprule
					& & & \multicolumn{3}{c}{Test MSE} & \multicolumn{3}{c}{Computation time} \\
					\cmidrule(lr){4-6}\cmidrule(lr){7-9}
					$m_0-d/2$ & $m-(d+1)/2$ & $n$ & \multicolumn{1}{c}{W} & \multicolumn{1}{c}{M} & \multicolumn{1}{c}{Ratio} & \multicolumn{1}{c}{$T_W$ (ms)} & \multicolumn{1}{c}{$T_M$ (ms)} & \multicolumn{1}{c}{Speedup} \\
					\midrule
					$1.5$ & $2$ & 200 & 0.01703\,(0.0085) & 0.0125\,(0.0046) & 1.36$\times$ & 0.400 & 1.090 & 2.72$\times$ \\
					&  & 400 & 0.01215\,(0.0035) & 0.008998\,(0.0023) & 1.35$\times$ & 1.497 & 2.990 & 2.00$\times$ \\
					&  & 800 & 0.008957\,(0.0028) & 0.005821\,(0.0012) & 1.54$\times$ & 7.880 & 15.35 & 1.95$\times$ \\
					&  & 1500 & 0.006449\,(0.0013) & 0.003939\,(0.00065) & 1.64$\times$ & 22.60 & 62.92 & 2.78$\times$ \\
					\cmidrule(lr){3-9}
					& $3$ & 200 & 0.01472\,(0.0064) & 0.0125\,(0.0046) & 1.18$\times$ & 0.411 & 1.090 & 2.65$\times$ \\
					&  & 400 & 0.01116\,(0.0028) & 0.008998\,(0.0023) & 1.24$\times$ & 1.753 & 2.990 & 1.71$\times$ \\
					&  & 800 & 0.007857\,(0.0021) & 0.005821\,(0.0012) & 1.35$\times$ & 7.832 & 15.35 & 1.96$\times$ \\
					&  & 1500 & 0.005706\,(0.0011) & 0.003939\,(0.00065) & 1.45$\times$ & 25.09 & 62.92 & 2.51$\times$ \\
					\cmidrule(lr){3-9}
					$2.5$ & $3$ & 200 & 0.01287\,(0.0077) & 0.007732\,(0.0037) & 1.66$\times$ & 0.305 & 1.088 & 3.56$\times$ \\
					&  & 400 & 0.007203\,(0.0026) & 0.004728\,(0.0016) & 1.52$\times$ & 1.235 & 3.020 & 2.45$\times$ \\
					&  & 800 & 0.004946\,(0.0023) & 0.003009\,(0.0011) & 1.64$\times$ & 5.839 & 13.43 & 2.30$\times$ \\
					&  & 1500 & 0.002985\,(0.00086) & 0.00171\,(0.00046) & 1.75$\times$ & 18.03 & 48.15 & 2.67$\times$ \\
					\bottomrule
					\multicolumn{9}{@{}p{0.97\linewidth}@{}}{\footnotesize W = Wendland (compactly supported); M = Mat\'{e}rn-kernel baseline. Test MSE values shown as mean\,(std) over 30 replications. $\text{Ratio}=\widehat{\text{MSE}}_W/\widehat{\text{MSE}}_M$;\ $\text{Speedup}=T_M/T_W$.}\\
				\end{tabular}
			\end{table}
			
			\begin{table}[htbp]
				\centering
				\renewcommand{\arraystretch}{1.08}
				\small
				\caption{Empirical convergence slopes from the regenerated GP simulations, computed by regressing $\log(\mathrm{MSE})$ on $\log(n)$. W\,=\,Wendland, M\,=\,Mat\'{e}rn.}
				\label{tab:slopes}
				\begin{tabular}{@{}ccc r r r r@{}}
					\toprule
					$m_0-d/2$ & $m-(d+1)/2$ & $d$ & Theory & W slope & M slope & $\Delta$ \\
					\midrule
					$0.5$ & $1$ & $2$ & $-0.333$ & $-0.479$ & $-0.380$ & $-0.146$ \\
					$0.5$ & $1$ & $3$ & $-0.250$ & $-0.186$ & $-0.295$ & $+0.064$ \\
					$0.5$ & $2$ & $2$ & $-0.333$ & $-0.451$ & $-0.380$ & $-0.118$ \\
					$0.5$ & $2$ & $3$ & $-0.250$ & $-0.160$ & $-0.295$ & $+0.090$ \\
					$0.5$ & $3$ & $2$ & $-0.333$ & $-0.452$ & $-0.380$ & $-0.119$ \\
					$0.5$ & $3$ & $3$ & $-0.250$ & $-0.164$ & $-0.295$ & $+0.086$ \\
					\midrule
					$1.5$ & $2$ & $2$ & $-0.600$ & $-0.578$ & $-0.633$ & $+0.022$ \\
					$1.5$ & $2$ & $3$ & $-0.500$ & $-0.477$ & $-0.578$ & $+0.023$ \\
					$1.5$ & $3$ & $2$ & $-0.600$ & $-0.559$ & $-0.633$ & $+0.041$ \\
					$1.5$ & $3$ & $3$ & $-0.500$ & $-0.473$ & $-0.578$ & $+0.027$ \\
					\midrule
					$2.5$ & $3$ & $2$ & $-0.714$ & $-0.700$ & $-0.734$ & $+0.014$ \\
					$2.5$ & $3$ & $3$ & $-0.625$ & $-0.706$ & $-0.738$ & $-0.081$ \\
					\bottomrule
				\end{tabular}
			\end{table}
			
			\subsection{Additional Scalable-Baseline Experiments}
			
			Tables~\ref{tab:additional_gp_baselines_full} and \ref{tab:additional_det_baselines_full} give the full numerical summaries for the additional $d=2$ scalable-baseline comparison in Section~\ref{subsec:additional_gp_baselines}. The direct Wendland and tapered parent methods use the selected imposed smoothness \(m=7/2\). The tapered parent method uses the fixed support radius \(\phi=1\). For dense parent, Wendland, and tapered parent methods, time is the global linear-solver time; for local GP and Nystr\"om, time is their natural evaluation time.
			
			\begin{table}[htbp]
				\centering
				\renewcommand{\arraystretch}{1.04}
				\setlength{\tabcolsep}{2.4pt}
				\scriptsize
				\caption{Additional scalable-baseline comparison for the regular Mat\'{e}rn $d=2$ GP truth with \(m_0-d/2=0.5\). Values are means over five replications, with standard deviations in parentheses for test MSE.}
				\label{tab:additional_gp_baselines_full}
				\begin{tabular}{@{}r l c r r r c@{}}
					\toprule
					$n$ & Method & Test MSE & Time (ms) & Zeros (\%) & CG iter. & Structure \\
					\midrule
					200 & Dense parent & 0.0411\,(0.0129) & 3.84 & 0.0 & 28.4 & dense \\
					& Wendland ($m=7/2$) & 0.0541\,(0.00933) & 2.10 & 38.7 & 48.6 & $\phi=1.19$ \\
					& Tapered parent ($m=7/2$) & 0.0657\,(0.0123) & 1.61 & 51.9 & 49.6 & $\phi=1$ \\
					& Local GP ($k=0.1n$) & 0.0416\,(0.0134) & 2.45 & -- & -- & $k=20$ \\
					& Nystr\"om & 0.0566\,(0.0252) & 0.136 & -- & -- & landmarks=28 \\
					\cmidrule(lr){1-7}
					400 & Dense parent & 0.0443\,(0.0161) & 5.15 & 0.0 & 38.0 & dense \\
					& Wendland ($m=7/2$) & 0.0523\,(0.0255) & 3.71 & 48.9 & 65.8 & $\phi=1.05$ \\
					& Tapered parent ($m=7/2$) & 0.0549\,(0.0259) & 3.39 & 52.0 & 65.0 & $\phi=1$ \\
					& Local GP ($k=0.1n$) & 0.0453\,(0.0164) & 3.15 & -- & -- & $k=40$ \\
					& Nystr\"om & 0.0543\,(0.0173) & 0.266 & -- & -- & landmarks=56 \\
					\cmidrule(lr){1-7}
					800 & Dense parent & 0.0275\,(0.00348) & 17.9 & 0.0 & 50.0 & dense \\
					& Wendland ($m=7/2$) & 0.0327\,(0.00451) & 9.09 & 57.5 & 88.4 & $\phi=0.917$ \\
					& Tapered parent ($m=7/2$) & 0.0324\,(0.00473) & 10.6 & 51.7 & 88.4 & $\phi=1$ \\
					& Local GP ($k=0.1n$) & 0.0275\,(0.00339) & 11.5 & -- & -- & $k=80$ \\
					& Nystr\"om & 0.0318\,(0.00490) & 0.732 & -- & -- & landmarks=112 \\
					\cmidrule(lr){1-7}
					1600 & Dense parent & 0.0233\,(0.00144) & 128 & 0.0 & 65.0 & dense \\
					& Wendland ($m=7/2$) & 0.0254\,(0.00164) & 37.3 & 65.4 & 116.0 & $\phi=0.803$ \\
					& Tapered parent ($m=7/2$) & 0.0247\,(0.00170) & 55.4 & 51.7 & 115.8 & $\phi=1$ \\
					& Local GP ($k=0.1n$) & 0.0233\,(0.00143) & 98.8 & -- & -- & $k=160$ \\
					& Nystr\"om & 0.0254\,(0.00157) & 2.62 & -- & -- & landmarks=224 \\
					\bottomrule
					\multicolumn{7}{@{}p{0.96\linewidth}@{}}{\footnotesize Zeros denotes the fraction of exact zero entries in the global covariance matrix. CG iterations are reported only for methods solved through a global iterative linear system.}\\
				\end{tabular}
			\end{table}
			
			\begin{table}[htbp]
				\centering
				\renewcommand{\arraystretch}{1.04}
				\setlength{\tabcolsep}{2.4pt}
				\scriptsize
				\caption{Additional scalable-baseline comparison for the $d=2$ deterministic truth. Values are means over five replications, with standard deviations in parentheses for test MSE against the noiseless deterministic function.}
				\label{tab:additional_det_baselines_full}
				\begin{tabular}{@{}r l c r r r c@{}}
					\toprule
					$n$ & Method & Test MSE & Time (ms) & Zeros (\%) & CG iter. & Structure \\
					\midrule
					200 & Dense parent & 0.0577\,(0.0114) & 0.865 & 0.0 & 27.2 & dense \\
					& Wendland ($m=7/2$) & 0.0252\,(0.00775) & 0.456 & 76.1 & 44.2 & $\phi=0.638$ \\
					& Tapered parent ($m=7/2$) & 0.0105\,(0.00717) & 0.599 & 51.9 & 46.2 & $\phi=1$ \\
					& Local GP ($k=0.1n$) & 0.0589\,(0.00997) & 0.634 & -- & -- & $k=20$ \\
					& Nystr\"om & 0.1748\,(0.0475) & 0.0646 & -- & -- & landmarks=28 \\
					\cmidrule(lr){1-7}
					400 & Dense parent & 0.0258\,(0.00665) & 2.87 & 0.0 & 36.8 & dense \\
					& Wendland ($m=7/2$) & 0.0163\,(0.00968) & 0.823 & 81.0 & 55.4 & $\phi=0.559$ \\
					& Tapered parent ($m=7/2$) & 0.00524\,(0.00164) & 1.73 & 52.0 & 61.0 & $\phi=1$ \\
					& Local GP ($k=0.1n$) & 0.0299\,(0.00703) & 1.61 & -- & -- & $k=40$ \\
					& Nystr\"om & 0.0594\,(0.0154) & 0.128 & -- & -- & landmarks=56 \\
					\cmidrule(lr){1-7}
					800 & Dense parent & 0.00719\,(0.00120) & 13.8 & 0.0 & 45.4 & dense \\
					& Wendland ($m=7/2$) & 0.00493\,(0.000585) & 2.20 & 84.8 & 69.6 & $\phi=0.490$ \\
					& Tapered parent ($m=7/2$) & 0.00204\,(0.000271) & 8.05 & 51.7 & 80.4 & $\phi=1$ \\
					& Local GP ($k=0.1n$) & 0.00794\,(0.00123) & 6.48 & -- & -- & $k=80$ \\
					& Nystr\"om & 0.0175\,(0.00366) & 0.371 & -- & -- & landmarks=112 \\
					\cmidrule(lr){1-7}
					1600 & Dense parent & 0.00259\,(0.000179) & 78.7 & 0.0 & 62.4 & dense \\
					& Wendland ($m=7/2$) & 0.00314\,(0.000779) & 9.29 & 88.0 & 86.6 & $\phi=0.430$ \\
					& Tapered parent ($m=7/2$) & 0.00123\,(0.000101) & 51.6 & 51.7 & 106.6 & $\phi=1$ \\
					& Local GP ($k=0.1n$) & 0.00273\,(0.000189) & 41.6 & -- & -- & $k=160$ \\
					& Nystr\"om & 0.00348\,(0.000392) & 1.51 & -- & -- & landmarks=224 \\
					\bottomrule
					\multicolumn{7}{@{}p{0.96\linewidth}@{}}{\footnotesize Zeros denotes the fraction of exact zero entries in the global covariance matrix. CG iterations are reported only for methods solved through a global iterative linear system.}\\
				\end{tabular}
			\end{table}
			
			\subsection{Deterministic Function Experiments}
			
			Tables~\ref{tab:det_full_d2_a}--\ref{tab:det_full_d3_b} present the analogous deterministic-function kernel-ridge-regression results. They show that the Wendland estimator is competitive with, and in several settings improves on, the dense Mat\'{e}rn baseline in test MSE, while also providing systematic runtime reductions that reflect the solver-time value of sparse compact support.
			
			\begin{table}[htbp]
				\centering
				\renewcommand{\arraystretch}{1.10}
				\setlength{\tabcolsep}{2.2pt}
				\footnotesize
				\caption{KRR: test MSE and computation time ($d=2$, $m_0-d/2=0.5$).}
				\label{tab:det_full_d2_a}
				\begin{tabular}{@{}cc r c c c r r c@{}}
					\toprule
					& & & \multicolumn{3}{c}{Test MSE} & \multicolumn{3}{c}{Computation time} \\
					\cmidrule(lr){4-6}\cmidrule(lr){7-9}
					$m_0-d/2$ & $m-(d+1)/2$ & $n$ & \multicolumn{1}{c}{W} & \multicolumn{1}{c}{D} & \multicolumn{1}{c}{Ratio} & \multicolumn{1}{c}{$T_W$ (ms)} & \multicolumn{1}{c}{$T_D$ (ms)} & \multicolumn{1}{c}{Speedup} \\
					\midrule
					$0.5$ & $1$ & 100 & 0.15\,(0.096) & 0.1686\,(0.067) & 0.89$\times$ & 0.266 & 0.349 & 1.31$\times$ \\
					&  & 200 & 0.05443\,(0.035) & 0.06731\,(0.02) & 0.81$\times$ & 0.390 & 0.796 & 2.04$\times$ \\
					&  & 400 & 0.02512\,(0.017) & 0.02966\,(0.012) & 0.85$\times$ & 0.650 & 3.196 & 4.92$\times$ \\
					&  & 800 & 0.008003\,(0.0025) & 0.008751\,(0.0027) & 0.91$\times$ & 1.496 & 14.12 & 9.44$\times$ \\
					&  & 1600 & 0.004327\,(0.0012) & 0.002938\,(0.0005) & 1.47$\times$ & 5.655 & 73.04 & 12.92$\times$ \\
					\cmidrule(lr){3-9}
					& $2$ & 100 & 0.08311\,(0.066) & 0.1686\,(0.067) & 0.49$\times$ & 0.307 & 0.349 & 1.14$\times$ \\
					&  & 200 & 0.03236\,(0.022) & 0.06731\,(0.02) & 0.48$\times$ & 0.460 & 0.796 & 1.73$\times$ \\
					&  & 400 & 0.0178\,(0.014) & 0.02966\,(0.012) & 0.60$\times$ & 0.838 & 3.196 & 3.81$\times$ \\
					&  & 800 & 0.006515\,(0.002) & 0.008751\,(0.0027) & 0.74$\times$ & 2.250 & 14.12 & 6.28$\times$ \\
					&  & 1600 & 0.003726\,(0.00098) & 0.002938\,(0.0005) & 1.27$\times$ & 8.458 & 73.04 & 8.64$\times$ \\
					\cmidrule(lr){3-9}
					& $3$ & 100 & 0.07356\,(0.061) & 0.1686\,(0.067) & 0.44$\times$ & 0.321 & 0.349 & 1.09$\times$ \\
					&  & 200 & 0.03028\,(0.021) & 0.06731\,(0.02) & 0.45$\times$ & 0.474 & 0.796 & 1.68$\times$ \\
					&  & 400 & 0.01775\,(0.014) & 0.02966\,(0.012) & 0.60$\times$ & 0.986 & 3.196 & 3.24$\times$ \\
					&  & 800 & 0.006819\,(0.0021) & 0.008751\,(0.0027) & 0.78$\times$ & 2.563 & 14.12 & 5.51$\times$ \\
					&  & 1600 & 0.004015\,(0.0011) & 0.002938\,(0.0005) & 1.37$\times$ & 9.393 & 73.04 & 7.78$\times$ \\
					\bottomrule
					\multicolumn{9}{@{}p{0.97\linewidth}@{}}{\footnotesize W = Wendland KRR (compactly supported); D = Mat\'{e}rn KRR. Test MSE against the deterministic truth shown as mean\,(std) over 30 replications. $\text{Ratio}=\widehat{\text{MSE}}_W/\widehat{\text{MSE}}_D$;\ $\text{Speedup}=T_D/T_W$.}\\
				\end{tabular}
			\end{table}

			\begin{table}[htbp]
				\centering
				\renewcommand{\arraystretch}{1.10}
				\setlength{\tabcolsep}{2.2pt}
				\footnotesize
				\caption{KRR: test MSE and computation time ($d=2$, $m_0-d/2 \in \{1.5, 2.5\}$).}
				\label{tab:det_full_d2_b}
				\begin{tabular}{@{}cc r c c c r r c@{}}
					\toprule
					& & & \multicolumn{3}{c}{Test MSE} & \multicolumn{3}{c}{Computation time} \\
					\cmidrule(lr){4-6}\cmidrule(lr){7-9}
					$m_0-d/2$ & $m-(d+1)/2$ & $n$ & \multicolumn{1}{c}{W} & \multicolumn{1}{c}{D} & \multicolumn{1}{c}{Ratio} & \multicolumn{1}{c}{$T_W$ (ms)} & \multicolumn{1}{c}{$T_D$ (ms)} & \multicolumn{1}{c}{Speedup} \\
					\midrule
					$1.5$ & $2$ & 100 & 0.6432\,(0.21) & 0.4463\,(0.14) & 1.44$\times$ & 0.167 & 0.261 & 1.56$\times$ \\
					&  & 200 & 0.3604\,(0.12) & 0.3628\,(0.086) & 0.99$\times$ & 0.247 & 0.677 & 2.75$\times$ \\
					&  & 400 & 0.189\,(0.069) & 0.2599\,(0.042) & 0.73$\times$ & 0.389 & 1.983 & 5.10$\times$ \\
					&  & 800 & 0.04941\,(0.017) & 0.1568\,(0.017) & 0.32$\times$ & 0.737 & 9.054 & 12.28$\times$ \\
					&  & 1600 & 0.017\,(0.0066) & 0.08748\,(0.006) & 0.19$\times$ & 2.277 & 46.47 & 20.40$\times$ \\
					\cmidrule(lr){3-9}
					& $3$ & 100 & 0.6101\,(0.21) & 0.4463\,(0.14) & 1.37$\times$ & 0.172 & 0.261 & 1.52$\times$ \\
					&  & 200 & 0.3442\,(0.12) & 0.3628\,(0.086) & 0.95$\times$ & 0.249 & 0.677 & 2.72$\times$ \\
					&  & 400 & 0.1882\,(0.069) & 0.2599\,(0.042) & 0.72$\times$ & 0.400 & 1.983 & 4.96$\times$ \\
					&  & 800 & 0.05357\,(0.018) & 0.1568\,(0.017) & 0.34$\times$ & 0.765 & 9.054 & 11.83$\times$ \\
					&  & 1600 & 0.01994\,(0.0075) & 0.08748\,(0.006) & 0.23$\times$ & 2.470 & 46.47 & 18.81$\times$ \\
					\cmidrule(lr){3-9}
					$2.5$ & $3$ & 100 & 0.9981\,(0.25) & 0.5578\,(0.17) & 1.79$\times$ & 0.096 & 0.209 & 2.19$\times$ \\
					&  & 200 & 0.9303\,(0.2) & 0.5444\,(0.14) & 1.71$\times$ & 0.141 & 0.516 & 3.67$\times$ \\
					&  & 400 & 0.801\,(0.14) & 0.4816\,(0.077) & 1.66$\times$ & 0.215 & 1.761 & 8.17$\times$ \\
					&  & 800 & 0.5512\,(0.059) & 0.3895\,(0.041) & 1.42$\times$ & 0.337 & 7.399 & 21.97$\times$ \\
					&  & 1600 & 0.3311\,(0.037) & 0.319\,(0.022) & 1.04$\times$ & 0.851 & 33.41 & 39.26$\times$ \\
					\bottomrule
					\multicolumn{9}{@{}p{0.97\linewidth}@{}}{\footnotesize W = Wendland KRR (compactly supported); D = Mat\'{e}rn KRR. Test MSE against the deterministic truth shown as mean\,(std) over 30 replications. $\text{Ratio}=\widehat{\text{MSE}}_W/\widehat{\text{MSE}}_D$;\ $\text{Speedup}=T_D/T_W$.}\\
				\end{tabular}
			\end{table}

			\begin{table}[htbp]
				\centering
				\renewcommand{\arraystretch}{1.10}
				\setlength{\tabcolsep}{2.2pt}
				\footnotesize
				\caption{KRR: test MSE and computation time ($d=3$, $m_0-d/2=0.5$).}
				\label{tab:det_full_d3_a}
				\begin{tabular}{@{}cc r c c c r r c@{}}
					\toprule
					& & & \multicolumn{3}{c}{Test MSE} & \multicolumn{3}{c}{Computation time} \\
					\cmidrule(lr){4-6}\cmidrule(lr){7-9}
					$m_0-d/2$ & $m-(d+1)/2$ & $n$ & \multicolumn{1}{c}{W} & \multicolumn{1}{c}{D} & \multicolumn{1}{c}{Ratio} & \multicolumn{1}{c}{$T_W$ (ms)} & \multicolumn{1}{c}{$T_D$ (ms)} & \multicolumn{1}{c}{Speedup} \\
					\midrule
					$0.5$ & $1$ & 200 & 0.3219\,(0.13) & 0.2938\,(0.076) & 1.10$\times$ & 0.324 & 0.818 & 2.52$\times$ \\
					&  & 400 & 0.1312\,(0.048) & 0.1497\,(0.037) & 0.88$\times$ & 0.633 & 3.029 & 4.79$\times$ \\
					&  & 800 & 0.0483\,(0.016) & 0.06354\,(0.013) & 0.76$\times$ & 1.566 & 15.21 & 9.71$\times$ \\
					&  & 1500 & 0.02142\,(0.0071) & 0.02772\,(0.0041) & 0.77$\times$ & 5.702 & 66.93 & 11.74$\times$ \\
					\cmidrule(lr){3-9}
					& $2$ & 200 & 0.3471\,(0.13) & 0.2938\,(0.076) & 1.18$\times$ & 0.342 & 0.818 & 2.40$\times$ \\
					&  & 400 & 0.1604\,(0.055) & 0.1497\,(0.037) & 1.07$\times$ & 0.620 & 3.029 & 4.89$\times$ \\
					&  & 800 & 0.06611\,(0.021) & 0.06354\,(0.013) & 1.04$\times$ & 1.590 & 15.21 & 9.57$\times$ \\
					&  & 1500 & 0.03122\,(0.0092) & 0.02772\,(0.0041) & 1.13$\times$ & 5.315 & 66.93 & 12.59$\times$ \\
					\cmidrule(lr){3-9}
					& $3$ & 200 & 0.2966\,(0.12) & 0.2938\,(0.076) & 1.01$\times$ & 0.487 & 0.818 & 1.68$\times$ \\
					&  & 400 & 0.1335\,(0.049) & 0.1497\,(0.037) & 0.89$\times$ & 0.827 & 3.029 & 3.66$\times$ \\
					&  & 800 & 0.05724\,(0.019) & 0.06354\,(0.013) & 0.90$\times$ & 2.363 & 15.21 & 6.44$\times$ \\
					&  & 1500 & 0.0288\,(0.0087) & 0.02772\,(0.0041) & 1.04$\times$ & 7.544 & 66.93 & 8.87$\times$ \\
					\bottomrule
					\multicolumn{9}{@{}p{0.97\linewidth}@{}}{\footnotesize W = Wendland KRR (compactly supported); D = Mat\'{e}rn KRR. Test MSE against the deterministic truth shown as mean\,(std) over 30 replications. $\text{Ratio}=\widehat{\text{MSE}}_W/\widehat{\text{MSE}}_D$;\ $\text{Speedup}=T_D/T_W$.}\\
				\end{tabular}
			\end{table}

			\begin{table}[htbp]
				\centering
				\renewcommand{\arraystretch}{1.10}
				\setlength{\tabcolsep}{2.2pt}
				\footnotesize
				\caption{KRR: test MSE and computation time ($d=3$, $m_0-d/2 \in \{1.5, 2.5\}$).}
				\label{tab:det_full_d3_b}
				\begin{tabular}{@{}cc r c c c r r c@{}}
					\toprule
					& & & \multicolumn{3}{c}{Test MSE} & \multicolumn{3}{c}{Computation time} \\
					\cmidrule(lr){4-6}\cmidrule(lr){7-9}
					$m_0-d/2$ & $m-(d+1)/2$ & $n$ & \multicolumn{1}{c}{W} & \multicolumn{1}{c}{D} & \multicolumn{1}{c}{Ratio} & \multicolumn{1}{c}{$T_W$ (ms)} & \multicolumn{1}{c}{$T_D$ (ms)} & \multicolumn{1}{c}{Speedup} \\
					\midrule
					$1.5$ & $2$ & 200 & 0.8365\,(0.22) & 0.507\,(0.1) & 1.65$\times$ & 0.196 & 0.714 & 3.64$\times$ \\
					&  & 400 & 0.5737\,(0.14) & 0.3981\,(0.071) & 1.44$\times$ & 0.322 & 2.416 & 7.50$\times$ \\
					&  & 800 & 0.2714\,(0.066) & 0.2645\,(0.033) & 1.03$\times$ & 0.636 & 11.49 & 18.06$\times$ \\
					&  & 1500 & 0.1194\,(0.02) & 0.1889\,(0.014) & 0.63$\times$ & 1.857 & 49.21 & 26.50$\times$ \\
					\cmidrule(lr){3-9}
					& $3$ & 200 & 0.7607\,(0.21) & 0.507\,(0.1) & 1.50$\times$ & 0.224 & 0.714 & 3.18$\times$ \\
					&  & 400 & 0.5055\,(0.12) & 0.3981\,(0.071) & 1.27$\times$ & 0.373 & 2.416 & 6.47$\times$ \\
					&  & 800 & 0.239\,(0.06) & 0.2645\,(0.033) & 0.90$\times$ & 1.217 & 11.49 & 9.44$\times$ \\
					&  & 1500 & 0.1086\,(0.019) & 0.1889\,(0.014) & 0.57$\times$ & 2.692 & 49.21 & 18.28$\times$ \\
					\cmidrule(lr){3-9}
					$2.5$ & $3$ & 200 & 0.8559\,(0.22) & 0.6269\,(0.12) & 1.37$\times$ & 0.201 & 0.656 & 3.26$\times$ \\
					&  & 400 & 0.5836\,(0.14) & 0.5727\,(0.096) & 1.02$\times$ & 0.341 & 2.232 & 6.54$\times$ \\
					&  & 800 & 0.2682\,(0.066) & 0.4547\,(0.056) & 0.59$\times$ & 0.725 & 9.806 & 13.53$\times$ \\
					&  & 1500 & 0.1124\,(0.02) & 0.4005\,(0.027) & 0.28$\times$ & 2.414 & 37.29 & 15.45$\times$ \\
					\bottomrule
					\multicolumn{9}{@{}p{0.97\linewidth}@{}}{\footnotesize W = Wendland KRR (compactly supported); D = Mat\'{e}rn KRR. Test MSE against the deterministic truth shown as mean\,(std) over 30 replications. $\text{Ratio}=\widehat{\text{MSE}}_W/\widehat{\text{MSE}}_D$;\ $\text{Speedup}=T_D/T_W$.}\\
				\end{tabular}
			\end{table}

			\subsection{Real cases}
			
			\begin{table}[htbp]
				\centering
				\footnotesize
				\setlength{\tabcolsep}{2.5pt}
				\caption{Real-data convergence results for the ERA5 temperature experiment using each of 68 time snapshots separately ($N=70,600,320$, input dimension $d=2$). The Wendland column uses the constant scheduled scale $\phi_n=30$, Wendland smoothness $m=5/2$, and $\rho_n=10^{-4}$; MSE entries are averages across time points with standard deviations in parentheses when applicable. The empirical log-log MSE slope is $-0.95$ with $R^2=0.98$. Dense Mat\'{e}rn results are reported through $n=5,000$ only because the dense matrix storage and matvec cost become prohibitive for larger samples.}
				\label{tab:real_convergence}
				\begin{tabular}{@{}rrrrrrrr@{}}
					\toprule
					$n$ & $\phi_n$ & Test MSE & Sparsity & Time (s) & NNZ & Mat\'{e}rn MSE & Mat\'{e}rn time (s) \\
					\midrule
					200 & 30 & 0.4327\,(0.059) & 95.2\% & 0.002188 & 950 & 0.9979\,(0.15) & 0.0009216 \\
					500 & 30 & 0.1112\,(0.021) & 95.9\% & 0.003251 & 5014 & 0.5736\,(0.05) & 0.005272 \\
					1000 & 30 & 0.07253\,(0.014) & 96.0\% & 0.013 & 19672 & 0.419\,(0.051) & 0.03387 \\
					2000 & 30 & 0.03468\,(0.0058) & 96.1\% & 0.07016 & 76652 & 0.2065\,(0.03) & 0.1952 \\
					5000 & 30 & 0.01363\,(0.0014) & 96.1\% & 0.6978 & 481262 & 0.03983\,(0.004) & 2.69 \\
					10000 & 30 & 0.0106\,(0.00089) & 96.1\% & 4.492 & 1919538 & -- & -- \\
					\bottomrule
				\end{tabular}
			\end{table}
			
			\begin{table}[htbp]
				\centering
				\small
				\setlength{\tabcolsep}{4pt}
				\caption{Bandwidth sensitivity summary for the real-data experiment at $n=5,000$. The selected bandwidth is $\phi=50$. MSE values are averaged across time points.}
				\label{tab:real_bandwidth}
				\begin{tabular}{@{}rrrrr@{}}
					\toprule
					$\phi$ & Test MSE & Sparsity & Time (s) & NNZ \\
					\midrule
					5 & 0.398\,(0.037) & 99.9\% & 0.00634 & 18120 \\
					10 & 0.06382\,(0.0084) & 99.5\% & 0.02668 & 61038 \\
					20 & 0.01722\,(0.0018) & 98.2\% & 0.1669 & 224858 \\
					30 & 0.01578\,(0.0015) & 96.0\% & 0.6625 & 483938 \\
					40 & 0.01525\,(0.0015) & 93.3\% & 1.578 & 826062 \\
					\textbf{50} & \textbf{0.01488\,(0.0014)} & \textbf{89.9\%} & \textbf{3.143} & \textbf{1242110} \\
					\bottomrule
				\end{tabular}
			\end{table}
			
			\begin{table}[htbp]
				\centering
				\small
				\setlength{\tabcolsep}{4pt}
				\caption{Bandwidth scaling in the real-data experiment.}
				\label{tab:real_scaling}
				\begin{tabular}{@{}rrrrr@{}}
					\toprule
					$n$ & $\phi_{opt}$ & Test MSE & Sparsity & Time (s) \\
					\midrule
					200 & 50 & 0.1124\,(0.034) & 88.6\% & 0.001231 \\
					500 & 50 & 0.1076\,(0.029) & 89.8\% & 0.008018 \\
					1000 & 50 & 0.04041\,(0.0063) & 89.7\% & 0.03759 \\
					2000 & 50 & 0.02444\,(0.003) & 90.0\% & 0.2805 \\
					5000 & 50 & 0.01245\,(0.0011) & 89.9\% & 3.192 \\
					\bottomrule
				\end{tabular}
			\end{table}

	\end{document}